\documentclass[12pt]{article}
\usepackage{amsmath}
\usepackage{amssymb}
\usepackage{amsthm}
\usepackage{mathrsfs}
\usepackage{geometry}
\usepackage{subfigure}
\usepackage{diagbox}
\usepackage{booktabs}
\usepackage[colorlinks,linkcolor=blue,citecolor=blue]{hyperref}
\newtheorem{theorem}{Theorem}[section]

\newtheorem{proposition}{Proposition}[section]
\newtheorem{lemma}{Lemma}[section]
\newtheorem{corollary}{Corollary}[section]

\newtheorem{assumption}{Assumption}[section]
\theoremstyle{definition}

\newtheorem{remark}{Remark}[section]

\usepackage{graphicx}
\usepackage{epstopdf}
\usepackage{enumerate}
\usepackage{color}
\usepackage[all]{xy}
\newcommand{\mfh}{\mathfrak{h}}
\newcommand{\mfc}{\mathfrak{c}}
\newcommand{\mfa}{\mathfrak{a}}
\newcommand{\mfb}{\mathfrak{b}}
\newcommand{\mfg}{\mathfrak{g}}

\newcommand{\mfv}{\mathfrak{v}}
\newcommand{\mfr}{\mathfrak{r}}
\newcommand{\mfw}{\mathfrak{w}}
\newcommand{\mfu}{\mathfrak{u}}

\newcommand{\mfN}{\mathfrak{N}}

\newcommand{\bbA}{\boldsymbol{A}}

\newcommand{\bbC}{\boldsymbol{C}}
\newcommand{\bbD}{\boldsymbol{D}}

\newcommand{\bbG}{\boldsymbol{G}}
\newcommand{\bbH}{\boldsymbol{H}}
\newcommand{\bbI}{\boldsymbol{I}}
\newcommand{\bbJ}{\boldsymbol{J}}

\newcommand{\bbP}{\boldsymbol{P}}

\newcommand{\bbT}{\boldsymbol{T}}
\newcommand{\bbU}{\boldsymbol{U}}
\newcommand{\bbV}{\boldsymbol{V}}
\newcommand{\bbW}{\boldsymbol{W}}
\newcommand{\bbX}{\boldsymbol{X}}
\newcommand{\bbY}{\boldsymbol{Y}}
\newcommand{\bbZ}{\boldsymbol{Z}}

\newcommand{\bba}{\boldsymbol{a}}
\newcommand{\bbb}{\boldsymbol{b}}

\newcommand{\bbk}{\boldsymbol{k}}

\newcommand{\bbu}{\boldsymbol{u}}
\newcommand{\bbv}{\boldsymbol{v}}
\newcommand{\bbw}{\boldsymbol{w}}

\newcommand{\bby}{\boldsymbol{y}}

\newcommand{\bbve}{\boldsymbol{\varepsilon}}

\newcommand{\bbGa}{\boldsymbol{\Gamma}}

\newcommand{\bbXi}{\boldsymbol{\Xi}}

\newcommand{\bbSig}{\boldsymbol{\Sigma}}

\newcommand{\mbC}{\mathbb{C}}
\newcommand{\mcD}{\mathcal{D}}
\newcommand{\mbE}{\mathbb{E}}
\newcommand{\mbP}{\mathbb{P}}
\newcommand{\mbM}{\mathbb{M}}
\newcommand{\mbN}{\mathbb{N}}
\newcommand{\mbT}{\mathbb{T}}
\newcommand{\mbR}{\mathbb{R}}

\newcommand{\mbZ}{\mathbb{Z}}
\newcommand{\mbU}{\mathbb{U}}

\newcommand{\mbG}{\mathbb{G}}

\newcommand{\mrO}{\mathrm{O}}
\newcommand{\mro}{\mathrm{o}}

\newcommand{\mcG}{\mathcal{G}}

\newcommand{\mcN}{\mathcal{N}}

\newcommand{\mcE}{\mathcal{E}}

\newcommand{\mcT}{\mathcal{T}}

\newcommand{\mcJ}{\mathcal{J}}
\newcommand{\mcK}{\mathcal{K}}

\newcommand{\RomanNumeralCaps}[1]
{\MakeUppercase{\romannumeral #1}}

\newcommand{\mfC}{\mathfrak{C}}

\newcommand{\hla}{\hat{\lambda}}
\newcommand{\hR}{\hat{\boldsymbol{R}}}
\newcommand{\hF}{\hat{F}}

\newcommand{\fM}{\mathcal{M}}

\newcommand{\fz}{\mathfrak{z}}

\newcommand{\fd}{\mathfrak{d}}

\newcommand{\tsigma}{\beta}

\newcommand{\tr}{\operatorname{Tr}}
\newcommand{\diag}{\operatorname{diag}}
\newcommand{\Cov}{\operatorname{Cov}}
\newcommand{\Var}{\operatorname{Var}}

\newcommand{\Corr}{\operatorname{Corr}}

\allowdisplaybreaks

\title{Asymptotic spectral properties and extreme eigenvalues of sample correlation matrices from nonstationary high-dimensional VAR processes}
\author{Ruihan Liu\thanks{rhliu@connect.hku.hk}\quad and\quad Chen Wang\thanks{stacw@hku.hk}\\
\small{Department of Statistics and Actuarial Science, The University of Hong Kong}}
\date{}
\begin{document}
\maketitle

\begin{abstract}
    We study the limiting spectral distribution and extreme eigenvalues of the sample correlation matrix $\hR$ from a high-dimensional vector autoregressive (VAR) processes. Depending on whether roots of the characteristic polynomial of given VAR lie on (nonstationary) or outside (super-nonstationary) the unit circle, $\hR$ presents different types of asymptotic spectral properties. Under both nonstationary and super-nonstationary regimes, the empirical spectral distribution of $\hR$ degenerates in probability to a point mass at zero. The rescaled extreme eigenvalues of $\hR$ are asymptotically equivalent to the eigenvalues of a finite-dimensional symmetric random matrix determined by the characteristic roots. Moreover, under super-nonstationarity, $\hR$ degenerates in probability to a matrix with finite rank.
\end{abstract}

\section{Introduction}
The study of high-dimensional sample correlation matrix has attracted huge attention. As the entries of the sample correlation matrix are standardized, the statistics based on the sample correlation matrix share an advantage of scaling invariance. Plenty of literature have demonstrated the efficiency of sample-correlation based statistics in many statistical inference problems, such as testing the cross-sectional dependence of high-dimensional samples \cite{gao2017high,jiang2019determinant,zheng2019test,yin2021spectral,yin2023central,bao2024spectral} and estimating the number of factors in spiked matrix models \cite{fan2022estimating,yin2024inference}. Typically, all previous works leverage the spectral statistics of the sample correlation matrix, such as the linear spectral statistics (LSS) and the extreme eigenvalue. Hence, theoretical analysis of the asymptotic spectral properties of high-dimensional sample correlation matrices is important for statistical inferences.

From a practical perspective, the observed samples are often modeled by a VAR process in many empirical studies, while most of existing theoretical literature focus on the sample correlation matrix of independent samples. In this article, we analyze the asymptotic spectral properties of the sample correlation matrix generated by a high-dimensional VAR process. Given the observed data matrix $\bbX:=[X_1,\cdots,X_T]\in\mbR^{n\times T}$, we consider $X_t$ generated by an \(n\)-dimensional VAR($d$) ($d\in\mbN^+$ fixed) process as below:
\begin{align}
	X_t+\sum_{l=1}^da_lX_{t-l}=\bbGa\varepsilon_t,\label{Eq of AR process}
\end{align} 
where \(a_l\in\mbR\) for \(1\leq l\leq d\), $\bbGa\in\mbR^{n\times n}$ is the cross-sectional matrix and $\varepsilon_t=(\varepsilon_{1,t},\cdots,\varepsilon_{n,t})'$ are noise vectors such that $\varepsilon_{i,t}$ are independent with zero mean and unit variance. Moreover, the following high-dimensional regime is assumed.
\begin{assumption}\label{Ap of high-dimensionality}
	The data dimension $n$ and the sample size $T$ tend to infinity proportionally such that $\lim_{T\to\infty}n/T=c\in(0,\infty)$.
\end{assumption}

It is well-known that $X_t$ is (weakly) stationary\footnote{That is, $\mbE[X_tX_{t+k}']$ depends only on $k\in\mbZ$ instead of $t\in\mbZ$.} if and only if all roots of its characteristic polynomial are inside the unit circle on the complex plane, see \cite[\S3.2]{shumway2006time}. Define the characteristic polynomial of \eqref{Eq of AR process} as
\begin{align}
	f_X(x)=x^d+\sum_{l=1}^da_lx^{d-l}=\prod_{l=1}^d(x-\mfr_l),\label{Eq of characteristic polynomial}
\end{align}
where \(\mfr_l\in\mbC\) for \(1\leq l\leq d\). Based on whether \(\mfr_l\) is inside, on or outside the unit circle, we classify all \(\mfr_l\) into three types. Namely, we say \(\mfr_l\) is a
\begin{itemize}
	\item \emph{stationary root} if \(|\mfr_l|<1\);
	\item \emph{nonstationary root} if \(|\mfr_l|=1\);
	\item \emph{super nonstationary root} if \(|\mfr_l|>1\).
\end{itemize}
The main purpose of this article is to investigate how (super) nonstationary roots affect the asymptotic spectral properties of the sample correlation matrix  
\begin{align}
	\hR:=\diag(\bbX\bbX')^{-1/2}\bbX\bbX'\diag(\bbX\bbX')^{-1/2}\label{Eq of sample correlation matrix}
\end{align}
of $\bbX$. Let $\hla_1\geq\cdots\geq\hla_n$ be eigenvalues of $\hR$, the empirical spectral distribution (ESD) of $\hR$ is defined by $\nu_n(x):=\frac{1}{n}\sum_{i=1}^n1_{x=\hla_i}$, and the limit\footnote{This limit is with respect to the weak convergence of measures.} of $\nu_n(x)$ (if exists) is defined as the limit spectral distribution (LSD) of $\hR$. Precisely, we are interested in the LSD of $\hR$ and the asymptotic distribution of extreme eigenvalues $(\hla_1,\cdots,\hla_K)$ ($K\in\mbN^+$ is fixed) under the following two scenarios:
\begin{itemize}
	\item[{\bf N}] All roots of $f_X(\cdot)$ are either nonstationary or stationary, and at least one root is nonstationary;
	\item[{\bf SN}] At least one root of $f_X(\cdot)$ is super nonstationary.
\end{itemize}

When $f_X(\cdot)$ only contains stationary roots, there are many relating works analyzing the asymptotic spectral properties of $\hR$. Suppose the noises $\varepsilon_t$ in \eqref{Eq of AR process} are not heavy-tailed\footnote{When the noises $\varepsilon_t$ are heavy-tailed, the corresponding LSD is very different from the above situations, and readers can refer to \cite{heiny2022limiting} for details.}, it has been demonstrated in \cite{jiang2004limiting,bai2008large} that $n^{-1}\diag(\bbX\bbX')$ almost surely converges to a deterministic diagonal matrix $\bbD\in\mbR^{T\times T}$ in spectral norm. Hence, the ESD of $\hR$ asymptotically weakly converges to the LSD of $n^{-1}\bbD^{-1/2}\bbX\bbX'\bbD^{-1/2}$, where the later one is the sample covariance matrix of $\bbD^{-1/2}\bbX$ and its LSD has been established in \cite{mei2023singular}. Moreover, for the asymptotic distribution of $(\hla_1,\cdots,\hla_K)$, when $\bbX$ is temporal independent\footnote{That is, $X_{t_1}$ and $X_{t_2}$ are independent for any $t_1\neq t_2$.}, \cite{bao2012tracy,pillai2012edge} showed that the largest eigenvalue of $\hR$ follows the Tracy–Widom (TW) distribution of type \RomanNumeralCaps{1}.

However, when $X_t$ is nonstationary, the relating work is quite few, as in these two scenarios analyzing spectral properties of $\hR$ is technically challenging in the following aspects. First, the entries of $\hR$ are standardized as ratios of quadratic forms of random variables. Such nonlinear structure complicates the theoretical analysis, as one needs to deal with both the numerator and the denominator simultaneously. Moreover, the cross-sectional dependence and temporal dependence together make entries of $\bbX$ highly correlated, which significantly deviates from traditional random matrix models. Finally, when $X_t$ is nonstationary, the spectral norm of its population covariance matrix tends to infinity as $t\to\infty$. All situations mentioned above are significantly distinct from the standard framework in RMT, so the general RMT tools do not work for our situations.

Here, we briefly summarize the main results from the perspectives of ESD/LSD and asymptotic behaviors of extreme eigenvalues. 
\begin{itemize}
	\item \emph{Extreme eigenvalue: } \begin{itemize}
		\item {\bf N: }$n^{-1}\hla_1$ converges to a deterministic constant $c_1\in(0,1)$ in probability. The asymptotic distribution of $\hla_1$ is the same as the distribution of the largest eigenvalue of a $\mathtt{k_0}\times\mathtt{k_0}$ symmetric random matrix, where $\mathtt{k_0}\in\mbN^+$ is a finite integer depending only on roots of $f_X(\cdot)$. 
		\item {\bf SN: }The rank of $\hR$ is asymptotically no more than a finite $d_{\mathtt{SN}}\in\mbN^+$ depending only on roots of $f_X(\cdot)$, and $n^{-1}\hla_1$ converges to a deterministic constant $c_1\in(0,1]$ in probability. If $d_{\mathtt{SN}}=1$, $\hla_1\overset{a.s.}{\longrightarrow}n$ and $\hR$ is asymptotically of rank-one. When $d_{\mathtt{SN}}>1$, the asymptotic distribution of $\hla_1$ is the same as the distribution of the largest eigenvalue of a $\mathtt{k_0}\times\mathtt{k_0}$ symmetric random matrix.
	\end{itemize}
	\item \emph{ESD/LSD: } \begin{itemize}
		\item {\bf N and SN: }The ESD of $\hR$ degenerates to a point distribution $\nu_0$ with all mass at 0, i.e. $\int_{-\infty}^x\nu_0(dy)=1_{x\geq0}$.
	\end{itemize}
\end{itemize}
In summary, the existence of (super) nonstationary roots lead the LSD of $\hR$ degenerate, that is, nearly all eigenvalues of $\hR$ tend to 0 in probability, while the extreme eigenvalues will diverge to infinity in probability. Moreover, the existence of super nonstationary roots can further lead $\hR$ degenerate, whose rank is asymptotically no more than $d_{\mathtt{SN}}$. As for the asymptotic behaviors of extreme eigenvalues, it depends on the roots of $f_X(\cdot)$ under the scenario {\bf N} instead of the TW law, and it could degenerate to an almost surely convergence to the data dimension $n$ under the scenario {\bf SN}. Readers can refer to \S\ref{Sec of Nonstationary scenario} and \S\ref{Sec of Super nonstationary scenario} for more precise details.

The rest of this article is organized as follows. In \S\ref{Sec of Nonstationary scenario} and \S\ref{Sec of Super nonstationary scenario}, we analyze the asymptotic spectral properties of $\hR$ under the scenario {\bf N} and {\bf SN}, respectively.

Finally, we end this section with the following notations.
\begin{itemize}
	\item For two real sequence \(\{a_n\}\) and \(\{b_n\}\), $a_n\lesssim b_n$ means that $a_n\leq Mb_n$ for a positive constant $M$ as $n$ is sufficiently large. $a_n\asymp b_n$ means that $a_n\lesssim b_n$ and $b_n\lesssim a_n$.
	\item Given any integrable random variable/vector \(X\), \(X^{\circ}:=X-\mbE[X]\) denotes its centered version.
	\item The almost surely convergence, convergence in probability and in distribution are denoted by $\overset{a.s.}{\longrightarrow}$, \(\overset{\mbP}{\longrightarrow}\) and \(\overset{d}{\longrightarrow}\), respectively.
	\item Given a matrix \(\bbA=[A_{i,j}]_{m\times n}\), , \(\bbA'/\bbA^*\) denotes the transpose/complex conjugate of \(\bbA\); \(\Vert\bbA\Vert\) denotes the spectral norm of \(\bbA\) and $\|\bbA\|_k=(\sum_{i,j}|A_{i,j}|^k)^{1/k}$ is the $L^k$ norm for $k>0$.
	\item Given a square matrix, \(\tr(\bbA)=\sum_{i=1}^nA_{i,i}\) and \(\diag(\bbA)\) is the diagonal matrix made with the main diagonal of \(\bbA\).
	\item Given a matrix $\bbA$, let $\sigma_k(\bbA)$ be the $k$-th largest singular value of $\bbA$. Particularly, if all eigenvalues of $\bbA$ are real, $\lambda_k(\bbA)$ denotes the $k$-th largest eigenvalue of $\bbA$.
	\item For two sequences of random variables/vectors $\{X_n\in\mbR^m:n\in\mbN\}$ and $\{Y_n\in\mbR^m:n\in\mbN\}$, $X_n\overset{\mbP}{\sim}Y_n$ means that $\Vert X_n-Y_n\Vert_2\overset{\mbP}{\longrightarrow}0$; $X_n\overset{d}{\sim}Y_n$ means that $X_n$ and $Y_n$ have the same asymptotic distribution. That is, for any bounded and continuous function $f(\cdot)\in\mbR^m$,
	\begin{align}
		\lim_{n\to\infty}|\mbE[f(X_n)]-\mbE[f(Y_n)]|=0.\label{Eq of asymptotic convergence in distribution}
	\end{align}
\end{itemize}

\section{Nonstationary scenario}\label{Sec of Nonstationary scenario}
In this section, we assume that the characteristic polynomial $f_X(\cdot)$ in \eqref{Eq of characteristic polynomial} has either stationary roots or nonstationary roots, and at least one root is nonstationary. For distinction, we denote $\hR_{\mathtt{N}}$ and $\hR_{\mathtt{SN}}$ as the sample correlation matrix of $\bbX$ under the scenario {\bf N} and {\bf SN}, respectively. Here, we make the following notational conventions for facilitating the presentation of our main results:
\begin{itemize}
	\item We rewrite the characteristic polynomial $f_X(\cdot)$ in \eqref{Eq of characteristic polynomial} as below:
	\begin{align}
		f_X(x)=\prod_{l=1}^{\fd}(x-\mfr_l)^{d_l},\label{Eq of characteristic polynomial 1}
	\end{align}
	where $\fd$ is the number of differents roots of $f_X(\cdot)$ and the multiplicities of $\mfr_1,\cdots,\mfr_{\fd}$ are $d_1,\cdots,d_{\fd}$, respectively.
	\item Let $\mfr_1,\cdots,\mfr_m$ be all different nonstationary roots of $f_X(\cdot)$ such that $d_1\geq\cdots\geq d_m\geq1$, where $1\leq m\leq d$.
\end{itemize}
We specify assumptions for the noises $\varepsilon_{i,t}$ and cross-sectional matrix $\bbGa$ in \eqref{Eq of AR process} as follows.
\begin{assumption}\label{Ap of main}
	\begin{enumerate}
		\item $\varepsilon_{i,t}$ are independent for $1\leq i\leq n,1\leq t\leq T$ such that $\mbE[\varepsilon_{i,t}]=0,\mbE[\varepsilon_{i,t}^2]=1,\mbE[\varepsilon_{i,t}^3]=0,\mbE[\varepsilon_{i,t}^4]=3$ and $\kappa_{10}:=\sup_{1\leq i\leq n,1\leq t\leq T}\mbE[\varepsilon_{i,t}^{10}]<\infty$.
		\item The density functions of $\varepsilon_{i,t}$ are uniformly bounded, where $p_{i,t}(\cdot)$ is the density function of $\varepsilon_{i,t}$ such that $\sup_{1\leq i\leq n,1\leq t\leq T}\sup_{x\in\mbR}p_{i,t}(x)<\infty$.
	\end{enumerate}
\end{assumption}
\begin{assumption}\label{Ap of bbGa}
	The population covariance matrix $\bbSig=\bbGa\bbGa'=[\Sigma_{i,j}]_{n\times n}$ satisfies that $\max\{\Vert\bbSig\Vert,\Vert\bbSig^{-1}\Vert\}\lesssim1$.
\end{assumption}
\begin{remark}\label{Rem of assumptions}
	\begin{enumerate}
		\item For Assumption \ref{Ap of main}.1, we assume the first four moments of noises $\varepsilon_{j,t}$ are matching with those of standard Gaussian variables. This assumption is coincide with the well-known \emph{Four Moment Theorem}, which widely appears in many RMT literature (e.g. \cite{erdHos2010bulk,tao2011random,10.1214/11-AOP648,jiang2025asymptotic}). 
		\item Assumption \ref{Ap of main}.2 requires that all $\varepsilon_{j,t}$ have uniformly bounded density, which is used in kernel density estimator literature, see \cite{birge1993rates,gine2002rates}. Technically, one key method in proving Theorem \ref{Thm of nonstationary} is the Lindeberg's principle \cite{chatterjee2006generalization}, which requires Assumption \ref{Ap of main}.2. Readers can refer to \S\ref{sec of proof outline} for a brief technical discussion. 
	\end{enumerate}
\end{remark}
We first present the asymptotic behaviors of extreme eigenvalues of $\hR_{\mathtt{N}}$. For any $x\in\mbC$, let
\begin{align}
	\mcT(x)=\left(\begin{array}{ccccc}
		1&-x&0&\cdots&\\
		0&1&-x&0&\cdots\\
		&\ddots&\ddots&\ddots&\\
		&\cdots&0&1&-x\\
		&&\cdots&0&1
	\end{array}\right)\label{Eq of mcT}
\end{align}
be an upper toeplitz matrix such that its main diagonals are 1 and the first sup-diagonals are $-x$. Note that the eigenvalues of $\mcT(x)$ are all 1's for any $x\in\mbC$, so $\mcT(x)^{-1}$ exists. Moreover, by Vieta’s theorem, $\mcT(x)\mcT(y)=\mcT(y)\mcT(x)$ for any $x,y\in\mbC$. Since $X_t$ is generated as in \eqref{Eq of AR process}, assuming that $X_0=\cdots=X_{1-d}=\boldsymbol{0}$, we have
\begin{align}
	\bbX\mbT=\bbve,\quad\ \mbT:=\prod_{l=1}^d\mcT(\mfr_l)\implies\bbX=\bbGa\bbve\mbU,\quad\ \mbU:=\mbT^{-1}.\label{Eq of mbU}
\end{align}
Here, $\bbve=[\varepsilon_1,\cdots,\varepsilon_T]$ and $\mbT$ is is uniquely determined and independent of the order of $\mcT(\mfr_l)$ in multiplication.

According to the expression of $\bbX=\bbGa\bbve\mbU$, $\bbGa\bbve$ is the standard random matrix with cross-sectional dependence, to establish asymptotic spectral properties of $\bbX$, the key is to analyze the limiting spectrum of $\mbU\mbU'$. Thus let $\sigma_k(\mbU)$ is the $k$-th largest singular value of $\mbU$ and 
\begin{align}
	\beta_k=\sigma_k(\mbU)/\sigma_1(\mbU),\label{Eq of beta k}
\end{align}
we provide the following result:
\begin{proposition}[Spectral Gap of $\mbU$]\label{Lem of hat M}
	Let $m_0:=\#\{l=1,\cdots,m:d_l=d_1\}$ and
	\begin{align}
		\mathtt{k_0}:=\max\Big\{k\in\mbN^+:\limsup_{n,T\to\infty}\sqrt{n}(1-\beta_k)=0\Big\},\label{Eq of bbG n}
	\end{align}
	then
	\begin{enumerate}
		\item $\mathtt{k}_0=\#\{l=1,\cdots,m_0:|\mfa_l|=\min\{|\mfa_1|,\cdots,|\mfa_{m_0}|\}\}$, where
		\begin{align*}
			\mfa_l=\lim_{z\to\mfr_l}\frac{f_X(z)}{(z-\mfr_l)^{d_1}},\quad l=1,\cdots,m_0.
		\end{align*}
		\item $\liminf_{T\to\infty}(1-\beta_{\mathtt{k_0}+1})>0$.
	\end{enumerate}
\end{proposition}
Next, to characterize the asymptotic behavior of extreme eigenvalues of $\hR_{\mathtt{N}}$, we define the following auxiliary random variables. 
\begin{align}
	\left\{\hat{M}_{i;k_1,k_2}=\frac{\beta_{k_1}\beta_{k_2}\fz_{i,k_1}\fz_{i,k_2}}{\sum_{t=1}^T\beta_t^2\fz_{i,t}^2}:1\leq k_1,k_2\leq T\right\},\label{Eq of hat M}
\end{align}
where $\{\vec{\fz}_t=(\fz_{1,t},\cdots,\fz_{n,t})'\overset{i.i.d.}{\sim}\mcN(\boldsymbol{0},\bbSig):t=1,\cdots,T\}$. For any $1\leq k_1,k_2\leq T$, denote 
\begin{align}
	\fM_{k_1,k_2}:=\mbE[\hat{M}_{i;k_1,k_2}],\label{Eq of asymptotic mean}
\end{align}
which is independent of $i\in\{1,\cdots,n\}$ by \eqref{Eq of hat M}. Moreover, we construct a $T\times T$ symmetric random matrix as below:
\begin{align}
	\bbG_n=[G_{t_1,t_2;n}],\quad G_{t_1,t_2;n}=\frac{1}{\sqrt{n}}\sum_{i=1}^n\hat{M}_{i;t_1,t_2}^{\circ},\label{Eq of limit matrix}
\end{align}
and let $\bbG_n[T_1:T_2]$ denote the contiguous principal submatrix of $\bbG_n$ formed by the rows and columns indexed by $\{T_1,\cdots,T_2\}$, where $1\leq T_1\leq T_2\leq T$.
\begin{remark}\label{Rem of hat M}
	\begin{enumerate}
		\item Based on Proposition \ref{Lem of hat M} and the definition of $\fM_{k,k}$ in \eqref{Eq of asymptotic mean}, we have
		\begin{align*}
			|\fM_{1,1}-\fM_{k,k}|\left\{\begin{array}{ll}
				\leq\mro(n^{-1/2}), & 1\leq k\leq \mathtt{k}_0, \\
				\asymp 1 & k\geq \mathtt{k}_0+1.
			\end{array}\right.
		\end{align*}
		\item Since $\{\hat{M}_{i;k_1,k_2}:i=1,\cdots,n\}$ is a collection of correlated random variables, the CLT for $n^{-1/2}\sum_{i=1}^n\hat{M}_{i;k_1,k_2}^{\circ}$ may be false. On the other hand, if the population covariance matrix $\bbSig=[\Sigma_{i,j}]_{n\times n}$ is either sparse or toeplitz as below: 
		\begin{itemize}
			\item Sparse: $\Sigma_{i,j}=0$ if $|i-j|>n^{1/2-\delta}$ for some sufficiently small constant $\delta>0$,
			\item Toeplitz: $\Sigma_{i,j}=r^{|i-j|}$ for some constant $|r|<1$,
		\end{itemize}
		$\{\hat{M}_{i;k_1,k_2}:i=1,\cdots,n\}$ is strong $\alpha$-mixing or $m$-dependent (see \cite[\S7.7]{durrett2019probability}), which ensures the CLT for $n^{-1/2}\sum_{i=1}^n\hat{M}_{i;k_1,k_2}^{\circ}$. In this case, if $T_2-T_1$ is fixed, $\bbG_n[T_1:T_2]$ asymptotically converges to a symmetric Gaussian random matrix in distribution.
	\end{enumerate}
\end{remark}

\begin{theorem}[Extreme eigenvalues]\label{Thm of nonstationary}
	Under Assumptions \ref{Ap of high-dimensionality}, \ref{Ap of main} and \ref{Ap of bbGa}, the joint distribution of the first $\mathtt{k_0}$ largest eigenvalues of $\hR_{\mathtt{N}}$ satisfies that
	\begin{align}
		&\sqrt{n}\left(\frac{\hla_{1,\mathtt{N}}}{n}-\fM_{1,1},\cdots,\frac{\hla_{\mathtt{k_0},\mathtt{N}}}{n}-\fM_{1,1}\right)\overset{d}{\sim}\big(\lambda_1(\bbG_n[1:\mathtt{k_0}]),\cdots,\lambda_{\mathtt{k_0}}(\bbG_n[1:\mathtt{k_0}])\big),\label{Eq of nonstationary extreme eigenvalue asymptotic}
	\end{align}
	where $\lambda_k(\bbG_n[1:\mathtt{k_0}])$ is the $k$-th largest eigenvalue of $\bbG_n[1:\mathtt{k_0}]$ in \eqref{Eq of limit matrix} and the notation ``$\overset{d}{\sim}$'' is defined in \eqref{Eq of asymptotic convergence in distribution}.
\end{theorem}
\begin{remark}\label{Rem of eigenvalue N}
	\begin{enumerate}
		\item According to \eqref{Eq of asymptotic mean}, the asymptotic mean $\fM_{1,1}$ of $\frac{\hla_{1,\mathtt{N}}}{n}$ only depends on the singular values of $\mbU$, while its asymptotic variance depends both on the singular values of $\mbU$ and the cross-sectional dependence matrix $\bbGa$.
		\item When $\mathtt{k_0}=1$, $\frac{\hla_{1,\mathtt{N}}}{n}$ is asymptotically normal. When $\mathtt{k_0}=2$, by direct calculations, we have
		\begin{align*}
			\sqrt{n}\left(\frac{\hla_{1,\mathtt{N}}}{n}-\fM_{1,1}\right)\overset{d}{\sim}\frac{1}{2}(G_{1,1;n}+G_{2,2;n})+\frac{1}{2}\sqrt{(G_{1,1;n}-G_{2,2;n})^2+4(G_{1,2;n})^2},
		\end{align*}
		which is not asymptotically normal. In precise, the asymptotic distribution of $\Vert\hR_{\mathtt{N}}\Vert$ depends on roots of the characteristic polynomial $f_X(\cdot)$ in \eqref{Eq of characteristic polynomial}.
		\item For asymptotic behaviors of $\hla_{k,\mathtt{N}}$ with fixed $k>\mathtt{k_0}$, let's define
		\begin{align*}
			\mathtt{k_1}:=\max\Big\{k>\mathtt{k_0}:\limsup_{n,T\to\infty}\sqrt{n}(\beta_{\mathtt{k_0}+1}-\beta_k)=0\Big\},
		\end{align*}
		then
		\begin{align*}
			&\sqrt{n}\left(\frac{\hla_{\mathtt{k_0}+1,\mathtt{N}}}{n}-\fM_{\mathtt{k_0}+1,\mathtt{k_0}+1},\cdots,\frac{\hla_{\mathtt{k_1},\mathtt{N}}}{n}-\fM_{\mathtt{k_0}+1,\mathtt{k_0}+1}\right)\\
			&\overset{d}{\sim}\big(\lambda_1(\bbG_n[\mathtt{k_0}+1:\mathtt{k_1}]),\cdots,\lambda_{\mathtt{k_1}-\mathtt{k_0}}(\bbG_n[\mathtt{k_0}+1:\mathtt{k_1}])\big),
		\end{align*}
		where $\bbG_n[\mathtt{k_0}+1:\mathtt{k_1}]\in\mbR^{\mathtt{k_1}-\mathtt{k_0}\times\mathtt{k_1}-\mathtt{k_0}}$ is defined in \eqref{Eq of limit matrix}. And we can inductively define $\mathtt{k_t}$ based on $\mathtt{k_{t-1}}$ and derive the asymptotic distribution of $\hla_{k,\mathtt{N}}$ for $\mathtt{k_{t-1}}<k\leq\mathtt{k_t}$.
		\item The spectral gap of $\hR_{\mathtt{N}}$ also depends on the roots of $f_X(\cdot)$. When $\mathtt{k_0}=1$, note that $\fM_{1,1}>\fM_{2,2}$ and $n^{-1}\hla_{k,\mathtt{N}}\overset{\mbP}{\longrightarrow}\fM_{k,k}$ for $1\leq k\leq 2$, so $\hla_{1,\mathtt{N}}-\hla_{2,\mathtt{N}}\asymp n$ in probability. When $\mathtt{k_0}\geq2$, Theorem \ref{Thm of nonstationary} implies that 
		$$\frac{1}{\sqrt{n}}(\hla_{1,\mathtt{N}}-\hla_{2,\mathtt{N}})\overset{d}{\sim}\lambda_1(\bbG_n[1:\mathtt{k_0}])-\lambda_2(\bbG_n[1:\mathtt{k_0}]),$$
		then $\hla_{1,\mathtt{N}}-\hla_{2,\mathtt{N}}\asymp n^{1/2}$ in probability. 
	\end{enumerate}
\end{remark}
Finally, we provide the following result for the LSD of $\hR_{\mathtt{N}}$.
\begin{theorem}[LSD]\label{Pro of nonstationary LSD}
	Under Assumptions \ref{Ap of high-dimensionality}, \ref{Ap of main} and \ref{Ap of bbGa}, for any positive $\epsilon\asymp\log^{-1}(n)$,
	\begin{align*}
		\frac{1}{n}\#\{\hla_{k,\mathtt{N}}<\epsilon:k=1,\cdots,n\}\overset{\mbP}{\longrightarrow}1,
	\end{align*}
	so the LSD of $\hR_{\mathtt{N}}$ is a point distribution $\nu_0$ with all mass at 0, i.e.
	\begin{align*}
		\int_{-\infty}^x\nu_0(dy)=1\{x\geq0\}.
	\end{align*}
\end{theorem}
\begin{remark}
	Under the scenario when $X_t$ is stationary, the LSD of $\hR$ follows a generalized Marčenko–Pastur (MP) distribution (see \cite{mei2023singular}), which is generally non-degenerate. By contrast, the existence of nonstationary roots will lead the LSD of $\hR_{\mathtt{N}}$ degenerate a point distribution with all mass at 0. Consequently, under the scenario {\bf N}, the LSS of $\hR_{\mathtt{N}}$ is degenerate due to nearly all eigenvalues of $\hR_{\mathtt{N}}$ are zero,  so the LSS of $\hR_{\mathtt{N}}$  is no longer a meaningful spectral statistic.
\end{remark}

\subsection{Proof of Proposition \texorpdfstring{\ref{Lem of hat M}}{2.4}}\label{sec of spectral gap}
Before presenting the proof of Proposition \ref{Lem of hat M}, we need the following results:
\begin{lemma}[Gantmacher-Krein]\label{Lem of operator}
	For the operator 
	\begin{align}
		\bbJ(f)(s):=\int_s^1f(z)dz,\quad f\in L^2([0,1]).\label{Eq of Volterra integration operator}
	\end{align}
	and a fixed $d\in\mbN^+$, we have:
	\begin{enumerate}
		\item 
		$$\|T^{-2d}(\mcT(1)^{-d})'\mcT(1)^{-d}-E_T^*(\bbJ^d)^*\bbJ^dE_T\|\lesssim T^{-1},$$
		where $\mcT(1)\in\mbR^{T\times T}$ is defined in \eqref{Eq of mcT} and $E_T:\mbC^T\longrightarrow L^2([0,1])$ is an embedding operator such that for any $\bba=(a_1,\cdots,a_T)'\in\mbC^n$,  
		$$E_T(\bba)(s)=\sqrt{T}a_t,\quad s\in[(t-1)/T,t/T).$$
		$E_T^*:L^2([0,1])\longrightarrow\mbC^T$ is the adjoint operator of $E_T$ such that for any $f\in L^2([0,1])$, the $t$-th entry of $E_T^*(f)\in\mbC^n$ is
		\begin{align}
			E_T^*(f)_t=\sqrt{T}\int_{(t-1)/T}^{t/T}f(z)dz.\label{Eq of embedding operator}
		\end{align}
		\item The eigenvalues of $(\bbJ^d)^*\bbJ^d$ are positive and strictly simple, i.e.
		\begin{align*}
			\lambda_1((\bbJ^d)^*\bbJ^d)>\lambda_2((\bbJ^d)^*\bbJ^d)>\cdots>0.
		\end{align*}
	\end{enumerate}
\end{lemma}
\begin{lemma}\label{Lem of singular structure}
	For $\beta_t$ in \eqref{Eq of beta k}, let \(d_1=\max_{1\leq l\leq m}d_l\) be the multiplicity of $\mfr_1$,
	\begin{enumerate}
		\item $\|\mbU\|\asymp T^{d_1}$;
		\item \(\tsigma_t\asymp t^{-d_1}\) for \(1\leq t\leq T\), which gives that \(\sum_{t=1}^T\tsigma_t^2\asymp 1\).
	\end{enumerate}
\end{lemma}
The proof of Lemmas \ref{Lem of operator} and \ref{Lem of singular structure} are postponed in the Appendix \ref{App of spectral gap}.
\begin{proof}[Proof of Proposition \ref{Lem of hat M}]
	By the characteristic polynomial $f_X(x)$ defined in \eqref{Eq of characteristic polynomial}, we can represent the matrix $\mbU$ in \eqref{Eq of mbU} as follows:
	$$\mbU=\sum_{l=1}^{\fd}\sum_{k_l=1}^{d_l}a_{k_l,l}\mcT(\mfr_l)^{-k_l},\quad a_{k_l,l}\in\mbC$$
	for $k_l=1,\cdots, d_l$ and $l=1,\cdots,\fd$, where we use the partial fraction decomposition of $\prod_{l=1}^{\fd}(1-\mfr_l x)^{d_l}$ as below:
	$$\frac{1}{\prod_{l=1}^{\fd}(1-\mfr_l x)^{d_l}}=\sum_{l=1}^{\fd}\sum_{k_l=1}^{d_l}a_{k_l,l}(1-\mfr_lx)^{-k_l}.$$
	Particularly, 
	\begin{align}
		a_{d_l,l}=\frac{1}{\prod_{r\neq l}^{\fd}(1-\mfr_r\mfr_l^{-1})^{d_r}},\quad l=1,\cdots,\fd.\label{Eq of a l}
	\end{align}
	Recall that $|\mfr_1|=\cdots=|\mfr_m|=1>|\mfr_{m+1}|\geq\cdots$ and the multiplicity $d_l$ of $\mfr_l$ satisfies that $d_1\geq d_2\geq\cdots$ and $d_1=\cdots=d_{m_0}$ for $1\leq m_0\leq m$. By the proof of Lemma \ref{Lem of singular structure}, we have
	\begin{align*}
		\|\mcT(\mfr_l)^{-k_l}\|\asymp\left\{\begin{array}{cc}
			T^{k_l}, & |\mfr_l|=1, \\
			1, & |\mfr_l|<1,
		\end{array}\right.
	\end{align*}
	then we rewrite $\mbU$ as follows:
	\begin{align*}
		\mbU=\underbrace{\sum_{l=1}^{m_0}a_{d_l,l}\mcT(\mfr_l)^{-d_l}}_{\widetilde{\mbU}:=}+\Delta\mbU,
	\end{align*}
	where $\|\Delta\mbU\|\lesssim T^{d_1-1}$. By the Weyl's theorem, we have
	\begin{align*}
		|\sigma_k(\mbU)-\sigma_k(\widetilde{\mbU})|\leq\|\Delta\mbU\|\lesssim T^{d_1-1}\implies|\beta_k-\sigma_k(\widetilde{\mbU})/\sigma_1(\mbU)|\lesssim T^{-1}.
	\end{align*}
	Next, let's anaylze the singular values of $\widetilde{\mbU}$, note that
	\begin{align*}
		\widetilde{\mbU}\widetilde{\mbU}^*=\sum_{l_1,l_2=1}^{m_0}a_{d_1,l_1}\overline{a_{d_1,l_2}}\mcT(\mfr_{l_1})^{-d_1}(\mcT(\mfr_{l_2})^{-d_1})^*,
	\end{align*}
	we claim that $\|\mcT(\mfr_{l_1})^{-d_1}(\mcT(\mfr_{l_2})^{-d_1})^*\|\lesssim T^{2d_1-1}$ for $l_1\neq l_2$. Since $\mcT(\mfr_{l_1})^{-1}=\bbD_l^*\mcT(1)^{-1}\bbD_l$, where $\bbD_l=\diag(1,\mfr_l,\cdots,\mfr_l^{T-1})$ and $\bbD_l^*=\diag(1,\mfr_l^{-1},\cdots,\mfr_l^{1-T})$, it implies that
	\begin{align*}
		&\mcT(\mfr_{l_1})^{-d_1}(\mcT(\mfr_{l_2})^{-d_1})^*=\bbD_{l_1}^*\mcT(1)^{-d_1}\bbD_{l_1}\bbD_{l_2}^*(\mcT(1)^{-d_1})'\bbD_{l_2}\\
		&=\bbD_{l_1}^*\mcT(1)^{-(d_1-1)}(\mcT(1)^{-1}\bbD_{l_1}\bbD_{l_2}^*(\mcT(1)^{-1})')(\mcT(1)^{-(d_1-1)})'\bbD_{l_2}^*.
	\end{align*}
	Denote $\bbW^{(l_1,l_2)}:=\mcT(1)^{-1}\bbD_{l_1}\bbD_{l_2}^*(\mcT(1)^{-1})'$, we have
	\begin{align*}
		W_{t_1,t_2}^{(l_1,l_2)}=\sum_{s=1}^{\min\{t_1,t_2\}}(\mfr_{l_1}/\mfr_{l_2})^{s-1}.
	\end{align*}
	For $l_1\neq l_2$, then $\mfr_{l_1}/\mfr_{l_2}\neq1$ and $|W_{t_1,t_2}|\leq\frac{2}{|1-\mfr_{l_1}/\mfr_{l_2}|}$, which implies that 
	$$\|\bbW^{(l_1,l_2)}\|\leq\|\bbW^{(l_1,l_2)}\|_F\lesssim T.$$
	Thus, for $l_1\neq l_2$,
	\begin{align}
		\|\mcT(\mfr_{l_1})^{-d_1}(\mcT(\mfr_{l_2})^{-d_1})^*\|\leq\|\mcT(\mfr_{l_1})^{-(d_1-1)}\|\cdot\|\mcT(\mfr_{l_2})^{-(d_1-1)}\|\cdot\|\bbW^{(l_1,l_2)}\|\lesssim T^{2d_1-1},\label{Eq of off-diagonal block}
	\end{align}
	and it further yields that
	$$|\beta_k-\lambda_k(\mbM)^{1/2}/\lambda_1(\mbM)^{1/2}|\lesssim T^{-1},$$
	where 
	\begin{align*}
		\mbM:=\frac{1}{T^{2d_1}}\sum_{l=1}^{m_0}|a_{d_1,l}|^2\mcT(\mfr_l)^{-d_1}(\mcT(\mfr_l)^{-d_1})^*=\frac{1}{T^{2d_1}}\sum_{l=1}^{m_0}|a_{d_1,l}|^2\bbD_l^*\mcT(1)^{-d_1}(\mcT(1)^{-d_1})'\bbD_l.
	\end{align*}
	Without loss of generality, we assume that $|a_{d_1,1}|^2\geq\cdots\geq|a_{d_1,m_0}|^2>0$ and let
	\begin{align}
		\bbk:=\#\{1\leq l\leq m_0:|a_{d_1,l}|^2=|a_{d_1,1}|^2\}.\label{Eq of bbk}
	\end{align}
	By direct calculations, we have
	\begin{align*}
		\mbM:=\mbG\mbG^*,\quad\mbG:=T^{-d_1}\big[|a_{d_1,1}|\bbD_1^*\mcT(1)^{-d_1},\cdots,|a_{d_1,m_0}|\bbD_{m_0}^*\mcT(1)^{-d_1}\big].
	\end{align*}
	Since we only focus on extreme eigenvalues of $\mbM$, which shares the same nonzero eigenvalues with $\mbG^*\mbG$, and
	\begin{align*}
		\mbG^*\mbG=\left(\begin{array}{ccc}
			\bbXi_{1,1} & \cdots & \bbXi_{1,m_0} \\
			\vdots & \ddots & \vdots \\
			\bbXi_{m_0,1} & \cdots & \bbXi_{m_0,m_0}
		\end{array}\right),\quad \bbXi_{l_1,l_2}:=\frac{|a_{d_1,l_1}a_{d_1,l_2}|}{T^{2d_1}}(\mcT(1)^{-d_1})'\bbD_{l_1}\bbD_{l_2}^*\mcT(1)^{-d_1}.
	\end{align*}
	For the off-diagonal block ($l_1\neq l_2$), by \eqref{Eq of off-diagonal block}, we have 
	\begin{align*}
		\|\bbXi_{l_1,l_2}\|=\frac{|a_{d_1,l_1}a_{d_1,l_2}|}{T^{2d_1}}\|(\mcT(\mfr_{l_1})^{-d_1})^*\mcT(\mfr_{l_2})^{-d_1}\|\lesssim T^{-1}.
	\end{align*}
	For the diagonal block, we have
	\begin{align*}
		\bbXi_{l,l}=\frac{|a_{d_1,l}|^2}{T^{2d_1}}(\mcT(1)^{-d_1})'\mcT(1)^{-d_1}.
	\end{align*}
	By Lemma \ref{Lem of operator}, we have
	\begin{align*}
		\|T^{-2d_1}(\mcT(1)^{-d_1})'\mcT(1)^{-d_1}-E_T^*(\bbJ^{d_1})^*\bbJ^{d_1}E_T\|\lesssim\frac{1}{T},
	\end{align*}
	and all eigenvalues of $(\bbJ^d)^*\bbJ^d$ are positive and strictly simple (multiplicity 1):
	\begin{align*}
		\lambda_1((\bbJ^d)^*\bbJ^d)>\lambda_2((\bbJ^d)^*\bbJ^d)')>\cdots>0.
	\end{align*}
	Finally, we construct
	$$\bbXi_{\diag}:=\left(\begin{array}{ccc}
		\bbXi_{1,1} & \cdots & \boldsymbol{0} \\
		\vdots & \ddots & \vdots \\
		\boldsymbol{0} & \cdots & \bbXi_{m_0,m_0}
	\end{array}\right)$$
	by letting all off-diagonal blocks of $\mbG^*\mbG$ be $\boldsymbol{0}$, combining with $\|\bbXi_{l_1,l_2}\|\lesssim T^{-1}$ for $l_1\neq l_2$, the Weyl's theorem gives that
	\begin{align*}
		|\lambda_l(\mbM)-\lambda_l(\bbXi_{\diag})|=|\lambda_l(\mbG^*\mbG)-\lambda_l(\bbXi_{\diag})|\lesssim T^{-1},\quad l=1,\cdots,\bbk+1,
	\end{align*}
	which further yields that for $l=1,\cdots,\bbk+1$
	\begin{align*}
		&|\beta_l-\lambda_l(\bbXi_{\diag})^{1/2}/\lambda_1(\bbXi_{\diag})^{1/2}|\\
		&\leq|\beta_l-\lambda_l(\mbM)^{1/2}/\lambda_1(\mbM)^{1/2}|+|\lambda_l(\mbM)^{1/2}/\lambda_1(\mbM)^{1/2}-\lambda_l(\bbXi_{\diag})^{1/2}/\lambda_1(\bbXi_{\diag})^{1/2}|\\
		&\lesssim T^{-1}+|\lambda_l(\mbM)/\lambda_1(\mbM)-\lambda_l(\bbXi_{\diag})/\lambda_1(\bbXi_{\diag})|\lesssim T^{-1}.
	\end{align*}
	Recall the definition of $\bbk$ in \eqref{Eq of bbk}, we have
	\begin{align*}
		\lambda_1(\bbXi_{\diag})=\cdots=\lambda_{\bbk}(\bbXi_{\diag})>\lambda_{\bbk+1}(\bbXi_{\diag}),
	\end{align*}
	so by Assumption \ref{Ap of high-dimensionality},
	\begin{align*}
		\sqrt{n}(1-\beta_l)\left\{\begin{array}{ll}
			\lesssim T^{-1/2} & l=1,\cdots,\bbk, \\
			\gtrsim T^{1/2} & l=\bbk+1.
		\end{array}\right.
	\end{align*}
	Compared the definition of $\mfa_l$ in Proposition \ref{Lem of hat M} with $a_{d_l,l}$ in \eqref{Eq of a l}, we have
	\begin{align*}
		\mfa_l=\prod_{r\neq l}^{\fd}(\mfr_l-\mfr_r)^{d_r}=\mfr_l^d\prod_{r\neq l}^{\fd}(1-\mfr_r\mfr_l^{-1})^{d_r}=\frac{\mfr_l^d}{a_{d_l,l}},
	\end{align*}
	i.e. $|\mfa_l|=|a_{d_l,l}|^{-1}$ for $l=1,\cdots,m_0$, then \eqref{Eq of bbk} implies that $\mathtt{k}_0=\bbk$, which completes the proof.
\end{proof}

\subsection{Outline for proving Theorem \texorpdfstring{\ref{Thm of nonstationary}}{2.6}}\label{sec of proof outline}
As we focus on extreme eigenvalues of $\hR_{\mathtt{N}}$, note that $(\bbX\bbX')^{-1/2}\bbX\bbX'\diag(\bbX\bbX')^{-1/2}$ and $\bbX\diag(\bbX\bbX')^{-1}\bbX'$ share the same nonzero eigenvalues, without loss of generality, we assume
$$\hR_{\mathtt{N}}=\bbX'\diag(\bbX\bbX')^{-1}\bbX.$$
The proof proceeds in the following five steps, relying heavily on the auxiliary lemmas established in \S\ref{ssec of Basic properties for noises} to \S\ref{sec of Covariance analysis}.

\vspace{3mm}\noindent
{\bf Step 1: Bilinear Form Representation and SVD Projection}\\
First, we denote the singular value decomposition (SVD) of $\mbU$ in \eqref{Eq of mbU} by
\begin{align}
	\mbU=\sum_{t=1}^T\sigma_t\mfv_t\mfw_t',\quad\sigma_1\geq\cdots\geq\sigma_T\geq0.\label{Eq of SVD mbU}
\end{align}
Let $\hF_{k,\mathtt{N}}:=\sum_{t=1}^T\alpha_{k,t}\mfw_t$ be the eigenvector of $\hR_{\mathtt{N}}$ corresponding to the $k$-th largest eigenvalue $\hla_{k,\mathtt{N}}$, where $\sum_{t=1}^T\alpha_{k,t}^2=1$, then we express $\hla_{k,\mathtt{N}}$ as a combination of self-normalized bilinear forms:
\begin{align}
	\hla_{k,\mathtt{N}}=\hF_{k,\mathtt{N}}'\hR_{\mathtt{N}}\hF_{k,\mathtt{N}}=\sum_{s,t=1}^T\alpha_{k,s}\alpha_{k,t}\sum_{i=1}^nM_{i;s,t},\quad M_{i;s,t}:=\frac{\tsigma_s\tsigma_t(\Gamma_i\bbve\mfv_s)(\Gamma_i\bbve\mfv_t)}{\sum_{r=1}^T\tsigma_r^2(\Gamma_i\bbve\mfv_r)^2},\label{Eq of Mi st}
\end{align}
where $\Gamma_i$ is the $i$-th row of $\bbGa$ and $\beta_t$ is defined in \eqref{Eq of beta k}. The key of proving Theorem \ref{Thm of nonstationary} is to show that
\begin{align}
	\left\{\begin{array}{l}
		\frac{1}{\sqrt{n}}\sum_{i=1}^nM_{i;k,l}^{\circ}\overset{d}{\sim}\frac{1}{\sqrt{n}}\sum_{i=1}^n\hat{M}_{i;k,l}^{\circ}, \\
		\frac{1}{\sqrt{n}}\big|\sum_{i=1}^n\mbE[M_{i;k,l}-\hat{M}_{i;k,l}]\big|\longrightarrow0, \\
		\frac{1}{n}\big|\Var\big(\sum_{i=1}^nM_{i;k,l}\big)-\Var\big(\sum_{i=1}^n\hat{M}_{i;k,l}\big)\big|\longrightarrow0,
	\end{array}\right.\label{Eq of Lindeberg principle 1}
\end{align}
where $\hat{M}_{i;k,l}$ is defined in \eqref{Eq of hat M}. 

\vspace{3mm}\noindent
{\bf Step 2: Multidimensional Lindeberg Replacement}\\
If we replace $\bbve$ by a standard Gaussian random matrix $\bbZ_{n,T}\in\mbR^{n\times T}$ with independent entries and denote
\begin{align*}
	M_{i;k,l}(\bbZ_{n,T})=\frac{\beta_k\beta_l(\Gamma_i\bbZ_{n,T}\mfv_k)(\Gamma_i\bbZ_{n,T}\mfv_l)}{\sum_{t=1}^T\beta_t^2(\Gamma_i\bbZ_{n,T}\mfv_t)^2},
\end{align*}
then $M_{i;k,l}(\bbZ_{n,T})$ and $\hat{M}_{i;k,l}$ have the same distribution. Hence, to conclude \eqref{Eq of Lindeberg principle 1}, we apply the Lindeberg’s principle \cite{chatterjee2006generalization} by replacing $\varepsilon_{i,t}$ by an independent normal variables $z_{i,t}\sim\mcN(0,1)$ one by one.

Intuitively, the Lindeberg replacement works since the singular vectors $\mfv_k$ of $\mbU$ are delocalized, i.e. $\|\mfv_k\|_{\infty}\lesssim T^{-1/2}$, so $\Gamma_i\bbve\mfv_k$ is asymptotic normal (CLT) as $n,T\to\infty$. However, there is no general explicit form for entries of $\mfv_k$. To solve this difficulty, we expand the singular vectors of $\mbU$ in a discrete Fourier transform (DFT) basis (Auxiliary Lemmas in \S\ref{ssec of discrete Fourier transform}), which guarantees the necessary $O(T^{-1/2})$ delocalization decay rates.

To control the error causing by the Lindeberg replacement, we use the 4th order Taylor expansions with integral remainders, the \emph{Four Moment Coincidence} (Assumption \ref{Ap of main}) ensures that the first four major terms in Taylor expansions are cancelled out. Remainders are vanishing based on $O(T^{-1/2})$ delocalization decay rates of $\mfv_k$.

\vspace{3mm}\noindent
{\bf Step 3: Cross-Sectional Variance Control}\\
Next, following the Gaussian substitution, we will prove that
\begin{align*}
	\frac{1}{n}\Var\left(\sum_{i=1}^n\hat{M}_{i;k,l}\right)\lesssim(kl)^{-2d_1},
\end{align*}
it suffices to analyze $\Cov(\hat{M}_{i_1;k,l},\hat{M}_{i_2;k,l})$. By \eqref{Eq of hat M}, $\hat{M}_{i_1;k,l}$ depends only on $\{\fz_{i_1,1},\cdots,\fz_{i_1,T}\}$, so does $\hat{M}_{i_2;k,l}$, where $\{(\fz_{i_1,t},\fz_{i_2,t}):t=1\cdots,T\}$ is a collection of i.i.d. normal vectors such that $\rho_{i_1,i_2}:=\Cov(\fz_{i_1,t},\fz_{i_2,t})=\langle\Gamma_{i_1},\Gamma_{i_2}\rangle$. Note that
\begin{align*}
	\left\{\begin{array}{ll}
		\Cov(\hat{M}_{i_1;k,l},\hat{M}_{i_2;k,l})=0, & \text{if }\ \rho_{i_1,i_2}=0  \\
		\Cov(\hat{M}_{i_1;k,l},\hat{M}_{i_2;k,l})=\Var(\hat{M}_{i_1;k,l})\lesssim\beta_k^2\beta_l^2, & \text{if }\ |\rho_{i_1,i_2}|=1,
	\end{array}\right.
\end{align*}
so we view $\Cov(\hat{M}_{i_1;k,l},\hat{M}_{i_2;k,l})$ as a function of $\rho_{i_1,i_2}$. Utilizing the covariance analysis in \S\ref{sec of Covariance analysis} and the decay rate of $\beta_k$ in Lemma \ref{Lem of singular structure}, we can prove that $|\Cov(\hat{M}_{i_1;k,l},\hat{M}_{i_2;k,l})|\lesssim\beta_k^2\beta_l^2\langle\Gamma_{i_1},\Gamma_{i_2}\rangle^2$, so  
\begin{align*}
	\frac{1}{n}\Var\left(\sum_{i=1}^n\hat{M}_{i;k,k}\right)\lesssim\frac{\beta_k^2\beta_l^2}{n}\Vert\bbSig\Vert_F^2\lesssim\beta_k^2\beta_l^2\lesssim(kl)^{-2d_1}
\end{align*}
by Assumption \ref{Ap of main}. Combining with the error caused by the Gaussian substitution, we further imply that $n^{-1}\Var\big(\sum_{i=1}^nM_{i;k,l}\big)\lesssim(kl)^{-2d_1}+T^{-1/2}$.

\vspace{3mm}\noindent
{\bf Step 4: Eigenvector Localization and Dimensionality Reduction}\\
We next demonstrate that the tail components of the extreme eigenvectors converge to zero in probability, specifically,
\begin{align}
	\sqrt{n}\sum_{t=\mathtt{k_0}+1}^T\alpha_{1,t}^2\overset{\mbP}{\longrightarrow}0.\label{Eq of eigenvector asymptotic}
\end{align}
Note that for $t\neq1$
\begin{align*}
	\hla_{1,\mathtt{N}}\alpha_{1,t}=\bbw_t'\hR_{\mathtt{N}}\hF_{1,\mathtt{N}}=\sum_{s=1}^T\alpha_{1,s}\sum_{i=1}^nM_{i;t,s}\implies\alpha_{1,t}=\frac{n^{-1}\sum_{s\neq t}^T\alpha_{1,s}\sum_{i=1}^nM_{i;t,s}}{n^{-1}(\hla_{1,\mathtt{N}}-\sum_{i=1}^nM_{i;t,t})},
\end{align*}
we conclude \eqref{Eq of eigenvector asymptotic} by leveraging the strict asymptotic eigengap separating the top $k_0$ singular values from the bulk, a property secured by approximating the scaled discrete Toeplitz matrix with a continuous Volterra operator (Auxiliary lemmas in \S\ref{sec of spectral gap}). Precisely, this eigengap ensures that $n^{-1}(\hla_{1,\mathtt{N}}-\sum_{i=1}^nM_{i;t,t})\gtrsim1$ in probability for $t\geq\mathtt{k_0}+1$, then we prove \eqref{Eq of eigenvector asymptotic} by showing that
$$n^{-3/2}\sum_{t=\mathtt{k_0}+1}^T\left(\sum_{s\neq t}^T\alpha_{1,s}\sum_{i=1}^nM_{i;t,s}\right)^2\overset{\mbP}{\longrightarrow}0.$$
{\bf Step 5: Finite-Dimensional Convergence}\\
Since the eigenvectors localize and $k_0$ is finite, the infinite-dimensional spectral problem reduces to a finite-dimensional one. The asymptotic distribution of $\sqrt{n}(n^{-1}\hla_{1,\mathtt{N}} - \fM_{1,1})$ is determined entirely by the largest eigenvalue of the top-left $\mathtt{k}_0 \times \mathtt{k}_0$ principal submatrix $[n^{-1/2}\sum_{i=1}^n M_{i;k_1,k_2}^\circ]_{\mathtt{k}_0 \times \mathtt{k}_0}$. Finally, we establish the finite-dimensional weakly convergence between the top-left $\mathtt{k}_0 \times \mathtt{k}_0$ principal submatrix and $\bbG_n[1:\mathtt{k}_0]$ defined in \eqref{Eq of limit matrix}.

\subsection{Auxiliary lemmas}\label{sec of Auxiliary lemmas}
In this part, we introduce several important auxiliary lemmas for proving Theorems \ref{Thm of nonstationary} and \ref{Pro of nonstationary LSD}. For simplicity of presentations, the proofs of all lemmas in \S\ref{sec of Auxiliary lemmas} are postponed in the Appendix.

\subsubsection{Basic properties for noises}\label{ssec of Basic properties for noises}
Note that $M_{i;s,t}$ in \eqref{Eq of Mi st} is a ratio of random quadratic forms with a random denominator $\sum_{l=1}^T \sigma_l^2(\Gamma_i \bbve \mfv_l)^2$, applying the Taylor expansion during the Lindeberg replacement generates high powers of the inverse denominator, e.g. $(\sum_{l=1}^T \sigma_l^2(\Gamma_i \bbve \mfv_l)^2)^{-R}$ for some fixed $R\in\mbN^+$. We need the following lemmas to guarantee that the expectations of these terms are uniformly bounded. 
\begin{lemma}\label{Thm of cdf order}
	For any fixed $R\in\mbN^+$, suppose the density function $p_{\vec{z}}(\cdot)$ of a real random vector $\vec{z}:=(z_1,\cdots,z_{2R+1})$ is bounded, then
	\begin{align*}
		\mbE\left[\frac{1}{\big(\sum_{i=1}^{2R+1}z_i^2\big)^R}\right]\leq C_{R,\Vert p_{\vec{z}}\Vert_{\infty}}.
	\end{align*}
\end{lemma}
The next Lemma \ref{Thm of bounded density} bridges Assumption \ref{Ap of main} to the hypothesis of Lemma \ref{Thm of cdf order}, ensuring that the projected noise vectors $(\Gamma_i \bbve \mfv_1, \dots, \Gamma_i \bbve \mfv_R)'$ possess a bounded joint density.

\begin{lemma}\label{Thm of bounded density}
	For any real random matrix $\bbA=[A_{i,j}]\in\mbR^{p\times m}$ with independent entries such that density functions $p_{i,j}$ of $A_{i,j}$ are uniformly bounded by a constant $M>0$, given deterministic vectors $\bba\in\mbR^p$ and $\bbb_1,\cdots,\bbb_R\in\mbR^m$, where $R\in\mbN^+$ is fixed and $R<m$, such that $\bbb_r$ are mutually orthogonal and $\Vert\bba\Vert_2\asymp1,\Vert\bbb_r\Vert_2=1$ for $r=1,\cdots,R$, then the density function of $(\bba'\bbA\bbb_1,\cdots,\bba'\bbA\bbb_R)'$ is bounded by $C_R\Vert\bba\Vert_2^{-R}e^{R/2}M^R$ for a constant $C_R>0$.
\end{lemma}

The final Lemma \ref{Lem of finite 8th moment} aims to control the numerators induced by the Lindeberg replacement.
\begin{lemma}\label{Lem of finite 8th moment}
	For any real random matrix $\bbA=[A_{i,j}]\in\mbR^{p\times m}$ with independent entries such that $\mbE[A_{i,j}]=0,\mbE[A_{i,j}^2]=1$ and $\sup_{i,j}\mbE[A_{i,j}^{2k}]\leq\kappa_{2k}$ for fixed $k\in\mbN^+,k>1$, given any deterministic vectors $\bba\in\mbR^p,\bbb\in\mbR^m$ with $\Vert\bbb\Vert_2=1,\Vert\bba\Vert_2\asymp1$, $\mbE[|\bba'\bbA\bbb|^{2k}]\leq C_k\kappa_{2k}\Vert\bba\Vert_2^{2k}$.
\end{lemma}

\subsubsection{Properties of the discrete Fourier transform on \texorpdfstring{$\mbU$}{U}}\label{ssec of discrete Fourier transform}
For simplicity, we denote all nonstationary roots by
\begin{align}
	\mfr_l=\exp(2\pi{\rm i}\mathtt{p}_l),\quad\mathtt{p}_l\in[0,1),1\leq l\leq m.\label{Eq of nonstationary roots}
\end{align}
As we mentioned in \S\ref{sec of proof outline}, the Lindeberg replacement works due to the delocalization of singular vectors $\mfv_t$. Although proving $\|\mfv_k\|_{\infty}\lesssim T^{-1/2}$ directly is hard, we expand $\mfv_t$ in a shifted discrete Fourier transform (DFT) basis $[\mfu_1,\cdots,\mfu_T]\in\mbC^{T\times T}$ as below:
\begin{align}
	\mfu_s:=(\mfu_{1,s},\cdots,\mfu_{T,s})',\quad\mfu_{r,s}:=\frac{1}{\sqrt{T}}e^{-2\pi{\rm i}(r-1)(s-1+\delta)/T},\label{Eq of mfu}
\end{align}
where $\delta\in(0,1)$ is chosen such that for $l=1,\cdots,m$
\begin{align*}
	\mathtt{p}_l\notin\{(s-1+\delta)/T:s=1,\cdots,T\}.
\end{align*}
In this part, we provide several important lemmas about $\{\mbU\mfu_s:s=1,\cdots,T\}$. Precisely, We sort all $\{\Vert\mbU\mfu_s\Vert_2:s=1,\cdots,T\}$ in a decreasing order as below:
$$\Vert\mbU\mfu_{\pi(1)}\Vert_2\geq\cdots\geq\Vert\mbU\mfu_{\pi(T)}\Vert_2,\quad\{1,\cdots,T\}=\{\pi(1),\cdots,\pi(T)\},$$
then define
\begin{align}
	\theta_s:=\Vert\mbU\mfu_{\pi(s)}\Vert_2\quad\text{and}\quad\bby_s:=\mbU\mfu_{\pi(s)}/\Vert\mbU\mfu_{\pi(s)}\Vert_2.\label{Eq of theta and bby}
\end{align}
\begin{lemma}\label{Lem of singular structure 2}
	Let \(d_1=\max_{1\leq l\leq m}d_l\) be the multiplicity of $\mfr_1$, then 
	\begin{enumerate}
		\item $\theta_1\asymp T^{d_1}$;
		\item $\Vert\bby_t\Vert_{\infty}\lesssim T^{-1/2}$ for \(1\leq t\leq T\);
		\item $\theta_t/\theta_1\lesssim\max\{t^{-d_1},T^{-1}\}$ for $1\leq t\leq T$, which implies that $\sum_{t=1}^T(\theta_t/\theta_1)^2\lesssim1$.
	\end{enumerate}
\end{lemma}
By expanding the right singular vector $\mfw_t=\sum_{r=1}^T\mfc_{t,r}\mfu_{\pi(r)}$ of $\mbU$, we can represent $\mfw_s'\hR\mfw_t=\sum_{i=1}^nM_{i;s,t}$ as a linear combination of DFT-projected ratios $\sum_{i=1}^n\mfN_{i;r_1,r_2}$ in \eqref{Eq of mfN} weighted by coefficients $\mfc_{t,r}$. And the Lindeberg replacement works for $\mfN_{i;r_1,r_2}$ due to $\Vert\bby_t\Vert_{\infty}\lesssim T^{-1/2}$.

By Lemma \ref{Lem of singular structure 2}, when $d_1=1$, the decay rate of $\theta_t/\theta_1\lesssim t^{-1}$ is not fast enough as $\sum_{t=1}^T(\theta_t/\theta_1)\lesssim\log(T)$, the next lemma provides the decay rate of coefficients $\mfc_{k,r}$, which facilitates the eigenvector localization proof when $d_1 = 1$.
\begin{lemma}\label{Lem of one nonstationary root}
	When $d_1=1$, $\mfc_{k,t}:=\langle\mfw_k,\mfu_{\pi(t)}\rangle$ satisfies that $|\mfc_{k,t}|\lesssim T^{-\epsilon}$ for $t<T^{\epsilon}$ and $k>T^{1/2+2\epsilon}$.
\end{lemma}

\subsubsection{Covariance analysis}\label{sec of Covariance analysis}
The following lemma is to control the upper bound of $|\Cov(\hat{M}_{i_1;k,l},\hat{M}_{i_2;k,l})|$.
\begin{lemma}\label{Lem of positive covariance}
	Given a collection of i.i.d. Gaussian vectors 
	\begin{align*}
		\left\{(h_{t},q_{t})'\overset{i.i.d.}{\sim}\mcN\left(\boldsymbol{0},\begin{pmatrix}
			1 & \rho \\
			\rho & 1
		\end{pmatrix}\right):t=1,\cdots,T\right\},\quad|\rho|\leq1,
	\end{align*}
	for any $k,l\in\{1,\cdots,T\}$, we have
	\begin{align}
		\left|\Cov\left(\frac{h_kh_l}{\sum_{t=1}^T\beta_t^2h_t^2},\frac{q_kq_l}{\sum_{t=1}^T\beta_t^2q_t^2}\right)\right|\asymp\rho^2,\label{Eq of positive covariance 2}
	\end{align}
	where $\beta_t$ is defined in \eqref{Eq of beta k}.
\end{lemma}
As a consequence of Lemma \ref{Lem of positive covariance}, $\{\hat{M}_{i;k,l}:i=1,\cdots,n\}$ is a correlated sequence of random variables such that
\begin{align*}
	\big|\Cov(\hat{M}_{i_1;k,l},\hat{M}_{i_2;k,l})\big|\lesssim\frac{\beta_k^2\beta_l^2\Sigma_{i,j}^2}{\Sigma_{i,i}\Sigma_{j,j}},
\end{align*}
which implies that
\begin{align*}
	\Var(G_{k,l;n})=\frac{1}{n}\Var\left(\sum_{i=1}^n\hat{M}_{i;k,l}^{\circ}\right)\lesssim\frac{\beta_k^2\beta_l^2}{n}\|\bbSig\|_F^2\lesssim(kl)^{-2d_1}
\end{align*}
by Assumption \ref{Ap of bbGa} and Lemma \ref{Lem of singular structure}.

\begin{remark}\label{Thm of CLT I1}
	As we mentioned in Remark \ref{Rem of hat M}, when the population covariance matrix $\bbSig$ is either sparse or Toeplitz, $G_{k,l;n}$ is asymptotically normal. Suppose $\bbSig$ is sparse, i.e. $\Sigma_{i,j}=0$ if $|i-j|\gtrsim n^{1/2-\delta}$, by the definition of $\hat{M}_{i;k,l}$ in \eqref{Eq of hat M}, since $\Cov(\fz_{i,t},\fz_{j,t})=\Sigma_{i,j}=0$ when $|i-j|\gtrsim n^{1/2-\delta}$, then $\fz_{i,t}$ and $\fz_{j,t}$ are independent as they are normal variables, so do $\hat{M}_{i;k,l}$ and $\hat{M}_{j;k,l}$. Thus, \(\{\hat{M}_{1;k,l}^{\circ},\cdots,\hat{M}_{n;k,l}^{\circ}\}\) is a collection of $m$-dependent random variables, see \cite[\S7.7]{durrett2019probability}, where $m=\max\{|i-j|:\Sigma_{i,j}=0\}$. By \cite{romano2000more}, we can conclude that $G_{k,l;n}=n^{-1/2}\sum_{i=1}^n\hat{M}_{i;k,l}^{\circ}$ is asymptotically normal.
\end{remark}

\subsection{Proof of Theorem \texorpdfstring{\ref{Thm of nonstationary}}{3.1}}\label{sec of More than two nonstationary roots}
%

\subsubsection{Lindeberg’s principle}
Let’s define a collection of random matrices as follows:
\begin{align}
	\bbZ_{j,r}=\big[\fz_{i,t}^{(j,r)}\big]\in\mbR^{n\times T},\quad \fz_{i,t}^{(j,r)}=\left\{\begin{array}{ll}
		z_{i,t}, & i<j\ \text{or }i=j,t\leq r, \\
		\epsilon_{i,t}, & \text{otherwise},
	\end{array}\right.\label{Eq of vec z}
\end{align}
with $\bbZ_{j,0}:=\bbZ_{j-1,T}$, where $1\leq j\leq n,1\leq r\leq T$ and $\{z_{i,t}\overset{i.i.d.}{\sim}\mcN(0,1):1\leq i\leq n,1\leq t\leq T\}$. Particularly, let $\bbZ_{1,0}=\bbve$. For $M_{i;s,t}$ defined in \eqref{Eq of Mi st}, we denote
\begin{align*}
	M_{i;s,t}(\bbZ_{j,r}):=\frac{\tsigma_s\tsigma_t(\Gamma_i\bbZ_{j,r}\mfv_s)(\Gamma_i\bbZ_{j,r}\mfv_t)}{\sum_{l=1}^T\tsigma_l^2(\Gamma_i\bbZ_{j,r}\mfv_l)^2},
\end{align*}
then we have
\begin{lemma}\label{Lem of Lindeberg two}
	Under Assumptions \ref{Ap of high-dimensionality}, \ref{Ap of main} and \ref{Ap of bbGa}, we have
	\begin{align*}
		&\left|\frac{1}{\sqrt{n}}\sum_{i=1}^n\mbE\left[M_{i;s,t}(\bbZ_{1,0})-M_{i;s,t}(\bbZ_{n,T})\right]\right|\lesssim T^{-1},\\
		&\frac{1}{n}\left|\Var\left(\sum_{i=1}^nM_{i_1;s,t}(\bbZ_{1,0})\right)-\Var\left(\sum_{i=1}^nM_{i_1;s,t}(\bbZ_{n,T})\right)\right|\lesssim T^{-1/2},
	\end{align*}
	uniformly in $1\leq s,t\leq T$.
\end{lemma}
\begin{proof}
	Recall the  unitary matrix $[\mfu_1,\cdots,\mfu_T]$ defined in \eqref{Eq of mfu}, for the right singular vector $\mfw_t$ in \eqref{Eq of beta k}, we denote $\mfw_t=\sum_{r=1}^T\mfc_{t,r}\mfu_{\pi(r)}$ for $\sum_{r=1}^T|\mfc_{t,\pi(r)}|^2=1$, where $\pi(r)$ is defined in \eqref{Eq of theta and bby}, then
	\begin{align}
		&\sum_{i=1}^nM_{i;s,t}=\mfw_s'\hR\mfw_t=\sum_{r_1,r_2=1}^T\mfc_{s,r_1}\mfc_{t,r_2}\mfu_{\pi(r_1)}'\hR\mfu_{\pi(r_2)}\notag\\
		&=\sum_{r_1,r_2=1}^T\mfc_{s,r_1}\mfc_{t,r_2}\sum_{i=1}^n\frac{\theta_{r_1}\theta_{r_2}(\Gamma_i\bbve\bby_{r_1})(\Gamma_i\bbve\bby_{r_2})}{\sum_{l=1}^T\theta_l^2(\Gamma_i\bbve\bby_l)^2}:=\sum_{r_1,r_2=1}^T\mfc_{s,r_1}\mfc_{t,r_2}\sum_{i=1}^n\mfN_{i;r_1,r_2},\label{Eq of mfN}
	\end{align}
	where $\theta_r$ and $\bby_r$ are defined in \eqref{Eq of theta and bby}. Let's show that
	\begin{align*}
		&\big|\mbE\left[\mfN_{i;s,t}(\bbZ_{1,0})-\mfN_{i;s,t}(\bbZ_{n,T})\right]\big|\lesssim(\theta_s\theta_t/\theta_1^2)T^{-3/2},\quad 1\leq s,t\leq T.
	\end{align*}
	Note that
	\begin{align*}
		&\mbE[\mfN_{i;s,t}(\bbZ_{1,0})-\mfN_{i;s,t}(\bbZ_{n,T})]=\sum_{j=1}^n\sum_{r=1}^{T}\mbE[\mfN_{i;s,t}(\bbZ_{j,r-1})-\mfN_{i;s,t}(\bbZ_{j,r})]\\
		&=\sum_{j=1}^n\sum_{r=1}^{T}\mbE[\mfN_{i;s,t}(\bbZ_{j,r-1})-\mfN_{i;s,t}(\bbZ_{j,r-1}^0)+\mfN_{i;s,t}(\bbZ_{j,r-1}^0)-\mfN_{i;s,t}(\bbZ_{j,r})],
	\end{align*}
	where $\bbZ_{j,r}^0=\big[\fz_{i,t}^{(j,r)}(1-\delta_{i,j}\delta_{t,r})\big]\in\mbR^{n\times T}$ is defined in \eqref{Eq of vec z} and $\delta_{i,j}$ is the Kronecker delta. By the 4th order Taylor's expansion with integral remainders, we have
	\begin{footnotesize}
		\begin{align*}
			&\mfN_{i;s,t}(\bbZ_{j,r})-\mfN_{i;s,t}(\bbZ_{j,r}^0)=\sum_{k=1}^4\frac{1}{k!}\partial_{j,r}^k\mfN_{i;s,t}(\bbZ_{j,r}^0)z_{j,r}^k+\int_0^1(1-x)^4\partial_{j,r}^5\mfN_{i;s,t}(\bbZ_{j,r}^0(z_x))z_{j,r}^5dx,\\
			&\mfN_{i;s,t}(\bbZ_{j,r-1})-\mfN_{i;s,t}(\bbZ_{j,r}^0)=\sum_{k=1}^4\frac{1}{k!}\partial_{j,r}^k\mfN_{i;s,t}(\bbZ_{j,r}^0)\epsilon_{j,r}^k+\int_0^1(1-x)^4\partial_{j,r}^5\mfN_{i;s,t}(\bbZ_{j,r}^0(\epsilon_x))\epsilon_{j,r}^5dx,
		\end{align*}
	\end{footnotesize}\noindent
	where $\partial_{j,r}^l=\frac{\partial^l}{(\partial z_{j,r})^l}$ and $\bbZ_{j,r}^0(z_x)=\big[\fz_{i,t}^{(j,r)}(1-\delta_{i,j}\delta_{t,r})+xz_{j,r}\delta_{i,j}\delta_{t,r}\big]\in\mbR^{n\times T}$ for $x\in[0,1]$. By Assumption \ref{Ap of main}, $\mbE[\epsilon_{j,r}^k]=\mbE[z_{j,r}^k]$ for $1\leq k\leq 4$ and $\mbE[\epsilon_{j,r}^{10}]\lesssim\kappa_{10}$, then
	\begin{align*}
		&\big|\mbE[\mfN_{i;s,t}(\bbZ_{j,r})-\mfN_{i;s,t}(\bbZ_{j,r}^0)+\mfN_{i;s,t}(\bbZ_{j,r}^0)-\mfN_{i;s,t}(\bbZ_{j,r+1})]\big|\\
		&\lesssim\int_0^1\mbE\big[\big|\partial_{j,r}^5\mfN_{i;s,t}(\bbZ_{j,r}^0(z_x))\big|^2\big]^{1/2}+\mbE\big[\big|\partial_{j,r}^5\mfN_{i;s,t}(\bbZ_{j,r}^0(\epsilon_x))\big|^2\big]^{1/2}dx,
	\end{align*}
	where we use Cauchy's inequality. By direct calculations, we have
	\begin{footnotesize}
		\begin{align*}
			&\partial_{j,r}^5\mfN_{i;s,t}(\bbZ_{j,r}^0)=C_{5,0}\frac{(\theta_s\theta_t/\theta_1^2)(\Gamma_i\bbZ_{j,r}^0\bby_s)(\Gamma_i\bbZ_{j,r}^0\bby_t)\big(\sum_{l=1}^T(\theta_l/\theta_1)^2(\Gamma_i\bbZ_{j,r}^0\bby_l)(\Gamma_{i,j}y_{r,l})\big)^5}{\big(\sum_{l=1}^T(\theta_l/\theta_1)^2|\Gamma_i\bbZ_{j,r}^0\bby_l|^2\big)^6}\\
			&+C_{5,1}\frac{(\theta_s\theta_t/\theta_1^2)(\Gamma_i\bbZ_{j,r}^0\bby_s)(\Gamma_i\bbZ_{j,r}^0\bby_t)\big(\sum_{l=1}^T(\theta_l/\theta_1)^2(\Gamma_i\bbZ_{j,r}^0\bby_l)(\Gamma_{i,j}y_{r,l})\big)^3\big(\sum_{l=1}^T(\theta_l/\theta_1)^2(\Gamma_{i,j}y_{r,l})^2\big)}{\big(\sum_{l=1}^T(\theta_l/\theta_1)^2|\Gamma_i\bbZ_{j,r}^0\bby_l|^2\big)^5}\\
			&+C_{5,2}\frac{(\theta_s\theta_t/\theta_1^2)(\Gamma_i\bbZ_{j,r}^0\bby_s)(\Gamma_i\bbZ_{j,r}^0\bby_t)\big(\sum_{l=1}^T(\theta_l/\theta_1)^2(\Gamma_i\bbZ_{j,r}^0\bby_l)(\Gamma_{i,j}y_{r,l})\big)\big(\sum_{l=1}^T(\theta_l/\theta_1)^2(\Gamma_{i,j}y_{r,l})^2\big)^2}{\big(\sum_{l=1}^T(\theta_l/\theta_1)^2|\Gamma_i\bbZ_{j,r}^0\bby_l|^2\big)^4}\\
			&+C_{4,0}\frac{(\theta_s\theta_t/\theta_1^2)\{(\Gamma_{i,j}y_{r,s})(\Gamma_i\bbZ_{j,r}^0\bby_t)+(\Gamma_i\bbZ_{j,r}^0\bby_s)(\Gamma_{i,j}y_{r,t})\}\big(\sum_{l=1}^T(\theta_l/\theta_1)^2(\Gamma_i\bbZ_{j,r}^0\bby_l)(\Gamma_{i,j}y_{r,l})\big)^4}{\big(\sum_{l=1}^T(\theta_l/\theta_1)^2|\Gamma_i\bbZ_{j,r}^0\bby_l|^2\big)^5}\\
			&+C_{4,1}\frac{(\theta_s\theta_t/\theta_1^2)\{(\Gamma_{i,j}y_{r,s})(\Gamma_i\bbZ_{j,r}^0\bby_t)+(\Gamma_i\bbZ_{j,r}^0\bby_s)(\Gamma_{i,j}y_{r,t})\}\big(\sum_{l=1}^T(\theta_l/\theta_1)^2(\Gamma_i\bbZ_{j,r}^0\bby_l)(\Gamma_{i,j}y_{r,l})\big)^2\big(\sum_{l=1}^T(\theta_l/\theta_1)^2(\Gamma_{i,j}y_{r,l})^2\big)}{\big(\sum_{l=1}^T(\theta_l/\theta_1)^2|\Gamma_i\bbZ_{j,r}^0\bby_l|^2\big)^4}\\
			&+C_{4,2}\frac{(\theta_s\theta_t/\theta_1^2)\{(\Gamma_{i,j}y_{r,s})(\Gamma_i\bbZ_{j,r}^0\bby_t)+(\Gamma_i\bbZ_{j,r}^0\bby_s)(\Gamma_{i,j}y_{r,t})\}\big(\sum_{l=1}^T(\theta_l/\theta_1)^2(\Gamma_{i,j}y_{r,l})^2\big)^2}{\big(\sum_{l=1}^T(\theta_l/\theta_1)^2|\Gamma_i\bbZ_{j,r}^0\bby_l|^2\big)^3}\\
			&+C_{3,0}\frac{(\theta_s\theta_t/\theta_1^2)(\Gamma_{i,j}y_{r,s})(\Gamma_{i,j}y_{r,t})\big(\sum_{l=1}^T(\theta_l/\theta_1)^2(\Gamma_i\bbZ_{j,r}^0\bby_l)(\Gamma_{i,j}y_{r,l})\big)^3}{\big(\sum_{l=1}^T(\theta_l/\theta_1)^2|\Gamma_i\bbZ_{j,r}^0\bby_l|^2\big)^4}\\
			&+C_{3,1}\frac{(\theta_s\theta_t/\theta_1^2)(\Gamma_{i,j}y_{r,s})(\Gamma_{i,j}y_{r,t})\big(\sum_{l=1}^T(\theta_l/\theta_1)^2(\Gamma_i\bbZ_{j,r}^0\bby_l)(\Gamma_{i,j}y_{r,l})\big)\big(\sum_{l=1}^T(\theta_l/\theta_1)^2(\Gamma_{i,j}y_{r,l})^2\big)}{\big(\sum_{l=1}^T(\theta_l/\theta_1)^2|\Gamma_i\bbZ_{j,r}^0\bby_l|^2\big)^3}.
		\end{align*}
	\end{footnotesize}\noindent
	By Lemma \ref{Lem of singular structure 2} and Cauchy's inequality, we have $|\Gamma_{i,j}y_{r,l}|\lesssim T^{-1/2}|\Gamma_{i,j}|$ and
	\begin{align*}
		&\left|\sum_{l=1}^T(\theta_l/\theta_1)^2(\Gamma_i\bbZ_{j,r}^0\bby_l)(\Gamma_{i,j}y_{r,l})\right|\lesssim T^{-1/2}|\Gamma_{i,j}|\left(\sum_{l=1}^T(\theta_l/\theta_1)^2\right)^{1/2}\left(\sum_{l=1}^T(\theta_l/\theta_1)^2|\Gamma_i\bbZ_{j,r}^0\bby_l|^2\right)^{1/2}\\
		&\lesssim T^{-1/2}|\Gamma_{i,j}|\left(\sum_{l=1}^T(\theta_l/\theta_1)^2|\Gamma_i\bbZ_{j,r}^0\bby_l|^2\right)^{1/2},
	\end{align*}
	it implies that
	\begin{align}
		\big|\partial_{j,r}^5\mfN_{i;s,t}\big|&\lesssim\frac{(\theta_s\theta_t/\theta_1^2)|\Gamma_{i,j}|^5}{T^{5/2}}\Bigg(\frac{1}{\big(\sum_{l=1}^T(\theta_l/\theta_1)^2|\Gamma_i\bbZ_{j,r}^0\bby_l|^2\big)^{5/2}}+\frac{|\Gamma_i\bbZ_{j,r}^0\bby_s|+|\Gamma_i\bbZ_{j,r}^0\bby_t|}{\big(\sum_{l=1}^T(\theta_l/\theta_1)^2|\Gamma_i\bbZ_{j,r}^0\bby_l|^2\big)^3}\notag\\
		&+\frac{|\Gamma_i\bbZ_{j,r}^0\bby_s|^2+|\Gamma_i\bbZ_{j,r}^0\bby_t)|^2}{\big(\sum_{l=1}^T(\theta_l/\theta_1)^2|\Gamma_i\bbZ_{j,r}^0\bby_l|^2\big)^{7/2}}\Bigg).\label{Eq of 5th derivative}
	\end{align}
	Here, we claim that $\mbE\big[\big|\partial_{j,r}^5\mfN_{i;s,t}(\bbZ_{j,r}^0(z_x))\big|^2\big]^{1/2}\lesssim(\theta_s\theta_t/\theta_1^2)|\Gamma_{i,j}|^5T^{-5/2}$ for $x\in[0,1]$. According to Lemmas \ref{Thm of cdf order}, \ref{Lem of finite 8th moment} and Parseval's identity, we have
	\begin{align*}
		&\mbE\left[\frac{|\Gamma_i\bbZ_{j,r}^0\bby_s|^4}{\big(\sum_{l=1}^T(\theta_l/\theta_1)^2|\Gamma_i\bbZ_{j,r}^0\bby_l|^2\big)^{7}}\right]\leq\mbE[|\Gamma_i\bbZ_{j,r}^0\bby_s|^{8}]^{1/2}\mbE\left[\frac{1}{\big(\sum_{l=1}^T(\theta_l/\theta_1)^2|\Gamma_i\bbZ_{j,r}^0\bby_l|^2\big)^{14}}\right]^{1/2}\\
		&\lesssim\kappa_{10}\mbE\left[\frac{1}{\big(\sum_{l=1}^T\beta_l^2(\Gamma_i\bbZ_{j,r}^0\bbv_l)^2\big)^{14}}\right]^{1/2}\lesssim\mbE\left[\frac{1}{\big(\sum_{l=1}^{29}(\Gamma_i\bbZ_{j,r}^0\mfv_l)^2\big)^{14}}\right]^{1/2}\lesssim1,
	\end{align*}
	so do other terms in \eqref{Eq of 5th derivative}, then it deduces that
	\begin{align*}
		&\int_0^1\mbE\big[\big|\partial_{j,r}^5\mfN_{i;s,t}(\bbZ_{j,r}^0(z_x))\big|^2\big]^{1/2}dx\lesssim(\theta_s\theta_t/\theta_1^2)|\Gamma_{i,j}|^5T^{-5/2}.
	\end{align*}
	Combining with Assumption \ref{Ap of bbGa}, we have
	\begin{align*}
		&\big|\mbE[\mfN_{i;s,t}(\bbZ_{1,0})-\mfN_{i;s,t}(\bbZ_{n,T})]\big|\\
		&\lesssim\sum_{j=1}^n\sum_{r=1}^T\int_0^1\mbE\big[\big|\partial_{j,r}^5\mfN_{i;s,t}(\bbZ_{j,r}^0(z_x))\big|^2\big]^{1/2}+\mbE\big[\big|\partial_{j,r}^5\mfN_{i;s,t}(\bbZ_{j,r-1}^0(\epsilon_x))\big|^2\big]^{1/2}dx\\
		&\lesssim(\theta_s\theta_t/\theta_1^2)\sum_{j=1}^n\sum_{r=1}^T|\Gamma_{i,j}|^5T^{-5/2}\lesssim(\theta_s\theta_t/\theta_1^2) T^{-3/2}.
	\end{align*}
	Consequently, we obtain that
	\begin{align}
		&\frac{1}{\sqrt{n}}\Bigg|\sum_{i=1}^n\mbE[M_{i;s,t}(\bbZ_{1,0})-M_{i;s,t}(\bbZ_{n,T})]\Bigg|\notag\\
		&\leq\frac{1}{\sqrt{n}}\sum_{i=1}^n\sum_{r_1,r_2=1}^T|\mfc_{s,r_1}\mfc_{t,r_2}|\times\mbE[\mfN_{i;r_1,r_2}(\bbZ_{1,0})-\mfN_{i;r_1,r_2}(\bbZ_{n,T})]\notag\\
		&\lesssim T^{-1}\Bigg(\sum_{r=1}^T|\mfc_{s,r}|\theta_r/\theta_1\Bigg)\Bigg(\sum_{r=1}^T|\mfc_{t,r}|\theta_r/\theta_1\Bigg)\lesssim T^{-1}\sum_{r=1}^T\theta_r^2/\theta_1^2\lesssim T^{-1}.\label{Eq of Lindeberg mean}
	\end{align}
	Similarly, by the previous arguments, we can obtain that
	\begin{align*}
		\big|\mbE[\mfN_{i_1;s,t}(\bbZ_{1,0})\mfN_{i_2;s,t}(\bbZ_{1,0})-\mfN_{i_1;s,t}(\bbZ_{n,T})\mfN_{i_2;s,t}(\bbZ_{n,T})]\big|\lesssim(\theta_s\theta_t/\theta_1^2)^2T^{-3/2},
	\end{align*}
	it further deduces that
	\begin{small}
		\begin{align*}
			&\frac{1}{n}\Bigg|\Var\Bigg(\sum_{i=1}^nM_{i;s,t}(\bbZ_{1,0})\Bigg)-\Var\Bigg(\sum_{i=1}^nM_{i;s,t}(\bbZ_{n,T})\Bigg)\Bigg|\\
			&\leq\frac{1}{n}\sum_{i_1,i_2=1}^n\sum_{r_1,r_2,r_3,r_4=1}^T|\mfc_{s,r_1}\mfc_{t,r_2}\mfc_{s,r_3}\mfc_{t,r_4}|\mbE[\mfN_{i_1;r_1,r_2}(\bbZ_{1,0})\mfN_{i_2;r_3,r_4}(\bbZ_{1,0})-\mfN_{i_1;r_1,r_2}(\bbZ_{n,T})\mfN_{i_2;r_3,r_4}(\bbZ_{n,T})]\\
			&\lesssim  T^{-1/2}\Bigg(\sum_{r=1}^T\mfc_{s,r}\theta_r/\theta_1\Bigg)^2\Bigg(\sum_{r=1}^T\mfc_{t,r}\theta_r/\theta_1\Bigg)^2\lesssim T^{-1/2},
		\end{align*}
	\end{small}\noindent
	which completes our proof.
\end{proof}
As a consequence of Lemma \ref{Lem of Lindeberg two}, by the definition of $\hat{M}_{i;s,t}$ in \eqref{Eq of hat M}, $\hat{M}_{i;s,t}$ and $M_{i;s,t}(\bbZ_{n,T})$ have the same distribution, then we have
\begin{align*}
	\sqrt{n}\Bigg|\frac{1}{n}\sum_{i=1}\mbE[M_{i;s,t}(\bbZ_{1,0})]-\fM_{s,t}\Bigg|\lesssim T^{-1}.
\end{align*}
The following corollary leverages Lemma \ref{Lem of one nonstationary root} to achieve a smaller error induced by the Lindeberg replacement when $d_1=1$.
\begin{corollary}\label{Cor of Lindeberg}
	Under Assumptions \ref{Ap of high-dimensionality}, \ref{Ap of main} and \ref{Ap of bbGa}, let $\epsilon>0$ be a sufficiently small constant, when $d_1=1$, we have
	\begin{align*}
		&\left|\frac{1}{\sqrt{n}}\sum_{i=1}^n\mbE\left[M_{i;s,t}(\bbZ_{1,0})-M_{i;s,t}(\bbZ_{n,T})\right]\right|\lesssim \log(T)T^{-1-\epsilon},\\
		&\frac{1}{n}\left|\Var\left(\sum_{i=1}^nM_{i_1;s,t}(\bbZ_{1,0})\right)-\Var\left(\sum_{i=1}^nM_{i_1;s,t}(\bbZ_{n,T})\right)\right|\lesssim \log(T)T^{-1/2-\epsilon},
	\end{align*}
	uniformly for $T^{1/2+2\epsilon}<\max\{s,t\}\leq T$.
\end{corollary}
\begin{proof}
	By \eqref{Eq of Lindeberg mean}, we have
	\begin{align*}
		\frac{1}{\sqrt{n}}\Bigg|\sum_{i=1}^n\mbE[M_{i;s,t}(\bbZ_{1,0})-M_{i;s,t}(\bbZ_{n,T})]\Bigg|\lesssim T^{-1}\Bigg(\sum_{r=1}^T|\mfc_{s,r}|\theta_r/\theta_1\Bigg)\Bigg(\sum_{r=1}^T|\mfc_{t,r}|\theta_r/\theta_1\Bigg).
	\end{align*}
	Without loss of generality, assume that $s>T^{1/2+\epsilon}$, by Lemma \ref{Lem of one nonstationary root}, when $r<T^{\epsilon}$, we have $|\mfc_{s,t}|\lesssim T^{-\epsilon}$, combining with $\theta_r/\theta_1\lesssim r^{-1}$ by Lemma \ref{Lem of singular structure 2}, it implies that
	\begin{align*}
		\sum_{r=1}^T|\mfc_{s,r}|\theta_r/\theta_1\lesssim T^{-\epsilon}\sum_{r=1}^{\lceil T^{\epsilon}\rceil}r^{-1}+\left(\sum_{r=\lceil T^{\epsilon}\rceil+1}^T(\theta_r/\theta_1)^2\right)^{1/2}\lesssim T^{-\epsilon}\log(T),
	\end{align*}
	which concludes the first conclusion. Similar arguments are also valid for the second term, we omit details here.
\end{proof}

\subsubsection{Asymptotic behaviors of eigenvectors}\label{Sec of eigenvectors}
In this part, we will prove that the eigenvector $\hF_{k,\mathtt{N}}=\sum_{t=1}^T\alpha_{k,t}\mfw_t$ for $1\leq k\leq\mathtt{k}_0$ is localized under the basis $[\mfw_1,\cdots,\mfw_T]$.
\begin{proposition}\label{Pro of eigenvectors two}
	Under Assumptions \ref{Ap of high-dimensionality}, \ref{Ap of main} and \ref{Ap of bbGa}, for $\mathtt{k_0}$ defined in \eqref{Eq of bbG n} and $1\leq k\leq\mathtt{k_0}$,
	\begin{align*}
		\sqrt{n}\sum_{t=\mathtt{k_0}+1}^T\alpha_{k,t}^2\overset{\mbP}{\longrightarrow}0,
	\end{align*}
	where $\alpha_{k,t}$ is defined in \eqref{Eq of Mi st}.
\end{proposition}
\begin{proof}
	When $k=1$, since $\hla_{1,\mathtt{N}}\alpha_{1,t}=\mfw_t'\hR_{\mathtt{N}}\hF_{1,\mathtt{N}}=\sum_{k=1}^T\alpha_{1,k}\sum_{i=1}^nM_{i;k,t}$, then
	$$\alpha_{1,t}=\frac{\sum_{k\neq t}^T\alpha_{1,k}n^{-1}\sum_{i=1}^nM_{i;k,t}}{n^{-1}\hla_1-n^{-1}\sum_{i=1}^nM_{i;t,t}}.$$
	By Lemmas \ref{Lem of positive covariance} and \ref{Lem of Lindeberg two}, $n^{-1}\Var\big(\sum_{i=1}^nM_{i;s,t}\big)\lesssim\beta_s^2\beta_t^2+T^{-1/2}$, combining with Chebyshev's inequality, for any sufficiently small constant $\delta>0$, we have
	\begin{align*}
		\mbP\left(\frac{1}{n}\Bigg|\sum_{i=1}^n\big(M_{i;s,t}-\mbE[M_{i;s,t}]\big)\Bigg|>n^{-\delta/2}\right)\lesssim\frac{\beta_s^2\beta_t^2+T^{-1/2}}{n^{1-\delta}},
	\end{align*}
	then let 
	\begin{align}
		\mcE:=\left\{\exists\ 1\leq s,t\leq T^{1/2-\delta}\ \text{such that }\ \frac{1}{n}\Bigg|\sum_{i=1}^n\big(M_{i;s,t}-\mbE[M_{i;s,t}]\big)\Bigg|>n^{-\delta/2}\right\},\label{Eq of mcE}
	\end{align}
	combining with Lemma \ref{Lem of singular structure}, it implies that
	\begin{align*}
		&\mbP\left(\mcE\right)\leq\sum_{s,t=1}^{\lceil T^{1/2-\delta}\rceil}\mbP\left(\frac{1}{n}\Bigg|\sum_{i=1}^n\big(M_{i;s,t}-\mbE[M_{i;s,t}]\big)\Bigg|>n^{-\delta/2}\right)\lesssim\frac{1}{n^{1/2+\delta}}.
	\end{align*}
	Next, since $\hla_{1,\mathtt{N}}\geq \mfw_1'\hR_{\mathtt{N}}\mfw_1= \sum_{i=1}^nM_{i;1,1}$ and $\mbE[M_{i;1,1}]>\mbE[M_{i;t,t}]$ for $\mathtt{k_0}<t\leq T^{1/2-\delta}$, conditional on $\mcE^c$, we have
	\begin{align*}
		&\frac{\hla_{1,\mathtt{N}}}{n}-\frac{1}{n}\sum_{i=1}^nM_{i;t,t}\geq\frac{1}{n}\sum_{i=1}^n(M_{i;1,1}-M_{i;t,t})\geq\frac{1}{n}\sum_{i=1}^n\mbE[M_{i;1,1}-M_{i;t,t}]-2n^{-\delta/2}\\
		&\geq\frac{1}{n}\sum_{i=1}^n\mbE[M_{i;1,1}-M_{i;\mathtt{k_0}+1,\mathtt{k_0}+1}]-2n^{-\delta/2}\geq\frac{\fM_{1,1}-\fM_{\mathtt{k_0}+1,\mathtt{k_0}+1}}{2},
	\end{align*}
	where we use the fact that $\big|\mbE[M_{i;t,t}]-\fM_{t,t}\big|\lesssim T^{-2/3}$ by Lemma \ref{Lem of Lindeberg two}. For $t>T^{1/2-\delta}$, since $\sum_{t=1}^TM_{i;t,t}=1$, then conditional on $\mcE^c$, for a large integer $K\asymp\log(T)$,
	\begin{align}
		&\frac{1}{n}\sum_{i=1}^nM_{i;t,t}\leq 1-\sum_{s=1}^{K^2}\frac{1}{n}\sum_{i=1}^nM_{i;s,s}\leq1-\sum_{s=1}^{K^2}\frac{1}{n}\sum_{i=1}^n\mbE[M_{i;s,s}]+K^2n^{-\delta/2}\label{Eq of denominator}\\
		&\leq\frac{1}{n}\sum_{r=K^2+1}^T\sum_{i=1}^n\mbE[M_{i;r,r}]+K^2n^{-\delta/2}\lesssim\sum_{r=K^2+1}^T\beta_r^2+K^2n^{-\delta/2}\lesssim K^{-2}+K^2n^{-\delta/2}\longrightarrow0,\notag
	\end{align}
	where we use $\beta_r\lesssim r^{-1}$ in Lemma \ref{Lem of singular structure}. Consequently, we obtain that $\frac{\hla_{1,\mathtt{N}}}{n}-\frac{1}{n}\sum_{i=1}^nM_{i;t,t}\geq\frac{\fM_{1,1}-\fM_{\mathtt{k_0}+1,\mathtt{k_0}+1}}{2}$ conditional on $\mcE^c$ and
	\begin{align}
		&\sqrt{n}\sum_{t=\mathtt{k_0}+1}^T\mbE[\alpha_{1,t}^2]\lesssim\sqrt{n}\mbP(\mcE)+\sqrt{n}\sum_{t=\mathtt{k_0}+1}^T\mbE[\alpha_{1,t}^2|\mcE^c]\mbP(\mcE^c)\notag\\
		&\lesssim n^{-\delta}+\frac{4}{n^{3/2}(\fM_{1,1}-\fM_{\mathtt{k_0}+1,\mathtt{k_0}+1})^2}\sum_{t=\mathtt{k_0}+1}^T\mbE\left[\left(\sum_{k\neq t}^T\alpha_{1,k}\sum_{i=1}^nM_{i;k,t}\right)^2\right].\label{Eq of alpha 1,1}
	\end{align}
	Without loss of generality, let $d_1=\max_{1\leq l\leq m}d_l$, we consider two cases that $d_1\geq2$ and $d_1=1$.
	
	\vspace{3mm}\noindent
	$\mathbf{d_1\geq2}$: By Cauchy's inequality, 
	\begin{align*}
		&\frac{1}{n^{3/2}}\sum_{t=\mathtt{k_0}+1}^T\mbE\left[\left(\sum_{k\neq t}^T\alpha_{1,k}\sum_{i=1}^nM_{i;k,t}\right)^2\right]\leq\frac{1}{n^{3/2}}\sum_{t=\mathtt{k_0}+1}^T\sum_{k\neq t}^T\mbE\left[\left(\sum_{i=1}^nM_{i;k,t}\right)^2\right],
	\end{align*}
	combining with Lemmas \ref{Lem of positive covariance} and \ref{Lem of Lindeberg two}, let $\mfh_T:=\lceil T^{1/2-\delta}\rceil$, we have
	\begin{align}
		&\frac{1}{n^{3/2}}\sum_{t=\mathtt{k_0}+1,k\neq t}^{\mfh_T}\mbE\left[\left(\sum_{i=1}^nM_{i;k,t}\right)^2\right]=\frac{1}{n^{3/2}}\sum_{t=\mathtt{k_0}+1,k\neq t}^{\mfh_T}\left\{\Var\left(\sum_{i=1}^nM_{i;k,t}\right)+\left(\sum_{i=1}^n\mbE[M_{i;k,t}]\right)^2\right\}\notag\\
		&\lesssim n^{-1/2}\sum_{k,t=1}^{\mfh_T}(\beta_k^2\beta_t^2+T^{-1/2})\lesssim T^{-2\delta},\label{Eq of alpha 1,1 1}
	\end{align}
	and
	\begin{align}
		&\frac{1}{n^{3/2}}\sum_{\substack{k\neq t\\ k>\mfh_T\ {\rm or}\ t>\mfh_T}}^T\mbE\left[\left(\sum_{i=1}^nM_{i;k,t}\right)^2\right]\leq\frac{1}{n^{1/2}}\sum_{\substack{k\neq t\\ k>\mfh_T\ {\rm or}\ t>\mfh_T}}^T\sum_{i=1}^n\mbE[M_{i;k,t}^2]\notag\\
		&\lesssim\sum_{\substack{k\neq t\\ k>\mfh_T\ {\rm or}\ t>\mfh_T}}^T\sqrt{n}\beta_k^2\beta_t^2\leq\sqrt{n}\sum_{t>\mfh_T}^T\beta_t^2\lesssim T^{-1+3\delta},\label{Eq of alpha 1,1 2}
	\end{align}
	where we use $\beta_t\lesssim t^{-d_1}\lesssim t^{-2}$ in Lemma \ref{Lem of singular structure}. Combining \eqref{Eq of alpha 1,1}, \eqref{Eq of alpha 1,1 1} and \eqref{Eq of alpha 1,1 2}, we obtain that
	\begin{align*}
		\sqrt{n}\sum_{t=\mathtt{k_0}+1}^T\mbE[\alpha_{1,t}^2]\lesssim T^{-\delta}.
	\end{align*}
	Next, for $2\leq k\leq K$, by Courant–Fischer–Weyl min-max principle, we have
	\begin{align*}
		\hla_{k,\mathtt{N}}\geq\min_{\mfb_{k,1}^2+\cdots+\mfb_{k,k}^2=1}\sum_{s,r=1}^k\mfb_{k,s}\mfb_{k,r}\sum_{i=1}^nM_{i;s,r},
	\end{align*}
	conditional $\mcE^c$ in \eqref{Eq of mcE}, for $\mathtt{k_0}+1\leq t\leq\lceil T^{1/2-\delta}\rceil$, we have
	\begin{align}
		&\frac{\hla_{k,\mathtt{N}}}{n}-\frac{1}{n}\sum_{i=1}^nM_{i;t,t}\geq\sum_{s,r=1}^k\mfb_{k,s}\mfb_{k,r}\frac{1}{n}\sum_{i=1}^nM_{i;s,r}-\frac{1}{n}\sum_{i=1}^nM_{i;t,t}\notag\\
		&\geq\sum_{s=1}^k\mfb_{k,s}^2\frac{1}{n}\sum_{i=1}^n\mbE[M_{i;s,s}]-\frac{1}{n}\sum_{i=1}^n\mbE[M_{i;t,t}]-n^{-\delta}\left(1+\sum_{s,r=1}^k|\mfb_{k,s}\mfb_{k,r}|\right)\notag\\
		&\geq\frac{1}{n}\sum_{i=1}^n\mbE[M_{i;k,k}-M_{i;\mathtt{k_0}+1,\mathtt{k_0}+1}]-n^{-\delta}(1+k^2)\geq\frac{\fM_{1,1}-\fM_{\mathtt{k_0}+1,\mathtt{k_0}+1}}{2},\label{Eq of kth denominator}
	\end{align}
	where we use the fact that $\fM_{1,1}=\cdots=\fM_{\mathtt{k_0},\mathtt{k_0}}$. For $\lceil T^{1/2-\delta}\rceil<t\leq T$, we can use \eqref{Eq of denominator} to show that \eqref{Eq of kth denominator} conditional on $\mcE^c$. By the same arguments as those for \eqref{Eq of alpha 1,1}, we derive that
	\begin{align*}
		&\sqrt{n}\sum_{t=\mathtt{k_0}+1}^T\mbE[\alpha_{k,t}^2]\lesssim\sqrt{n}\mbP(\mcE)+\sum_{t=\mathtt{k_0}+1}^T\mbE[\alpha_{k,t}^2|\mcE^c]\mbP(\mcE^c)\\
		&\lesssim n^{-\delta}+\frac{4}{n^{3/2}(\fM_{1,1}-\fM_{\mathtt{k_0}+1,\mathtt{k_0}+1})^2}\sum_{t=\mathtt{k_0}+1}^T\mbE\left[\left(\sum_{l\neq t}^T\alpha_{k,l}\sum_{i=1}^nM_{i;l,t}\right)^2\right]\lesssim T^{-\delta},
	\end{align*}\noindent
	$\mathbf{d_1=1}$: We first show that $\sqrt{n}\sum_{t=\mathtt{k_0}+1}^T\mbE[\alpha_{1,t}^2]\lesssim1$. Similar as \eqref{Eq of alpha 1,1 1} and \eqref{Eq of alpha 1,1 2}, we have
	\begin{align}
		\frac{1}{n^{3/2}}\sum_{t=\mathtt{k_0}+1,k\neq t}^{\lceil T^{1/2}\rceil}\mbE\left[\left(\sum_{i=1}^nM_{i;k,t}\right)^2\right]\lesssim n^{-1/2}\sum_{k,t=1}^{\lceil T^{1/2}\rceil}(\beta_k^2\beta_t^2+T^{-1/2})\lesssim1,\label{Eq of d1=1 1}
	\end{align}
	and
	\begin{align}
		&\frac{1}{n^{3/2}}\sum_{\substack{k\neq t\\ k>\lceil T^{1/2}\rceil\ {\rm or}\ t>\lceil T^{1/2}\rceil}}^T\mbE\left[\left(\sum_{i=1}^nM_{i;k,t}\right)^2\right]\lesssim\sqrt{n}\sum_{t>\lceil T^{1/2}\rceil}^T\beta_t^2\lesssim 1,\label{Eq of d1=1 2}
	\end{align}
	where we use $\beta_t\lesssim t^{-1}$, so we obtain that $\sqrt{n}\sum_{t=\mathtt{k_0}+1}^T\mbE[\alpha_{1,t}^2]\lesssim1$ by \eqref{Eq of alpha 1,1}. Next, to conclude $\sqrt{n}\sum_{t=\mathtt{k_0}+1}^T\alpha_{1,t}^2\overset{\mbP}{\longrightarrow}0$, by \eqref{Eq of alpha 1,1}, it suffices to show that 
	\begin{align}
		\frac{1}{n^{3/2}}\sum_{t=\mathtt{k_0}+1}^T\left(\sum_{k\neq t}^T\alpha_{1,k}\sum_{i=1}^nM_{i;k,t}\right)^2\overset{\mbP}{\longrightarrow}0.\label{Eq of d1=1 5}
	\end{align}
	Let $K_T\asymp T^{\delta/2}$, note that
	\begin{small}
		\begin{align*}
			\sum_{t=\mathtt{k_0}+1}^T\left(\sum_{k\neq t}^T\alpha_{1,k}\sum_{i=1}^nM_{i;k,t}\right)^2\lesssim\sum_{t=\mathtt{k_0}+1}^T\left[\left(\sum_{k\neq t}^{K_T}\alpha_{1,k}\sum_{i=1}^nM_{i;k,t}\right)^2+\left(\sum_{k\neq t,k\geq K_T}^T\alpha_{1,k}\sum_{i=1}^nM_{i;k,t}\right)^2\right],
		\end{align*}
	\end{small}\noindent
	by Corollary \ref{Cor of Lindeberg}, for any sufficiently small constant $\delta>0$, we have
	\begin{align}
		&\frac{1}{n^{3/2}}\sum_{t=\mathtt{k_0}+1}^T\mbE\left[\left(\sum_{k\neq t}^{K_T}\alpha_{1,k}\sum_{i=1}^nM_{i;k,t}\right)^2\right]\lesssim\frac{1}{n^{3/2}}\sum_{t=\mathtt{k_0}+1}^T\sum_{k\neq t}^{K_T}\mbE\left[\left(\sum_{i=1}^nM_{i;k,t}\right)^2\right]\notag\\
		&\lesssim n^{-1/2}\sum_{k\neq t}^{K_T}\sum_{t=\mathtt{k_0}+1}^{\lceil T^{1/2+2\delta}\rceil}(\beta_t^2\beta_k^2+T^{-1/2})+n^{-1/2}\sum_{k\neq t}^{K_T}\sum_{t=\lceil T^{1/2+2\delta}\rceil+1}^{T}(\beta_t^2\beta_k^2+\log(T)T^{-1/2-\delta})\notag\\
		&\lesssim \log(T)T^{-\delta/2}.\label{Eq of d1=1 3}
	\end{align}
	Moreover, define an event 
	\begin{align*}
		\tilde{\mcE}_{\delta}:=\left\{\frac{1}{n^{3/2}}\sum_{t=\mathtt{k_0}+1}^T\sum_{k\neq t,k> K_T}^T\left(\sum_{i=1}^nM_{i;k,t}\right)^2\leq n^{\delta}\right\},
	\end{align*}
	by \eqref{Eq of d1=1 1} and \eqref{Eq of d1=1 2}, we have $\mbP(\tilde{\mcE}_{\delta}^c)\lesssim n^{-\delta}$, so conditional on $\tilde{\mcE}_{\delta}$, we have
	\begin{small}
		\begin{align}
			&\frac{1}{n^{3/2}}\sum_{t=\mathtt{k_0}+1}^T\left(\sum_{k\neq t,k> K_T}^T\alpha_{1,k}\sum_{i=1}^nM_{i;k,t}\right)^2\lesssim\frac{1}{n^{3/2}}\left(\sum_{k\neq t,k> K_T}^T\alpha_{1,k}^2\right)\left(\sum_{t=\mathtt{k_0}+1}^T\sum_{k\neq t,k> K_T}^T\left(\sum_{i=1}^nM_{i;k,t}\right)^2\right)\notag\\
			&\lesssim n^{\delta}\sum_{k\neq t,k> K_T}^T\alpha_{1,k}^2\overset{\mbP}{\longrightarrow}0,\label{Eq of d1=1 4}
		\end{align}
	\end{small}\noindent
	where we use $\sqrt{n}\sum_{k\neq t,k> K_T}^T\mbE[\alpha_{1,k}^2]\lesssim1$. Combining \eqref{Eq of d1=1 3} and \eqref{Eq of d1=1 4}, we conclude \eqref{Eq of d1=1 5}, which completes the proof.
\end{proof}
\begin{corollary}\label{Cor of eigenvector}
	As a consequence of Proposition \ref{Pro of eigenvectors two}, for $1\leq l\leq\mathtt{k_0}$, we have
	\begin{align*}
		\sqrt{n}\left(\frac{\hla_{l,\mathtt{N}}}{n}-\sum_{k_1,k_2=1}^{\mathtt{k_0}}\alpha_{l,k_1}\alpha_{l,k_2}\frac{1}{n}\sum_{i=1}^nM_{i;k_1,k_2}\right)\overset{\mbP}{\longrightarrow}0.
	\end{align*}
\end{corollary}
\begin{proof}
	Note that
	\begin{small}
		\begin{align*}
			&\frac{\hla_{l,\mathtt{N}}}{n}-\sum_{k_1,k_2=1}^{\mathtt{k_0}}\alpha_{l,k_1}\alpha_{l,k_2}\frac{1}{n}\sum_{i=1}^nM_{i;k_1,k_2}=\sum_{\substack{r_1\neq r_2\\ r_1,r_2=\mathtt{k_0}+1}}^T\alpha_{l,r_1}\alpha_{l,r_2}\frac{1}{n}\sum_{i=1}^nM_{i;r_1,r_2}+\sum_{r=\mathtt{k_0}+1}^T\alpha_{l,r}^2\frac{1}{n}\sum_{i=1}^nM_{i;r,r},
		\end{align*}
	\end{small}\noindent
	where
	\begin{align*}
		&\sqrt{n}\sum_{r=\mathtt{k_0}+1}^T\alpha_{l,r}^2\frac{1}{n}\sum_{i=1}^nM_{i;r,r}\leq\sqrt{n}\sum_{r=\mathtt{k_0}+1}^T\alpha_{l,r}^2\times\frac{1}{n}\sum_{r=\mathtt{k_0}+1}^T\sum_{i=1}^nM_{i;r,r}\leq\sqrt{n}\sum_{r=\mathtt{k_0}+1}^T\alpha_{l,r}^2\overset{\mbP}{\longrightarrow}0.
	\end{align*}
	By Cauchy's inequality, Proposition \ref{Pro of eigenvectors two}, \eqref{Eq of alpha 1,1 1}, \eqref{Eq of alpha 1,1 2} and \eqref{Eq of d1=1 5}, we have
	\begin{small}
		\begin{align*}
			&\Bigg|\sum_{r_1=\mathtt{k_0}+1}^T\alpha_{l,r_1}\sum_{r_2\neq r_1,r_2=1}^T\alpha_{l,r_2}\frac{1}{\sqrt{n}}\sum_{i=1}^nM_{i;r_1,r_2}\Bigg|^2\leq\sum_{r_1=\mathtt{k_0}+1}^T\alpha_{l,r_1}^2\sum_{r_1=\mathtt{k_0}+1}^T\Bigg|\sum_{r_2\neq r_1,r_2=1}^T\alpha_{l,r_2}\frac{1}{\sqrt{n}}\sum_{i=1}^nM_{i;r_1,r_2}\Bigg|^2\\
			&\leq\sqrt{n}\sum_{r_1=\mathtt{k_0}+1}^T\alpha_{l,r_1}^2\times\frac{1}{n^{3/2}}\sum_{r_1=\mathtt{k_0}+1}^T\Bigg|\sum_{r_2\neq r_1,r_2=1}^T\alpha_{l,r_2}\sum_{i=1}^nM_{i;r_1,r_2}\Bigg|^2\overset{\mbP}{\longrightarrow}0,
		\end{align*}
	\end{small}\noindent
	which completes the proof.
\end{proof}

\subsubsection{Asymptotic behaviors of eigenvalues}\label{sec of weak convergence}
According to Corollary \ref{Cor of eigenvector}, we show that the asymptotic distribution of $n^{-1}\hla_{1,\mathtt{N}}$ depending on $n^{-1/2}\sum_{i=1}^nM_{i;k_1,k_2}$ for $1\leq k_1,k_2\leq\mathtt{k}_0$, we provide the following finite dimensional weakly convergence:
\begin{lemma}\label{Lem of asymptotic eigenvalues two}
	Under Assumptions \ref{Ap of high-dimensionality}, \ref{Ap of main} and \ref{Ap of bbGa}, for any fixed $K\in\mbN^+$, we have
	\begin{align*}
		\left\{\frac{1}{\sqrt{n}}\sum_{i=1}^nM_{i;s,t}(\bbZ_{1,0}):1\leq s\leq t\leq K\right\}\overset{d}{\sim}\left\{\frac{1}{\sqrt{n}}\sum_{i=1}^n\hat{M}_{i;s,t}:1\leq s\leq t\leq K\right\},
	\end{align*}
	where $M_{i;s,t}(\bbZ_{1,0})$ is defined in \eqref{Eq of Mi st} and \eqref{Eq of vec z}, $\hat{M}_{i;s,t}$ is defined in \eqref{Eq of hat M}.
\end{lemma}
\begin{proof}
	For $1\leq k\leq l\leq K$, denote $\mcK:=K(K+1)/2$ and
	\begin{align*}
		g_{K(k-1)+l}(\bbZ_{j,r}):=\frac{1}{\sqrt{n}}\sum_{i=1}^nM_{i;k,l}^{\circ}(\bbZ_{j,r}).
	\end{align*}
	Note that $\hat{M}_{i;s,t}$ and $M_{i;s,t}(\bbZ_{n,T})$ are identically distributed, it suffices to show that for any bounded and Lipschitz continuous $\mcK$-dimensional function $f(x_1,\cdots,x_{\mcK})$,
	\begin{align*}
		\lim_{n\to\infty}\big|\mbE[f(g_1(\bbZ_{1,0}),\cdots,g_{\mcK}(\bbZ_{1,0}))-f(g_1(\bbZ_{n,T}),\cdots,g_{\mcK}(\bbZ_{n,T}))]\big|=0.
	\end{align*}
	Although $f$ may not be differentiable, let's construct 
	\begin{align*}
		&f_{\epsilon}(x_1,\cdots,x_{\mcK})=\frac{1}{(2\pi)^{-\mcK/2}}\int_{\mbR^{\mcK}}f(x_{1}-\epsilon y_{1},\cdots,x_{\mcK}-\epsilon y_{\mcK})e^{-(y_{1}^2+\cdots+y_{\mcK}^2)/2}dy_{1}\cdots dy_{\mcK},
	\end{align*}
	by \cite[\S\RomanNumeralCaps{8}]{feller-vol-2}, we have
	\begin{align*}
		\lim_{\epsilon\downarrow0}\sup_{\vec{x}\in\mbR^{\mcK}}\big|f_{\epsilon}(\vec{x})-f(\vec{x})\big|=0.
	\end{align*}
	Notably, for any $\epsilon>0$, $f_{\epsilon}$ is smooth and its $l$-th partial derivatives are uniformly bounded by $\mrO(\epsilon^{-l})$. Thus, it suffices to show that for any $\epsilon=\epsilon_T>0$ such that $\epsilon\asymp\log(T)^{-1}$,
	\begin{align*}
		\lim_{n\to\infty}\big|\mbE[f_{\epsilon}(g_{1}(\bbZ_{1,0}),\cdots,g_{\mcK}(\bbZ_{1,0}))-f_{\epsilon}(g_{1}(\bbZ_{n,T}),\cdots,g_{\mcK}(\bbZ_{n,T}))]\big|=0.
	\end{align*}
	For simplicity, denote $h(\bbZ_{j,r}):=f_{\epsilon}(g_{1}(\bbZ_{j,r}),\cdots,g_{\mcK}(\bbZ_{j,r}))$, note that
	\begin{align*}
		&\mbE[h(\bbZ_{1,0})-h(\bbZ_{n,T})]=\sum_{j=1}^n\sum_{r=1}^{T}\mbE[h(\bbZ_{j,r-1})-h(\bbZ_{j,r-1}^0)+h(\bbZ_{j,r-1}^0)-h(\bbZ_{j,r})],
	\end{align*}
	where $\bbZ_{j,r}^0$ is defined in Lemma \ref{Lem of Lindeberg two}. By the 4th order Taylor's expansion with integral remainder, we have
	\begin{align*}
		&h(\bbZ_{j,r})-h(\bbZ_{j,r}^0)=\sum_{l=1}^4\frac{1}{l!}\partial_{j,r}^lh(\bbZ_{j,r}^0)z_{j,r}^l+\int_0^1(1-x)^4\partial_{j,r}^5h(\bbZ_{j,r}^0(z_x))z_{j,r}^5dx,\\
		&h(\bbZ_{j,r-1})-h(\bbZ_{j,r}^0)=\sum_{l=1}^4\frac{1}{l!}\partial_{j,r}^lh(\bbZ_{j,r}^0)\epsilon_{j,r}^l+\int_0^1(1-x)^4\partial_{j,r}^5h(\bbZ_{j,r}^0(\epsilon_x))\epsilon_{j,r}^5dx,
	\end{align*}
	where $\partial_{j,r}^l=\frac{\partial^l}{(\partial z_{j,r})^l}$ and $\bbZ_{j,r}^0(z_x)$ are defined in Lemma \ref{Lem of Lindeberg two}. Since $z_{j,r}$ and $\bbZ_{j,r}^0$ are independent, combining with Assumption \ref{Ap of main}, it concludes that
	\begin{align}
		&\big|\mbE[h(\bbZ_{1,0})-h(\bbZ_{n,T})]\big|\label{Eq of Lindeberg}\\
		&\leq\sum_{j=1}^n\sum_{r=1}^{T}\int_0^1(1-x)^4\mbE\big[\big|\partial_{j,r}^5h(\bbZ_{j,r}^0(\epsilon_x))\big|\epsilon_{j,r}^4+\big|\partial_{j,r}^5h(\bbZ_{j,r}^0(z_x))\big|z_{j,r}^4\big]dx\notag\\
		&\lesssim\kappa_8\sum_{j=1}^n\sum_{r=1}^{T}\int_0^1\mbE\big[\big|\partial_{j,r}^5h(\bbZ_{j,r}^0(\epsilon_x))\big|^2\big]^{1/2}+\mbE\big[\big|\partial_{j,r}^5h(\bbZ_{j,r}^0(z_x))\big|^2\big]^{1/2}dx.\notag
	\end{align}
	By direct calculations, we have
	\begin{align*}
		&\partial_{j,r}^5h=\sum_{k_1,\cdots,k_5=1}^{\mcK}\partial_{k_1\cdots k_5}f_{\epsilon}\prod_{l=1}^5\partial_{j,r}g_{k_l}+10\sum_{k_1,\cdots,k_4=1}^{\mcK}\partial_{k_1\cdots k_4}f_{\epsilon}\partial_{j,r}^2g_{k_1}\prod_{l=2}^4\partial_{j,r}g_{k_l}\\
		&+\sum_{k_1,k_2,k_3=1}^{\mcK}\partial_{k_1k_2k_3}f_{\epsilon}\{12\partial_{j,r}^3g_{k_1}\partial_{j,r}g_{k_2}\partial_{j,r}g_{k_3}+15\partial_{j,r}^2g_{k_1}\partial_{j,r}^2g_{k_2}\partial_{j,r}g_{k_3}\}\\
		&+\sum_{k_1,k_2=1}^{\mcK}\partial_{k_1k_2}f_{\epsilon}\{7\partial_{j,r}^4g_{k_1}\partial_{j,r}g_{k_2}+9\partial_{j,r}^3g_{k_1}\partial_{j,r}^2g_{k_2}\}+\sum_{k_1=1}^{\mcK}\partial_{k_1}f_{\epsilon}\partial_{j,r}^5g_{k_1},
	\end{align*}
	since $|\partial_{k_1\cdots k_5}f_{\epsilon}|\lesssim\log(T)^5$, then 
	\begin{small}
		\begin{align}
			&\big|\partial_{j,r}^5h\big|\label{Eq of partial5 h}\\
			&\leq \log(T)^5\Bigg(\Bigg(\sum_{k=1}^{\mcK}\big|\partial_{j,r}g_k\big|\Bigg)^5+\Bigg(\sum_{k=1}^{\mcK}\big|\partial_{j,r}g_k\big|\Bigg)^3\sum_{k=1}^{\mcK}\big|\partial_{j,r}^2g_k\big|+\Bigg(\sum_{k=1}^{\mcK}\big|\partial_{j,r}g_k\big|\Bigg)^2\sum_{k=1}^{\mcK}\big|\partial_{j,r}^3g_k\big|\notag\\
			&+\Bigg(\sum_{k=1}^{\mcK}\big|\partial_{j,r}^2g_k\big|\Bigg)^2\sum_{k=1}^{\mcK}\big|\partial_{j,r}g_k\big|+\sum_{k=1}^{\mcK}\big|\partial_{j,r}^4g_k\big|\sum_{k=1}^{\mcK}\big|\partial_{j,r}g_k\big|+\sum_{k=1}^{\mcK}\big|\partial_{j,r}^2g_k\big|\sum_{k=1}^{\mcK}\big|\partial_{j,r}^3g_k\big|+\sum_{k=1}^{\mcK}\big|\partial_{j,r}^5g_k\big|\Bigg).\notag
		\end{align}
	\end{small}\noindent
	Recall that $\sum_{i=1}^nM_{i;k,l}^{\circ}=\sum_{r_1,r_2=1}^T\mfc_{k,r_1}\mfc_{l,r_2}\sum_{i=1}^n\mfN_{i;r_1,r_2}^{\circ}$ defined in \eqref{Eq of mfN}, the upper bound of $\mbE\big[\big|\partial_{j,r}^5h(\bbZ_{j,r}^0(\epsilon_x))\big|^2\big]^{1/2}$ includes the following term:
	\begin{align*}
		&\log(T)^5\mbE\big[\big|\partial_{j,r}g_{k,l}(\bbZ_{j,r}(\epsilon_x))\big|^{10}\big]^{1/2}\lesssim\frac{\log(T)^5}{n^{5/2}}\mbE\Bigg[\Bigg|\sum_{i=1}^n\partial_{j,r}M_{i;k,l}^{\circ}(\bbZ_{j,r}(\epsilon_x))\Bigg|^{10}\Bigg]^{1/2}\\
		&=\frac{\log(T)^5}{n^{5/2}}\mbE\Bigg[\Bigg|\sum_{r_1,r_2=1}^T\mfc_{k,r_1}\mfc_{l,r_2}\sum_{i=1}^n\partial_{j,r}\mfN_{i;r_1,r_2}^{\circ}(\bbZ_{j,r}(\epsilon_x))\Bigg|^{10}\Bigg]^{1/2},
	\end{align*}
	by direction calculations as in \eqref{Eq of 5th derivative}, we have
	\begin{align*}
		&\big|\partial_{j,r}\mfN_{i;s,t}^{\circ}(\bbZ_{j,r}^0)\big|\lesssim\frac{(\theta_s\theta_t/\theta_1^2)^{1/2}|\Gamma_{i,j}|}{T^{1/2}}\frac{|\Gamma_i\bbZ_{j,r}^0\bby_s|+|\Gamma_i\bbZ_{j,r}^0\bby_t|}{\sum_{l=1}^T(\theta_l/\theta_1)^2|\Gamma_i\bbZ_{j,r}^0\bby_l|^2},
	\end{align*}
	then by Lemmas \ref{Thm of cdf order}, \ref{Thm of bounded density} and Assumptions \ref{Ap of main}, \ref{Ap of bbGa},
	\begin{align*}
		&\log(T)^5\mbE\big[\big|\partial_{j,r}g_{k,l}(\bbZ_{j,r}(\epsilon_x))\big|^{10}\big]^{1/2}\\
		&\lesssim\frac{\log(T)^5}{T^5}\mbE\Bigg[\Bigg|\sum_{r_1,r_2=1}^T\mfc_{k,r_1}\mfc_{l,r_2}(\theta_{r_1}\theta_{r_2}/\theta_1^2)^{1/2}\sum_{i=1}^n|\Gamma_{i,j}|\frac{|\Gamma_i\bbZ_{j,r}^0\bby_{r_1}|+|\Gamma_i\bbZ_{j,r}^0\bby_{r_2}|}{\sum_{t=1}^T(\theta_l/\theta_1)^2(\Gamma_i\bbZ_{j,r}^0\bby_t)^2}\Bigg|^{10}\Bigg]^{1/2}\\
		&\lesssim\frac{\kappa_{10}\log(T)^5}{T^5}\left(\sum_{i=1}^n|\Gamma_{i,j}|\right)^5\left(\sum_{r=1}^T|\mfc_{k,r}|(\theta_r/\theta_1)^{1/2}\right)^5\left(\sum_{r=1}^T|\mfc_{l,r}|(\theta_r/\theta_1)^{1/2}\right)^5\lesssim\log(T)^{10}T^{-5/2},
	\end{align*}
	where we use the fact that $\theta_r/\theta_1\lesssim r^{-1}$ in Lemma \ref{Lem of singular structure 2}. For other terms in \eqref{Eq of partial5 h}, we can show that they are bounded by $\log(T)^{10}T^{-5/2}$ by the same way, we omit details here for simplicity. Consequently, we have $\big|\mbE[h(\bbZ_{1,0})-h(\bbZ_{n,T})]\big|\lesssim\log(T)^{10}T^{-1/2}$, which completes our proof.
\end{proof}
Finally, we present the proof of Theorem \ref{Thm of nonstationary}.
\begin{proof}[Proof of Theorem \ref{Thm of nonstationary}]
	By Corollary \ref{Cor of eigenvector}, we can obtain that
	\begin{align*}
		\sqrt{n}\left(\frac{\hla_{1,\mathtt{N}}}{n}-\sum_{k_1,k_2=1}^{\mathtt{k_0}}\alpha_{1,k_1}\alpha_{1,k_2}\frac{1}{n}\sum_{i=1}^nM_{i;k_1,k_2}(\bbZ_{1,0})\right)\overset{\mbP}{\longrightarrow}0,
	\end{align*}
	combining Lemma \ref{Lem of Lindeberg two} and the fact that $\fM_{k_1,k_2}=0$ for $k_1\neq k_2$, it implies that
	\begin{align*}
		&\sqrt{n}\left(\frac{\hla_{1,\mathtt{N}}}{n}-\fM_{1,1}\right)\overset{\mbP}{\sim}\sum_{k_1,k_2=1}^{\mathtt{k_0}}\alpha_{1,k_1}\alpha_{1,k_2}\frac{1}{\sqrt{n}}\sum_{i=1}^nM_{i;k_1,k_2}^{\circ}(\bbZ_{1,0})\\
		&+\sum_{k=1}^{\mathtt{k}_0}\alpha_{1,k}^2\sqrt{n}(\fM_{k,k}-\fM_{1,1})-\fM_{1,1}\sqrt{n}\sum_{k=\mathtt{k}_0+1}^T\alpha_{1,k}^2.
	\end{align*}
	By Proposition \ref{Pro of eigenvectors two} and the definition of $\mathtt{k_0}$ in \eqref{Eq of bbG n}, we have
	\begin{align*}
		&\sum_{k=1}^{\mathtt{k}_0}\alpha_{1,k}^2\sqrt{n}(\fM_{k,k}-\fM_{1,1})-\fM_{1,1}\sqrt{n}\sum_{k=\mathtt{k}_0+1}^T\alpha_{1,k}^2\overset{\mbP}{\longrightarrow}0,
	\end{align*}
	so 
	\begin{align*}
		\sqrt{n}\left(\frac{\hla_{1,\mathtt{N}}}{n}-\fM_{1,1}\right)\overset{\mbP}{\sim}\sum_{k_1,k_2=1}^{\mathtt{k_0}}\alpha_{1,k_1}\alpha_{1,k_2}\frac{1}{\sqrt{n}}\sum_{i=1}^nM_{i;k_1,k_2}^{\circ}(\bbZ_{1,0}).
	\end{align*}
	For simplicity, let
	\begin{align*}
		\bbH_n^{\mathtt{N}}(\mathtt{k_0})=\left[\frac{1}{\sqrt{n}}\sum_{i=1}^nM_{i;k_1,k_2}^{\circ}(\bbZ_{1,0})\right]\in\mbR^{\mathtt{k_0}\times\mathtt{k_0}},
	\end{align*}
	then we have
	\begin{align*}
		\sqrt{n}\left(\frac{\hla_{1,\mathtt{N}}}{n}-\fM_{1,1}\right)\leq\lambda_1(\bbH_n^{\mathtt{N}}(\mathtt{k_0})),\quad\text{in probability.}
	\end{align*}
	On the other hand, by Proposition \ref{Pro of eigenvectors two} and the definition of $\mathtt{k_0}$ in \eqref{Eq of bbG n},
	$$\|\bbH_n^{\mathtt{N}}(\mathtt{k_0})-\tilde{H}_n(\mathtt{k_0})\|\overset{\mbP}{\longrightarrow}0,$$
	where $\tilde{H}_n(\mathtt{k_0})$ is the left-upper $\mathtt{k_0}\times\mathtt{k_0}$ principle submatrix of $\sqrt{n}\big(n^{-1}\hR_{\mathtt{N}}-\fM_{1,1}\bbI_T\big)$, so we have
	\begin{align*}
		\sqrt{n}\left(\frac{\hla_{1,\mathtt{N}}}{n}-\fM_{1,1}\right)\geq\lambda_1(\bbH_n^{\mathtt{N}}(\mathtt{k_0}))+\mro_{\mbP}(1),
	\end{align*}
	which implies that
	\begin{align*}
		\sqrt{n}\left(\frac{\hla_{1,\mathtt{N}}}{n}-\fM_{1,1}\right)\overset{\mbP}{\sim}\lambda_1(\bbH_n^{\mathtt{N}}(\mathtt{k_0})).
	\end{align*}
	For $2\leq l\leq \mathtt{k_0}$, by Lemma \ref{Lem of Lindeberg two}, we have
	\begin{align*}
		\sqrt{n}\left(\frac{\hla_{l,\mathtt{N}}}{n}-\fM_{1,1}\right)\overset{\mbP}{\sim}\sum_{k_1,k_2=1}^{\mathtt{k_0}}\alpha_{l,k_1}\alpha_{l,k_2}\frac{1}{\sqrt{n}}\sum_{i=1}^nM_{i;k_1,k_2}^{\circ}(\bbZ_{1,0}),
	\end{align*}
	which implies that 
	\begin{align*}
		\sum_{r=1}^l\sqrt{n}\left(\frac{\hla_{r,\mathtt{N}}}{n}-\fM_{1,1}\right)\leq\sum_{r=1}^l\lambda_r(\bbH_n^{\mathtt{N}}(\mathtt{k_0})),\quad\text{in probability,}
	\end{align*}
	by the induction assumption of $\sqrt{n}(n^{-1}\hla_{r,\mathtt{N}}-\fM_{1,1})\overset{\mbP}{\sim}\lambda_r(\bbH_n^{\mathtt{N}}(\mathtt{k_0}))$ for $1\leq r\leq l-1$, we have $\sqrt{n}(n^{-1}\hla_{l,\mathtt{N}}-\fM_{1,1})\leq\lambda_l(\bbH_n^{\mathtt{N}}(\mathtt{k_0}))$ in probability. On the other hand,
	\begin{align*}
		\sum_{r=1}^l\sqrt{n}\left(\frac{\hla_{r,\mathtt{N}}}{n}-\fM_{1,1}\right)\geq\sum_{r=1}^l\lambda_r(\bbH_n^{\mathtt{N}}(\mathtt{k_0})),
	\end{align*}
	by the induction assumption of $\sqrt{n}(n^{-1}\hla_{r,\mathtt{N}}-\fM_{1,1})\overset{\mbP}{\sim}\lambda_r(\bbH_n^{\mathtt{N}}(\mathtt{k_0}))$ for $1\leq r\leq l-1$, we have $\sqrt{n}(n^{-1}\hla_{l,\mathtt{N}}-\fM_{1,1})\geq\lambda_l(\bbH_n^{\mathtt{N}}(\mathtt{k_0}))$, i.e.
	\begin{align*}
		\sqrt{n}\left(\frac{\hla_{l,\mathtt{N}}}{n}-\fM_{1,1}\right)\overset{\mbP}{\sim}\lambda_l(\bbH_n^{\mathtt{N}}(\mathtt{k_0})),\quad l=1,\cdots,\mathtt{k_0}.
	\end{align*}
	Finally, combining Lemma \ref{Lem of asymptotic eigenvalues two}, we have $\bbH_n^{\mathtt{N}}(\mathtt{k_0})\overset{d}{\sim}\bbG_n(\mathtt{k_0})$, where $\bbG_n$ is defined in \eqref{Eq of limit matrix}. Now, we complete the proof by the continuous mapping theorem.
\end{proof}

\subsection{Proof of Theorem \texorpdfstring{\ref{Pro of nonstationary LSD}}{2.2}}\label{sec of LSD}
\begin{proof}
	Let $M_T=\lceil T/\log(T)\rceil$, since $\tr(\hR_{s2})=n$, then $\sum_{t=1}^T\hla_{t,\mathtt{N}}=n$ and
	\begin{align*}
		\sum_{t=M_T+1}^T\hla_{t,\mathtt{N}}=n-\sum_{t=1}^{M_T}\hla_{t,\mathtt{N}}\leq n-\sum_{t=1}^{M_T}\sum_{i=1}^nM_{i;t,t}=\sum_{t=M_T+1}^T\sum_{i=1}^nM_{i;t,t},
	\end{align*}
	where we use the fact that $\sum_{t=1}^TM_{i;t,t}=1$. By Lemma \ref{Thm of cdf order}, we have
	\begin{align*}
		\sum_{t=M_T+1}^T\sum_{i=1}^n\mbE[M_{i;t,t}]\lesssim n\sum_{t=M_T+1}^T\beta_t^2\lesssim nM_T^{-1}\lesssim\log(T).
	\end{align*}
	By Markov's inequality, it gives that
	\begin{align*}
		\mbP\left(\sum_{t=M_T+1}^T\sum_{i=1}^nM_{i;t,t}>\log^2(T)\right)\lesssim\frac{1}{\log(T)}.
	\end{align*}
	Let $R_k(\epsilon):=\#\{\hla_{t,\mathtt{N}}>\epsilon:t=k+1,\cdots,T\}$ for $k=0,1,\cdots,T-1$, where $\epsilon\asymp\log^{-1}(T)$, conditional on $\sum_{t=M_T+1}^T\hla_{t,\mathtt{N}}\lesssim\log^2(T)$, we have
	\begin{align*}
		R_{M_T}(\epsilon)\epsilon\leq\sum_{t=M_T+1}^T\hla_{t,\mathtt{N}}\lesssim\log^2(T)\implies R_{M_T}(\epsilon)\lesssim\log^3(T).
	\end{align*}
	Finally, note that $T^{-1}R_0(\epsilon)\leq T^{-1}(R_{M_T}(\epsilon)+M_T)\lesssim\log^{-1}(T)$ in probability, we complete the proof.
\end{proof}

\section{Super nonstationary scenario}\label{Sec of Super nonstationary scenario}
In this section, we assume that there exists at least one root $\mfr_l$ of $f_X(\cdot)$ in \eqref{Eq of characteristic polynomial 1} outside the unit circle. Similar as the beginning of \S\ref{Sec of Nonstationary scenario}, we make the following notational conventions for roots of $f_X(\cdot)$:
\begin{itemize}
	\item Let $\mfr_1,\cdots,\mfr_{m_{\mathtt{SN}}}$ be all different super nonstationary roots of $f_X(\cdot)$ such that $|\mfr_1|\geq\cdots\geq|\mfr_{m_{\mathtt{SN}}}|$.
	\item Assume that $|\mfr_1|=\cdots=|\mfr_{\tilde{m}}|$, where  $\tilde{m}:=\#\{l=1,\cdots,m_{\mathtt{SN}}:|\mfr_l|=|\mfr_1|\}$, and the corresponding multiplicities of $d_1\geq\cdots\geq d_{\tilde{m}}$.
	\item Let $\tilde{m}_0:=\#\{l=1,\cdots,\tilde{m}:d_l=d_1\}$.
\end{itemize}
As we mentioned in \S\ref{Sec of Nonstationary scenario}, to analyze the asymptotic spectral properties of $\hR$, the key is to study the spectral gap of $\mbU$. Under the scenario of {\bf SN}, we have
\begin{proposition}\label{Pro of spectral gap SN}
	For $\mathtt{k}_0$ and $\beta_k$ defined in \eqref{Eq of bbG n} and \eqref{Eq of beta k}, respectively, we have
	\begin{enumerate}
		\item $\mathtt{k}_0$ is the multiplicities of the largest eigenvalue of $\mcG_T:=\bbY\mcD_T^*\bbZ\mcD_T\in\mbR^{\tilde{m}_0\times \tilde{m}_0}$, and for $1\leq k\leq\tilde{m}_0$,
		$$|\beta_k-\lambda_k(\mcG_T)^{1/2}/\lambda_1(\mcG_T)^{1/2}|\lesssim T^{-1},$$
		where $\bbY=[Y_{l_1,l_2}],\bbZ=[Z_{l_1,l_2}]\in\mbR^{\tilde{m}_0\times \tilde{m}_0}$ are two positive definite matrices such that
		\begin{align*}
			\left\{\begin{array}{l}
				Y_{l_1,l_2}=(a_{d_1,l_1}\bar{a}_{d_1,l_2})(\mfr_{l_1}\bar{\mfr}_{l_2}-1)^{-1}, \\
				Z_{l_1,l_2}=\{(d_1-1)!\}^{-2}|\mfr_1|^2(\mfr_{l_2}\bar{\mfr}_{l_1}-1)^{-1},
			\end{array}\right.\quad a_{d_1,l}=\prod_{r\neq l}^{\fd}(1-\mfr_r/\mfr_l)^{-d_r}
		\end{align*}
		and $\mcD_T:=\diag((\mfr_1/|\mfr_1|)^{T+1},\cdots,(\mfr_{\tilde{m}_0}/|\mfr_1|)^{T+1})$.
		\item $\beta_k\lesssim T^{-1}$ for $\tilde{m}_0<k\leq d_{\mathtt{SN}}$, where $d_{\mathtt{SN}}=\sum_{l=1}^{\tilde{m}}d_l$.
		\item $\log(\beta_k)\lesssim-T\log(|\mfr_1|/|\mfr_{\tilde{m}+1}|)$ for $k>d_{\mathtt{SN}}$.
		\item For left/right singular vectors $\mfv_k,\mfw_k$ of $\mbU$ defined in \eqref{Eq of SVD mbU} corresponding to its $k$-th largest singular value, 
		\begin{align*}
			\left\{\begin{array}{l}
				\log(|\mfv_{t,k}|)\lesssim-t\log(|\mfr_1|),\\
				\log(|\mfw_{t,k}|)\lesssim-(T-t)\log(|\mfr_1|),
			\end{array}\right.\quad k=1,\cdots,\tilde{m}_0.
		\end{align*}
	\end{enumerate}
\end{proposition}
\begin{remark}
	Proposition \ref{Pro of spectral gap SN} shows that the rank of $\mbU/\|\mbU\|$ is $\tilde{m}_0$ as $T\to\infty$, which is finite. This is quite different from those under the scenario of {\bf N}, see Proposition \ref{Lem of hat M}. The key reason is that $\|\mbU\|$ tends to infinity in an exponential rate of $T$ instead of a polynomial rate. 
\end{remark}
Consequently, we have
\begin{theorem}\label{Thm of super nonstationary 1}
	Under Assumptions \ref{Ap of high-dimensionality}, \ref{Ap of main} and \ref{Ap of bbGa}, the eigenvalues of $\hR_{\mathtt{SN}}$ satisfies that for any $\epsilon>0$,
	\begin{align*}
		\left\{\begin{array}{l}
			\sum_{k=d_{\mathtt{SN}}+1}^T\hla_{k,\mathtt{SN}}\overset{\mbP}{\longrightarrow}0, \\
			n^{-\epsilon}\hla_{k,\mathtt{SN}}\overset{\mbP}{\longrightarrow}0,\quad\tilde{m}_0<k\leq d_{\mathtt{SN}}.
		\end{array}\right.
	\end{align*}
\end{theorem}
The existence of super nonstationary roots further lead $\hR_{\mathtt{SN}}$ degenerate to a low structure. Thus, the LSD of $\hR_{\mathtt{SN}}$ is a point distribution with all mass at 0. Particularly, when $d_{\mathtt{SN}}=1$, $\hla_{1,\mathtt{SN}}\overset{\mbP}{\longrightarrow}n$ and $\hR_{\mathtt{SN}}$ asymptotically becomes a rank-one matrix.

Next, when $\hR_{\mathtt{SN}}$ does not degenerate to a rank-one matrix, we have
\begin{theorem}\label{Thm of super nonstationary 2}
	Under Assumptions \ref{Ap of high-dimensionality} and \ref{Ap of bbGa}, if $\varepsilon_t\overset{i.i.d.}{\sim}\mcN(\boldsymbol{0},\bbI_n)$, we have
	\begin{align*}
		&\sqrt{n}\left(\frac{\hla_{1,\mathtt{SN}}}{n}-\fM_{1,1},\cdots,\frac{\hla_{\mathtt{k_0},\mathtt{SN}}}{n}-\fM_{1,1}\right)\overset{d}{\sim}\big(\lambda_1(\bbG_n[1:\mathtt{k_0}]),\cdots,\lambda_{\mathtt{k_0}}(\bbG_n[1:\mathtt{k_0}])\big),
	\end{align*}
	where $\fM_{1,1}$ and $\bbG_n$ are defined in \eqref{Eq of asymptotic mean} and \eqref{Eq of limit matrix}, respectively.
\end{theorem}
\begin{remark}
	\begin{enumerate}
		\item Compared with Theorem \ref{Thm of nonstationary}, we use Gaussian noises $\varepsilon_t$ instead of Assumption \ref{Ap of main} in Theorem \ref{Thm of super nonstationary 2}. The key reason is that the singular vectors of $\mbU$ in \eqref{Eq of mbU} are not delocalized, see Proposition \ref{Pro of spectral gap SN}, so the Lindeberg's principle does not work under the scenario {\bf SN}. Recall $M_{i;k,l}$ in \eqref{Eq of Mi st}, it is a ratio of quadratic forms of $\beta_k\Gamma_i\bbve\mfv_k$, which does not degnerate to 0 when $1\leq k,l\leq\tilde{m}_0$ by Lemma \ref{Lem of spectral gap SN}. Since $|\mfv_{t,k}|$ has an exponential decay rate in $t$, then the CLT for $\Gamma_i\bbve\mfv_k$ does hot hold.
		\item For $\mathtt{k_0}<k\leq\tilde{m}_0$, Proposition \ref{Pro of spectral gap SN} implies that $\fM_{k,k}\asymp 1$ since $\bbY,\bbZ\in\mbR^{\tilde{m}_0\times\tilde{m_0}}$ are positive definite. We can inductively establish the asymptotic distribution of $\hla_{k,\mathtt{SN}}$ as we have discussed in Remark \ref{Rem of eigenvalue N}, and proofs are totally the same as those for Theorem \ref{Thm of super nonstationary 2}, so we omit details here.
	\end{enumerate}
\end{remark}

\subsection{Proof of Proposition \texorpdfstring{\ref{Pro of spectral gap SN}}{3.1}}
Before presenting the proof of Proposition \ref{Pro of spectral gap SN}, we need the following lemma.
\begin{lemma}\label{Lem of spectral gap SN}
	Given any fixed $\rho\in\mbC$ such that $|\rho|>1$ and fixed $d\in\mbN^+$, for $\mcT(\rho)$ defined in \eqref{Eq of mcT}, we have
	\begin{enumerate}
		\item $\|\mcT(\rho)^{-d}\|\asymp T^{d-1}|\rho|^T$.
		\item Let $\bbu_l(\rho)=(u_{1,l}(\rho),\cdots,u_{T,l}(\rho))'$ and $\bbv_l(\rho)=(v_{1,l}(\rho),\cdots,v_{T,l}(\rho))'$ for $l=0,1,\cdots,d-1$ and
		\begin{align}
			\left\{\begin{array}{l}
				u_{l,t}(\rho):=(-1)^{d-1-l}\rho^{-t}\binom{t-l-1}{d-1-l},\\
				v_{l,t}(\rho):=\bar{\rho}^t\binom{t}{l},
			\end{array}\right.\quad t=1,\cdots,T,\label{Eq of rank-one approximation}
		\end{align}
		then 
		\begin{align*}
			\left\{\begin{array}{l}
				\|\mcT(\rho)^{-d}-\sum_{l=0}^{d-1}\bbu_l(\rho)\bbv_l(\rho)^*\|\lesssim1, \\
				\|\mcT(\rho)^{-d}-\bbu_{d-1}(\rho)\bbv_{d-1}(\rho)^*\|\lesssim T^{-1}\|\mcT(\rho)^{-d}\|.
			\end{array}\right.
		\end{align*}
	\end{enumerate}
\end{lemma}
The proof of Lemma \ref{Lem of spectral gap SN} is postponed in the Appendix \ref{App of spectral gap SN}. Next, in the following proof of Proposition \ref{Pro of spectral gap SN}, we only consider the non-degenerate case that $\tilde{m}_0<\tilde{m}<m_{\mathtt{SN}}$.
\begin{proof}[Proof of Proposition \ref{Pro of spectral gap SN}]
	\begin{enumerate}
		\item Using the same partial fraction decomposition as in the proof of Proposition \ref{Lem of hat M}, we have
		$$\mbU=\sum_{l=1}^{\fd}\sum_{k_l=1}^{d_l}a_{k_l,l}\mcT(\mfr_l)^{-k_l},\quad a_{k_l,l}\in\mbC.$$
		By Lemma \ref{Lem of spectral gap SN}, we have
		\begin{align*}
			\left\{\begin{array}{ll}
				\log(\|\mcT(\mfr_l)^{-k_l}\|/\|\mcT(\mfr_1)^{-d_1}\|)\lesssim-T, & l>\tilde{m}, \\
				\|\mcT(\mfr_l)^{-k_l}\|/\|\mcT(\mfr_1)^{-d_1}\|\lesssim T^{-1}, & \tilde{m}_0<l\leq\tilde{m}, \\
				\|\mcT(\mfr_l)^{-k_l}\|/\|\mcT(\mfr_1)^{-d_1}\|\lesssim T^{-1} & 1\leq l\leq \tilde{m}_0,k_l<d_1,
			\end{array}\right.
		\end{align*}
		let $\widetilde{\mbU}:=\sum_{l=1}^{\tilde{m}_0}a_{d_1,l}\mcT(\mfr_l)^{-d_1}$, then $\|\widetilde{\mbU}\|\asymp T^{d_1-1}|\mfr_1|^T$ and $\|\mbU-\widetilde{\mbU}\|\lesssim T^{-1}\|\widetilde{\mbU}\|$. Next, let's analyze leading singular values of $\widetilde{\mbU}$. According to Lemma \ref{Lem of spectral gap SN}, we have 
		$$\|\mcT(\mfr_l)^{-d_1}-\bbu_{d_1-1}(\mfr_l)\bbv_{d_1-1}(\mfr_l)'\|\lesssim T^{-1}\|\widetilde{\mbU}\|,$$
		which implies that
		\begin{align}
			\left\|\widetilde{\mbU}-\sum_{l=1}^{\tilde{m}_0}a_{d_1,l}\bbu_{d_1-1}(\mfr_l)\bbv_{d_1-1}(\mfr_l)'\right\|\lesssim T^{-1}\|\widetilde{\mbU}\|.\label{Eq of low rank}
		\end{align}
		Consequently, it suffices to analyze the singular values of $\sum_{l=1}^{\tilde{m}_0}a_{d_1,l}\bbu_{d_1-1}(\mfr_l)\bbv_{d_1-1}(\mfr_l)'$. Here, we define two auxiliary matrices as below:
		\begin{align*}
			\left\{\begin{array}{l}
				\bbY_T:=[Y_{l_1,l_2;T}]_{\tilde{m}_0\times \tilde{m}_0},\quad Y_{l_1,l_2;T}=a_{d_1,l_1}\bar{a}_{d_1,l_2}\langle\bbu_{d_1-1}(\mfr_{l_1}),\bbu_{d_1-1}(\bar{\mfr}_{l_2})\rangle, \\
				\bbZ_T:=[Z_{l_1,l_2;T}]_{\tilde{m}_0\times \tilde{m}_0},\quad Z_{l_1,l_2;T}=T^{2-2d_1}|\mfr_1|^{-2T}|\langle\bbv_{d_1-1}(\mfr_{l_1}),\bbv_{d_1-1}(\bar{\mfr}_{l_2})\rangle.
			\end{array}\right.
		\end{align*}
		Note that
		\begin{align*}
			Z_{l_1,l_2;T}=T^{2-2d_1}|\mfr_1|^{-2T}\sum_{t=1}^T(\bar{\mfr}_{l_1}\mfr_{l_2})^{T+1-t}\binom{T+1-t}{d_1-1}^2,
		\end{align*}
		by direct calculations, $|Z_{l_1,l_2;T}(\bar{\mfr}_{l_1}\mfr_{l_2}/|\mfr_1|^2)^{T-1}-Z_{l_1,l_2}|\lesssim T^{-1}$ as $T\to\infty$, so does $|Y_{l_1,l_2;T}-Y_{l_1,l_2}|\lesssim T^{-1}$. Thus, we have $\|\bbZ_T-\mcD_T^*\bbZ\mcD_T\|,\|\bbY_T-\bbY\|\lesssim T^{-1}$. Since the square of singular values of $T^{1-d_1}|\mfr_1|^{-T}|\sum_{l=1}^{\tilde{m}_0}a_{d_1,l}\bbu_{d_1-1}(\mfr_l)\bbv_{d_1-1}(\mfr_l)'$ are the same as eigenvalues of $\bbY_T\bbZ_T$, it implies that for $k=1,\cdots,\tilde{m}_0$,
		\begin{align*}
			|\beta_k-\lambda_k(\mcG_T)^{1/2}/\lambda_1(\mcG_T)^{1/2}|\lesssim T^{-1},\quad \mcG_T=\bbY\mcD_T^*\bbZ\mcD_T,
		\end{align*}
		and $\mathtt{k}_0$ is the multiplicities of $\lambda_1(\bbY_T\bbZ_T)$. Note that $\{\bbu_{d_1-1}(\mfr_l):l=1,\cdots,\tilde{m}_0\}$ are linear independent for different $\mfr_l$, so does $\{\bbv_{d_1-1}(\mfr_l):l=1,\cdots,\tilde{m}_0\}$. Thus, $\bbY_T,\bbZ_T$ are positive definite, so do $\bbY,\bbZ$.
		\item For $\tilde{m}_0<k\leq d_{\mathtt{SN}}$, by the Weyl's inequality and \eqref{Eq of low rank}, we have $\beta_k\lesssim T^{-1}$.
		\item Let
		$$\mbT_1:=\prod_{l=1}^{\tilde{m}}\mcT(\mfr_l)^{d_l},\quad\mbU_1:=\mbT_1^{-1},\quad\mbT_2:=\prod_{l=\tilde{m}+1}^{\fd}\mcT(\mfr_l)^{d_l},\quad\mbU_2:=\mbT_2^{-1},$$
		and denote 
		$$\prod_{l=1}^{\tilde{m}}(x-\mfr_l)^{d_l}=x^{d_{\mathtt{SN}}}+\sum_{l=1}^{d_{\mathtt{SN}}}b_lx^{d_{\mathtt{SN}}-l},$$
		where $d_{\mathtt{SN}}=\sum_{l=1}^{\tilde{m}}d_l$ and \(\mfb_0:=1+\sum_{l=1}^{d_{\mathtt{SN}}}b_l^2\) and \(\mfb_k:=b_k+\sum_{l=1}^{d_{\mathtt{SN}}-k}b_lb_{l+k}\) for \(1\leq k\leq d_{\mathtt{SN}}\). Here, we define a \(T\times T\) symmetric $d_{\mathtt{SN}}$-banded toeplitz matrix as follows:
		$$\bbT={\rm Toeplitz}(\mfb_0,\cdots,\mfb_{d_{\mathtt{SN}}},0,\cdots,0),$$
		i.e. $\mfb_0$ on main diagonals and $\mfb_k$ on $k$-th sub/sup-diagonals for $1\leq k\leq d_{\mathtt{SN}}$. Recall that $\mfr_1,\cdots,\mfr_{\tilde{m}}$ are all super nonstationary, by \cite[Lemma 4.1]{gray2006toeplitz}, 
		\begin{align}
			0<\prod_{l=1}^{\tilde{m}}\left|1-|\mfr_l|\right|^{2d_l}\leq\lambda_T(\bbT)\leq\lambda_1(\bbT)\leq\prod_{l=1}^{\tilde{m}}\left|1+|\mfr_l|\right|^{2d_l}<\infty.\label{Eq of bbT singular value bound}
		\end{align}
		On the other hand, note that
		\begin{align*}
			\mbT_1'\mbT_1+\sum_{k=1}^{d_{\mathtt{SN}}}\mfv_k\mfv_k'=\bbT,
		\end{align*}
		where \(\mfv_k=(b_k,\cdots,b_1,0,\cdots,0)'\in\mbR^T\) for \(k=1,\cdots,d_{\mathtt{SN}}\). By \cite[Theorem 1.1]{tyrtyshnikov1996unifying}, we have
		$$\lambda_T(\bbT-\mfv_1\mfv_1')\leq\lambda_T(\bbT)\leq\lambda_{T-1}(\bbT-\mfv_1\mfv_1')\leq\lambda_{T-1}(\bbT)\leq\cdots\leq\lambda_1(\bbT-\mfv_1\mfv_1')\leq\lambda_1(\bbT),$$
		Inductively, we can obtain that
		$$\lambda_T(\mbT_1'\mbT_1)\leq\lambda_{T-1}(\mbT_1'\mbT_1)\leq\cdots\leq\lambda_{T-d_{\mathtt{SN}}+1}(\mbT_1'\mbT_1)\leq\lambda_T(\bbT)\leq\cdots,$$
		i.e.
		$$\lambda_1(\mbU_1\mbU_1')\geq\lambda_2(\mbU_1\mbU_1')\geq\cdots\geq\lambda_{d_{\mathtt{SN}}}(\mbU_1\mbU_1')\geq\lambda_T(\bbT)^{-1}\geq\cdots.$$
		Hence, the number of \(\sigma_t(\mbU_1)\) greater than \(\lambda_T(\bbT)^{-1/2}\) is at most \(d_{\mathtt{SN}}\), i.e.
		\begin{align}
			\sigma_1(\mbU_1)\geq\sigma_2(\mbU_1)\geq\cdots\sigma_{m_{\mathtt{SN}}}(\mbU_1)\geq\lambda_T(\bbT)^{-1/2}\geq\sigma_{d_{\mathtt{SN}}+1}(\mbU_1)\geq\cdots.\label{Eq of singular value inequality 1}
		\end{align}
		For \(\mbT_2=\prod_{l=m_{\mathtt{SN}}+1}^{\fd}\mcT(\mfr_l)^{d_l}\), where \(|\mfr_l|<|\mfr_1|\) for \(l=\tilde{m}+1,\cdots,\fd\). By the proof for the first conclusion, we have
		$$\log(\|\mbU_2\|)+T\log(|\mfr_{\tilde{m}+1}|)\lesssim\log(T).$$
		Since $\mbU_1\mbU_2=\mbU_2\mbU_1=\mbU$, for $t>d_{\mathtt{SN}}$,
		\begin{align*}
			&\log(\sigma_t(\mbU_1\mbU_2))+T\log(|\mfr_{\tilde{m}+1}|)\lesssim\log(T).
		\end{align*}
		Consequently, combining with $\log(\|\mbU_1\|)+T\log(|\mfr_1|)\lesssim\log(T)$, we obtain that for $t>d_{\mathtt{SN}}$
		\begin{align*}
			\log(\beta_t)\lesssim-T\log(|\mfr_1|/|\mfr_{\tilde{m}+1}|).
		\end{align*}
		\item By the definition of $\tilde{m}$, $|\mfr_l|<|\mfr_1|$ for $l>\tilde{m}$, combining with Lemma \ref{Lem of spectral gap SN}, we have
		\begin{align*}
			\log\left(\frac{1}{\|\mbU\|}\left\|\mbU-\sum_{l=1}^{\tilde{m}}\sum_{k_l=1}^{d_l}a_{k_l,l}\mcT(\mfr_l)^{-k_l}\right\|\right)\lesssim-T\log(|\mfr_1|/|\mfr_{\tilde{m}+1}|),
		\end{align*}
		and
		\begin{align*}
			\left\|\sum_{l=1}^{\tilde{m}}\sum_{k_l=1}^{d_l}a_{k_l,l}\mcT(\mfr_l)^{-k_l}-\sum_{l=1}^{\tilde{m}}\sum_{k_l=1}^{d_l}a_{k_l,l}\sum_{r_l=0}^{k_l-1}\bbu_{r_l}(\mfr_l)\bbv_{r_l}(\mfr_l)^*\right\|\lesssim1.
		\end{align*}
		Let $\tilde{\mfv}_k,\tilde{\mfw}_k$ be left/right singular vectors of $\tilde{\bbC}:=\sum_{l=1}^{\tilde{m}}\sum_{k_l=1}^{d_l}a_{k_l,l}\sum_{r_l=0}^{k_l-1}\bbu_{r_l}(\mfr_l)\bbv_{r_l}(\mfr_l)^*$ corresponding to its $k$-th largest singular value, then $\tilde{\mfv}_k$ is a linear combination of $V:=\{\bbu_r(\mfr_l)/\|\bbu_r(\mfr_l)\|_2:l=1,\cdots,\tilde{m},r=0,1,\cdots,d_l-1\}$. By \eqref{Eq of rank-one approximation}, we have
		\begin{align*}
			\log(|u_{t,r}(\mfr_l)|/\|\bbu_r(\mfr_l)\|_2)\lesssim-t\log(|\mfr_1|).
		\end{align*}
		Since $\tilde{m}$ and $d_1,\cdots,d_{\tilde{m}}$ are finite, it implies that
		\begin{align*}
			\log(|\tilde{\mfv}_{t,k}|)\lesssim-t\log(|\mfr_1|).
		\end{align*}
		By the Davis-Kahan $\sin(\Theta)$ theorem, let $P_V$ be a projection operator on the subspace $V$, for $k=1,\cdots,\tilde{m}_0$
		\begin{align*}
			\log(\|(\bbI_T-P_V)\mfv_k\|_2)\lesssim\log\left(\frac{\|\mbU-\tilde{\bbC}\|}{\|\mbU\|}\right)\lesssim-T\log(|\mfr_1|),
		\end{align*}
		which concludes that $\log(|\mfv_{t,k}|)\lesssim-t\log(|\mfr_1|)$. Similarly, by previous arguments, we can derive that $\log(|\mfw_{t,k}|)\lesssim-(T-t)\log(|\mfr_1|)$.
	\end{enumerate}	
\end{proof}

\subsection{Proof of Theorems \texorpdfstring{\ref{Thm of super nonstationary 1}}{3.1} and  \texorpdfstring{\ref{Thm of super nonstationary 2}}{3.2}}
\begin{proof}[Proof of Theorem \ref{Thm of super nonstationary 1}]
	Recall that in \eqref{Eq of Mi st}
	\begin{align*}
		\hR_{\mathtt{SN}}&=\bbX'\diag(\bbX\bbX')^{-1}\bbX=\mbU'\bbve'\bbGa'\diag(\bbGa\bbve\mbU\mbU'\bbve'\bbGa')^{-1}\bbGa\bbve\mbU\\
		&=\sum_{k,l=1}^T\mfw_k\mfw_l'\sum_{i=1}^n\frac{\tsigma_k\tsigma_l(\Gamma_i\bbve\mfv_k)(\Gamma_i\bbve\mfv_l)}{\sum_{t=1}^T\tsigma_t^2(\Gamma_i\bbve\mfv_t)^2}=\sum_{k,l=1}^T\mfw_k\mfw_l'\sum_{i=1}^nM_{i;k,l},
	\end{align*}
	since $\sum_{t=1}^T\hla_{t,\mathtt{SN}}=n=\sum_{k=1}^T\sum_{i=1}^nM_{i;k,k}$, then
	$$\sum_{t=d_{\mathtt{SN}}+1}^T\hla_{t,\mathtt{SN}}\leq\sum_{k=d_{\mathtt{SN}}+1}^T\sum_{i=1}^nM_{i;k,k}.$$
	Note that $M_{i;k,k}\leq1$, it gives that for fixed $p\in(0,1/2)$,
	\begin{align*}
		M_{i;k,k}\leq M_{i;k,k}^p\leq\beta_k^{2p}|\Gamma_i\bbve\mfv_k|^{2p}|\Gamma_i\bbve\mfv_1|^{-2p}.
	\end{align*}
	By Assumption \ref{Ap of main}, Lemmas \ref{Thm of cdf order} and \ref{Thm of bounded density}, $\Gamma_i\bbve\mfv_1$ has bounded density, when $k>d_{\mathtt{SN}}$,
	$$\mbE[M_{i;k,k}]\lesssim\beta_k^{2p}\lesssim |\mfr_1|^{-2pCT}$$
	for some constant $C>0$, then 
	\begin{align*}
		&\sum_{k=d_{\mathtt{SN}}+1}^T\sum_{i=1}^n\mbE[M_{i;k,k}]\lesssim nT|\mfr_1|^{-2pCT},
	\end{align*}
	which concludes that $\sum_{t=d_{\mathtt{SN}}+1}^T\hla_{t,\mathtt{SN}}\overset{\mbP}{\longrightarrow}0$. Moreover, for $\tilde{m}_0<t\leq d_{\mathtt{SN}}$ and any $\epsilon>0$, since
	\begin{align*}
		\frac{\hla_{t,\mathtt{SN}}}{n^{\epsilon}}\leq\frac{1}{n^{\epsilon}}\sum_{k=\tilde{m}_0+1}^{d_{\mathtt{SN}}}\sum_{i=1}^nM_{i;k,k}+\frac{1}{n^{\epsilon}}\sum_{k=d_{\mathtt{SN}}+1}^T\sum_{i=1}^nM_{i;k,k}\leq\sum_{k=\tilde{m}_0+1}^{d_{\mathtt{SN}}}\frac{1}{n}\sum_{i=1}^nM_{i;k,k}+\mro_{\mbP}(n^{-\epsilon}).
	\end{align*}
	Here, let $p<(1-\epsilon)/2$, by Proposition \ref{Pro of spectral gap SN} and Assumption \ref{Ap of high-dimensionality}, for $\tilde{m}_0<k\leq d_{\mathtt{SN}}$
	$$\lim_{n,T\to\infty}\frac{1}{n^{\epsilon}}\sum_{i=1}^n\mbE[M_{i;k,k}]\lesssim \lim_{T\to\infty}T^{1-\epsilon}\beta_k^{2p}\lesssim \lim_{T\to\infty}T^{1-2p-\epsilon}=0,$$
	which implies that $\sum_{k=\tilde{m}_0+1}^{d_{\mathtt{SN}}}n^{-\epsilon}\sum_{i=1}^nM_{i;k,k}\overset{\mbP}{\longrightarrow}0$, so $n^{-\epsilon}\hla_{t,\mathtt{SN}}\overset{\mbP}{\longrightarrow}0$.
\end{proof}

\begin{proof}[Proof of Theorem \ref{Thm of super nonstationary 2}]
	Since we assume that $\varepsilon_t\overset{i.i.d.}{\sim}\mcN(0,1)$, then $\{\Gamma_i\bbve\mfv_t:t=1,\cdots,T\}$ are i.i.d. Gaussian, i.e. $\vec{\fz}_t\overset{d}{\sim}\bbGa\bbve\mfv_t$, where $\vec{\fz}_t=(\fz_{1,t},\cdots,\fz_{n,t})'$ are defined in \eqref{Eq of hat M}. Without loss of generality, we assume that
	\begin{align*}
		M_{i;k,l}=\frac{\beta_k\beta_l\fz_{i,k}\fz_{i,l}}{\sum_{t=1}^T\beta_t^2\fz_{i,t}^2},
	\end{align*}
	let's first show that $n^{-1}\Var\big(\sum_{i=1}^nM_{i;k,l}\big)\lesssim1$ for $1\leq k,l\leq d_{\mathtt{SN}}$. Here, we define
	\begin{align*}
		\tilde{M}_{i;k,l}:=\frac{\beta_k\beta_l\fz_{i,k}\fz_{i,l}}{\sum_{t=1}^{d_{\mathtt{SN}}}\beta_t^2\fz_{i,t}^2},
	\end{align*}
	then for any $p\in(0,1/4)$,
	\begin{align*}
		\mbE[|\tilde{M}_{i;k,l}-M_{i;k,l}|^2]\leq\mbE\left[\left(\frac{\sum_{t=d_{\mathtt{SN}}+1}^T\beta_t^2\fz_{i,t}^2}{\sum_{t=1}^T\beta_t^2\fz_{i,t}^2}\right)^{2p}\right]\leq\sum_{t=d_{\mathtt{SN}+1}}^T\beta_t^{4p}\mbE[|\fz_{i,t}/\fz_{i,1}|^{4p}]\lesssim T|\mfr_1|^{-4p CT}.
	\end{align*}
	Thus, it suffices to show that $n^{-1}\Var\big(\sum_{i=1}^n\tilde{M}_{i;k,l}\big)\lesssim1$ for $1\leq k,l\leq d_{\mathtt{SN}}$. Consider the Hermite polynomials $\{H_j(x)\}_{j=0}^\infty$ defined by 
	$$H_j(x):=\frac{(-1)^j}{\sqrt{j!}}e^{x^2/2}\frac{d^je^{-x^2/2}}{dx^j},$$
	which form a complete orthonormal basis of $L^2(\mbR)$. Since we can view $\tilde{M}_{i;k,l}^{\circ}$ as a bounded function of $(\fz_{i,1},\cdots,\fz_{t,d_{\mathtt{SN}}})$, due to $H_0(x)=1$, then
	\begin{align*}
		\tilde{M}_{i;k,l}^{\circ}:=\sum_{m=1}^{\infty}\sum_{j_1+\cdots j_{d_{\mathtt{SN}}}=m}c_{j_1,\cdots,j_{d_{\mathtt{SN}}}}\prod_{l=1}^{d_{\mathtt{SN}}} H_{j_l}(\fz_{i,l}),\quad c_{j_1,\cdots,j_{d_{\mathtt{SN}}}}=\mbE\left[\tilde{M}_{i;k,l}^{\circ}\prod_{l=1}^{d_{\mathtt{SN}}}H_{j_l}(\fz_{i,l})\right].
	\end{align*}
	Since $H_j(-x)=(-1)^jH_j(x)$ and $\tilde{M}_{i;k,l}^{\circ}(\fz_{i,1},\cdots,\fz_{t,d_{\mathtt{SN}}})=\tilde{M}_{i;k,l}^{\circ}(-\fz_{i,1},\cdots,-\fz_{t,d_{\mathtt{SN}}})$, it implies that
	\begin{align*}
		&c_{j_1,\cdots,j_{d_{\mathtt{SN}}}}=\mbE\left[\tilde{M}_{i;k,l}^{\circ}(\fz_{i,1},\cdots,\fz_{t,d_{\mathtt{SN}}})\prod_{l=1}^{d_{\mathtt{SN}}}H_{j_l}(\fz_{i,l})\right]\\
		&=\mbE\left[\tilde{M}_{i;k,l}^{\circ}(-\fz_{i,1},\cdots,-\fz_{t,d_{\mathtt{SN}}})\prod_{l=1}^{d_{\mathtt{SN}}}H_{j_l}(-\fz_{i,l})\right]=(-1)^{j_1+\cdots j_{d_{\mathtt{SN}}}}c_{j_1,\cdots,j_{d_{\mathtt{SN}}}},
	\end{align*}
	which implies that
	\begin{align*}
		\tilde{M}_{i;k,l}^{\circ}:=\sum_{m=1}^{\infty}\sum_{j_1+\cdots j_{d_{\mathtt{SN}}}=2m}c_{j_1,\cdots,j_{d_{\mathtt{SN}}}}\prod_{l=1}^{d_{\mathtt{SN}}} H_{j_l}(\fz_{i,l}).
	\end{align*}
	Therefore,
	\begin{align*}
		\Cov(\tilde{M}_{i_1;k,l},\tilde{M}_{i_2;k,l})=\sum_{m_1,m_2=1}^{\infty}\sum_{\substack{i_1+\cdots i_{d_{\mathtt{SN}}}=2m_1\\j_1+\cdots j_{d_{\mathtt{SN}}}=2m_2}}c_{i_1,\cdots,i_{d_{\mathtt{SN}}}}c_{j_1,\cdots,j_{d_{\mathtt{SN}}}}\prod_{l=1}^{d_{\mathtt{SN}}} \mbE[H_{i_l}(\fz_{i_1,l})H_{j_l}(\fz_{i_2,l})],
	\end{align*}
	where we use the fact that $\mbE[H_j(\fz_{i,k})]=0$ for all $j\in\mbN$. By the orthonormal properties of $H_j(\cdot)$, we have $\mbE[H_{i_l}(\fz_{i_1,l})H_{j_l}(\fz_{i_2,l})]=\Corr(\fz_{i_1,l},\fz_{i_2,l})^{i_l}\delta_{i_l,j_l}$, then
	\begin{align*}
		\Cov(\tilde{M}_{i_1;k,l},\tilde{M}_{i_2;k,l})=\sum_{m=1}^{\infty}\Cov(\fz_{i_1,l},\fz_{i_2,l})^{2m}\sum_{j_1+\cdots j_{d_{\mathtt{SN}}}=2m}c_{j_1,\cdots,j_{d_{\mathtt{SN}}}}^2\lesssim\Corr(\fz_{i_1,l},\fz_{i_2,l})^2,
	\end{align*}
	by Assumption \ref{Ap of bbGa}, we can conclude that $n^{-1}\Var\big(\sum_{i=1}^n\tilde{M}_{i;k,l}\big)\lesssim1$ for $1\leq k,l\leq d_{\mathtt{SN}}$.

	Finally, by Theorem \ref{Thm of super nonstationary 1}, it suffices to establish the asymptotic behaviors for the eigenvalues of
	$$\hR_{\mathtt{SN}}:=\hR_{\mathtt{SN}}(d_{\mathtt{SN}}):=\sum_{k,l=1}^{d_{\mathtt{SN}}}\mfw_k\mfw_l'\sum_{i=1}^nM_{i;k,l}.$$
	Let $\hF_{l,\mathtt{SN}}=\sum_{r=1}^{d_{\mathtt{SN}}}\alpha_{l,r}\mfw_r$ be the normalized eigenvector of $\hla_{l,\mathtt{SN}}(\hR_{\mathtt{SN}})$ for $l=1,\cdots,d_{\mathtt{SN}}$, then
	\begin{align*}
		\hla_{l,\mathtt{SN}}=\sum_{r_1,r_2=1}^{\tilde{m}_0}\alpha_{l,r_1}\alpha_{l,r_2}\sum_{i=1}^nM_{i;r_1,r_2}.
	\end{align*}
	Based on $n^{-1}\Var\big(\sum_{i=1}^nM_{i;k,l}\big)\lesssim1$ and Proposition \ref{Pro of spectral gap SN}, we can use the same arguments as those in Proposition \ref{Pro of eigenvectors two} to derive that $\sqrt{n}\sum_{r=\mathtt{k_0}+1}^{\tilde{m}_0}\mbE[\alpha_{r,l}^2]\to0$. Combining with the arguments in \S\ref{sec of weak convergence}, we can conclude Theorem \ref{Thm of super nonstationary 2}.
\end{proof}

\appendix
\section{Proof of Auxiliary lemmas in \texorpdfstring{\S\ref{sec of spectral gap}}{Section 2.1}}\label{App of spectral gap}
\begin{proof}[Proof of Lemma \ref{Lem of operator}]
	\begin{enumerate}
		\item Let's first show that the operator $\mcJ_T:=T^{-1}E_T\mcT(1)^{-1}E_T^*$ on $L^2([0,1])$ converges to $\bbJ$. By direct calculations, for any $f\in L^2([0,1])$, we have
		\begin{align*}
			&\mcJ_T(f)(s)=\int_{(t-1)/T}^1f(z)dz,\quad s\in[(t-1)/T,t/T),
		\end{align*}
		then
		\begin{align*}
			&\|\mcJ_T(f)-\bbJ(f)\|_{L^2}^2=\sum_{t=1}^T\int_{(t-1)/T}^{t/T}|\mcJ_T(f)(s)-\bbJ(f)(s)|^2ds\\
			&=\sum_{t=1}^T\int_{(t-1)/T}^{t/T}\left|\int_{(t-1)/T}^sf(z)dz\right|^2ds\leq\frac{1}{T^2}\sum_{t=1}^T\int_{(t-1)/T}^{t/T}|f(z)|^2dz=\frac{1}{T^2}\| f\|_{L^2}^2,
		\end{align*}
		which implies $\|\mcJ_T-\bbJ\|_{op}\leq T^{-1}$, where $\|\cdot\|_{op}$ is the operator norm on $L^2([0,1])$. Based on previous results, we can further derive that $\mcJ_T^{d_1}\longrightarrow\bbJ^{d_1}$ in operator norm. Actually, note that
		\begin{align*}
			&\mcJ_T^{d_1}-\bbJ^{d_1}=\sum_{l=0}^{d_1-1}\mcJ_T^l(\mcJ_T-\bbJ)\bbJ^{d_1-1-l},
		\end{align*}
		which implies that
		\begin{align*}
			\|\mcJ_T^{d_1}-\bbJ^{d_1}\|_{op}\leq\|\mcJ_T-\bbJ\|_{op}\sum_{l=0}^{d_1-1}\|\mcJ_T\|_{op}^l\|\bbJ\|_{op}^{d_1-1-l}\lesssim\frac{1}{T},
		\end{align*}
		where we use the fact that $\|\bbJ\|_{op}\lesssim1$ and we conclude the first claim.
		\item For the second claim, let's show that $\bbJ^d(\bbJ^d)^*$ is an oscillatory operator (see Defintion 1 in \cite[Chapter \RomanNumeralCaps{4}.2]{gantmacher2002oscillation}). Note that $\bbJ(f)(s)=\int_s^1f(z)dz$ and
		\begin{small}
			\begin{align*}
				(\bbJ^d)(f)(s)=\int_s^1\frac{(z-s)^{d-1}}{(d-1)!}f(z)dz:=\int_0^1K_d(s,z)f(z)dz,\quad K_d(s,z)=\left\{\begin{array}{cc}
					\frac{(z-s)^{d-1}}{(d-1)!}, & z\geq s, \\
					0, & z<s.
				\end{array}\right.
			\end{align*}
		\end{small}\noindent
		By the property of the adjoint operator, we have $(\bbJ^d)^*(f)(s)=\int_0^1K_d(z,s)f(z)dz$, so that $\bbJ^d(\bbJ^d)^*(f)(s)=\int_0^1G_d(s,z)f(z)dz$, where
		\begin{align*}
			G_d(s,z):=\int_0^1K_d(s,r)K_d(z,r)dr=\int_{\max\{s,z\}}^1\frac{(r-s)^{d-1}(r-z)^{d-1}}{\{(d-1)!\}^2}dr.
		\end{align*}
		To conclude that $\bbJ^d(\bbJ^d)^*$ is oscillatory, it suffices to show that for any $0<s_1<\cdots<s_N<1$ and $0< z_1<\cdots<z_N<1$, determinants of the matrix $[G_d(s_i,z_j)]_{N\times N}$ and $[G_d(s_i,s_i)]_{N\times N}$ are nonnegative and postive, respectively. By the Cauchy-Binet matrix formula, we have
		\begin{small}
			\begin{align*}
				\det([G_d(s_i,z_j)]_{N\times N})=\frac{1}{N!}\int_0^1\cdots\int_0^1\det([K_d(s_i,r_k)]_{N\times N})\det([K_d(z_j,r_k)]_{N\times N})dr_1\cdots dr_N,
			\end{align*}
		\end{small}\noindent
		so it suffices to show that $\det([K_d(s_i,r_k)]_{N\times N})\geq0$ and $\det([K_d(z_j,r_k)]_{N\times N})\geq0$. For simplicity, let's first consider the case when $d=1$. Since $K_1(s,r)=1_{r\geq s}$, suppose there exists an $i$ such that $K_1(s_i,r_i)=0$, it gives that $s_i<r_i$ and $K_1(s_j,r_k)=0$ for $j\leq i$ and $k\geq i$, which implies that $\det([K_1(s_i,r_k)]_{N\times N})=0$. On the other hand, if all $K_1(s_i,r_i)=1$, that is $s_i\geq r_i$ for all $i=1,\cdots,N$. Then, we perform the row operation of subtracting row $i+1$ from row $i$ (starting from the bottom), the $(i,k)$-th entry of this new derived matrix will be 1 only if $r_k \leq s_i < r_{k+1}$, otherwise be 0 due to $K_1(s_i,r_{k+1})\geq K_1(s_i,r_k)$. Thus, this new derived matrix is upper triangular with 1's on the main diagonal, so $\det([K_1(s_i,r_k)]_{N\times N})=1$.
		
		For general $d>1$, note that $K_d(s,z)=\int_0^1K_{d-1}(s,r)K_1(r,z)dr$, by the Cauchy-Binet matrix formula, we have
		\begin{small}
			\begin{align*}
				\det([K_d(s_i,z_j)]_{N\times N})=\frac{1}{N!}\int_0^1\cdots\int_0^1\det([K_{d-1}(s_i,r_k)]_{N\times N})\det([K_1(z_j,r_k)]_{N\times N})dr_1\cdots dr_N,
			\end{align*}
		\end{small}\noindent
		then $\det([K_d(s_i,z_j)]_{N\times N})\geq0$ by the mathematical induction, which further yields that $\det([G_d(s_i,z_j)]_{N\times N})$ is nonnegative. Particularly, we have $\det([G_d(s_i,s_i)]_{N\times N})>0$ due to $\det([K_1(s_i,r_k)]_{N\times N})>0$ for all $(r_1, \cdots, r_N)$ in the open box $(s_1, s_2) \times (s_2, s_3) \times \cdots \times (s_N, 1)$, then $G_d$ is oscillatory. Now, by Gantmacher-Krein theorem \cite[Chapter \RomanNumeralCaps{4}.3, Theorem 1]{gantmacher2002oscillation}, all eigenvalues of $\bbJ^d(\bbJ^d)^*$ are positive and strictly simple, which concludes the second claim.
	\end{enumerate}
\end{proof}

\begin{proof}[Proof of Lemma \ref{Lem of singular structure}]
	\begin{enumerate}
		\item Note that \(\mbT=\prod_{l=1}^d\mcT(\mfr_l)\) in \eqref{Eq of mbU} is an upper toeplitz matrix, so \(\mbU=\mbT^{-1}\) is also an upper toeplitz matrix such that
		$$\mbU=\left(\begin{array}{ccccc}
			1&b_1&\cdots&b_{T-2}&b_{T-1}\\
			0&1&b_1&\cdots&b_{T-2}\\
			&\ddots&\ddots&\ddots&\vdots\\
			0&\cdots&0&1&b_1\\
			0&&\cdots&0&1
		\end{array}\right),$$
		where \(b_k\) is the \(k\)-th sup-diagonal entry of \(\mbU\) for \(1\leq k\leq T-1\) determining by
		\begin{align}
			\sum_{k=0}^{\infty}b_kx^k=\prod_{l=1}^{\fd}\left(\sum_{k_l=0}^{\infty}\binom{k_l+d_l-1}{d_l-1}(\mfr_lx)^{k_l}\right).\label{Eq of k-th sup diagonal of mbU}
		\end{align}
		Hence, we have \(|b_k|\lesssim k^{d_1-1}\) and
		\begin{align}
			&\Vert\mbU\Vert_F\lesssim\left(\sum_{k=1}^T(T+1-k)k^{2(d_1-1)}\right)^{1/2}\lesssim T^{d_1}.\label{Eq of mbU frobenius norm}
		\end{align}
		On the other hand, since all \(|\mfr_l|\leq1\), then \(\Vert\mcT(\mfr_l)\Vert\leq2\) and \(\sigma_T(\mcT(\mfr_l)^{-1})\geq1/2\). Thus,
		$$\Vert\mbU\Vert\geq\sigma_1(\mcT(\mfr_1)^{-d_1})\prod_{l=2}^{\fd}\sigma_T(\mcT(\mfr_l)^{-d_l})>2^{d_1-d}\sigma_1(\mcT(\mfr_1)^{-d_1}).$$
		Moreover, by direct calculations, \(\mcT(\mfr_1)^{-d_1}\) is a \(T\times T\) upper toeplitz matrix with main diagonals are 1 and \(k\)-th sup-diagonals are \(\binom{k+d_1-1}{d_1-1}\mfr_1^k\). Let \(\vec{a}=T^{-1/2}(\mfr_1^T,\cdots,\mfr_1,1)'\) be a \(T\)-dimensional unit vector, then
		\begin{align*}
			&\Vert\mcT(\mfr_1)^{-d_1}\vec{a}\Vert_2^2=T^{-1}\sum_{t=1}^T\Bigg|\mfr_1^{T-t}\sum_{k=0}^{T-t}\binom{k+d_1-1}{d_1-1}\Bigg|^2=T^{-1}\sum_{t=1}^T\left(\sum_{k=0}^T\binom{k+d_1-1}{d_1-1}\right)^2\\
			&>T^{-1}\sum_{t=1}^T\left(\sum_{k=0}^T(k/d_1)^{d_1-1}\right)^2\gtrsim T^{-1}\sum_{t=1}^Tt^{2d_1}\gtrsim T^{2d_1},
		\end{align*}
		where we use the fact that \(|\mfr_1|=1\), so we conclude that \(\sigma_1(\mcT(\mfr_1)^{-d_1})=\Vert\mcT(\mfr_1)^{-d_1}\Vert\gtrsim T^{d_1}\) and \(\Vert\mbU\Vert>2^{d_1-d}\sigma_1(\mcT(\mfr_1)^{-d_1})\gtrsim T^{d_1}\). Recall that $\beta_t=\sigma_t(\mbU)/\sigma_1(\mbU)$ in \eqref{Eq of beta k}, then \eqref{Eq of mbU frobenius norm} and $\Vert\mbU\Vert\gtrsim T^{d_1}$ together imply that $\sum_{t=1}^T\beta_t^2=\Vert\mbU\Vert_F^2/\Vert\mbU\Vert^2\asymp1$.
		\item Let's construct a \((T+d)\times(T+d)\) circulant matrix as follows:
		$$\mfC(T+d):=\left(\begin{array}{ccccccc}
			1&a_1&\cdots&a_d&0&\cdots&0\\
			0&1&a_1&\cdots&a_d&\cdots&\vdots\\
			\vdots&\vdots&\ddots&\ddots&&\ddots&\\
			&&&1&a_1&\cdots&a_d\\
			a_d&0&\cdots&0&1&\cdots&a_{d-1}\\
			\vdots&\ddots&&&\ddots&\ddots&\vdots\\
			a_1&\cdots&a_d&0&\cdots&0&1
		\end{array}\right).$$
		It is easy to see that \(\mbT\) can be derived by deleting the first \(d\) rows and columns of \(\mfC(T+d)\). According to \cite{karner2003spectral}, the singular values of \(\mfC(T+d)\) are
		\begin{align*}
			s_k:&=\left|1+\sum_{l=1}^da_le^{-2\pi{\rm i}(k-1)l/(T+d)}\right|=\prod_{l=1}^{\fd}\big|1-\mfr_l e^{-2\pi{\rm i}(k-1)/(T+d)}\big|^{d_l}\\
			&=\prod_{l=1}^m\big|1-e^{2\pi{\rm i}[\mathtt{r}_l-(k-1)/(T+d)]}\big|^{d_l}\times\prod_{l=m+1}^{\fd}\big|1-\mfr_l e^{-2\pi{\rm i}(k-1)/(T+d)}\big|^{d_l},
		\end{align*}
		where \(k=1,\cdots,T+d\) and \(\mfr_l=e^{2\pi{\rm i}\mathtt{p}_l}\) as in \eqref{Eq of nonstationary roots} for \(1\leq l\leq m\). Here, we sort \(s_k\) as follows:
		\begin{align*}
			\begin{array}{l}
				\sigma_1(\mfC(T+d))=\max\{s_1,\cdots,s_{T+d}\}, \\
				\sigma_l(\mfC(T+d)):=\max\{s_1,\cdots,s_{T+d}\}\backslash\{\sigma_1(\mfC(T+d)),\cdots,\sigma_{l-1}(\mfC(T+d))\},\quad l\geq1.
			\end{array}
		\end{align*}
		Denote
		\begin{align}
			k_l:=\arg\min\{|\mathtt{p}_l-(k-1)/(T+d)|:1\leq k\leq T+d\}.\label{Eq of kl}
		\end{align}
		Without loss of generality, we assume that $\mathtt{p}_l\neq(k_l-1)/(T+d)$ for $l=1,\cdots,m$; otherwise, we can use a parameter $\delta\in(0,1)$ to rotate $\mfr_l$ as in \eqref{Eq of mfu}. Then for any small $\epsilon>0$ and $|k-k_l|<\epsilon T$, we have
		\begin{align*}
			s_k\asymp\big|1-e^{2\pi{\rm i}[\mathtt{p}_l-(k-1)/(T+d)]}\big|^{d_l}\asymp|\mathtt{p}_l-(k-1)/(T+d)|^{d_l}\asymp|(k-k_l)/(T+d)|^{d_l},
		\end{align*}
		where we use the fact that \(\left|e^{2\pi{\rm i}x}-1\right|\asymp|x|\) for \(|x|<\epsilon\), so
		\begin{align*}
			T^{-d_1}s_k^{-1}\lesssim T^{d_l-d_1}(k-k_l)^{-d_l}.
		\end{align*}
		On the other hand, for $k\in\{1,\cdots,T+d\}$ such that $\min_{1\leq l\leq m}|k-k_l|>\epsilon T$, we have $T^{-d_1}s_k^{-1}\lesssim\epsilon^{-d_l} T^{-d_1}$. Thus, if we sort $\{T^{-d_1}s_k^{-1}:k=1,\cdots,T+d\}$ in an ascending order, i.e. 
		$$T^{-d_1}s_{\pi(1)}^{-1}\leq\cdots\leq T^{-d_1}s_{\pi(T+d)}^{-1},$$
		we have $T^{-d_1}s_{\pi(t)}^{-1}=T^{-d_1}\sigma_{T+d-t}(\mfC(T+d))^{-1}\lesssim t^{-d_1}$. Finally, by Cauchy's interlacing theorem, we know that
		$$\sigma_{t+2d}(\mfC(T+d))\leq\sigma_t(\mbT)\leq\sigma_t(\mfC(T+d)),$$
		where \(1\leq t\leq T\). Since $d$ is fixed, we have
		\begin{align*}
			\beta_t\leq\frac{1}{{\sigma_{T+1-t}(\mfC(T+d))\sigma_1(\mbU)}}\lesssim T^{-d_1}s_{\pi(d+t)}^{-1}\lesssim t^{-d_1}.
		\end{align*}
	\end{enumerate}
\end{proof}

\section{Proof of Auxiliary lemmas in \texorpdfstring{\S\ref{ssec of Basic properties for noises}}{Section 2.2.1}}
\begin{proof}[Proof of Lemma \ref{Thm of cdf order}]
	For any $x>0$, since
	\begin{align}
		&\mathbb{P}\big(\Vert\vec{z}\Vert_2^2\leq x\big)=\int_{z_1^2+\cdots+z_{2R+1}^2\leq x}p_{\vec{z}}(z_1,\cdots,z_{2R+1}){\rm d}z_1\cdots {\rm d}z_{2R+1}\notag\\
		&=\int_{\mathcal{D}_x} r^{2R}\sin^{2R-1}\varphi_1\sin^{2R-2}\varphi_2\cdots\sin\varphi_{2R}p_{\vec{z}}(z_1,\cdots,z_{2R+1}){\rm d}r{\rm d}\varphi_1\cdots{\rm d}\varphi_{2R},\notag
	\end{align}
	where
	$$\left\{\begin{array}{l}
		z_1:=r\cos\varphi_1\\
		z_2:=r\sin\varphi_1\cos\varphi_2\\
		\vdots\\
		z_{2R-1}:=r\sin\varphi_1\sin\varphi_2\cdots\sin\varphi_{2R-2}\cos\varphi_{2R-1}\\
		z_{2R}:=r\sin\varphi_1\sin\varphi_2\cdots\sin\varphi_{2R-2}\sin\varphi_{2R-1}
	\end{array}\right.$$
	and \(\mathcal{D}_x:=\{r\in[0,\sqrt{x}];\varphi_1,\cdots,\varphi_{2R-2}\in[0,\pi];\varphi_{2R-1}\in[0,2\pi]\}\). Hence,
	\begin{align}
		\mathbb{P}\big(\Vert\vec{z}\Vert_2\leq x\big)&\leq\Vert p_{\vec{z}}\Vert_{\infty}\int_{\mathcal{D}_x} r^{2R}\big|\sin^{2R-1}\varphi_1\sin^{2R-2}\varphi_2\cdots\sin\varphi_{2R}\big|{\rm d}r{\rm d}\varphi_1\cdots{\rm d}\varphi_{2R}\notag\\
		&\leq\Vert p_{\vec{z}}\Vert_{\infty}(2\pi)^{2R}\int_0^{\sqrt{x}}r^{2R}{\rm d}r\leq C_{R,\Vert p_{\vec{z}}\Vert_{\infty}},\notag
	\end{align}
	and
	\begin{align*}
		&\mbE\left[\frac{1}{\big(\sum_{i=1}^{2R+1}z_i^2\big)^R}\right]=R\int_0^{\infty}r^{R-1}\mbP(\Vert\vec{z}\Vert_2^2\leq r^{-1})dr\leq R+R\int_1^{\infty}r^{R-1}\mbP(\Vert\vec{z}\Vert_2^2\leq r^{-1})dr\\
		&\leq R+C_{R,\Vert p_{\vec{z}}\Vert_{\infty}}\int_1^{\infty}r^{-3/2}dr\leq C_{R,\Vert p_{\vec{z}}\Vert_{\infty}},
	\end{align*}
	which completes our proof.
\end{proof}

\begin{proof}[Proof of Lemma \ref{Thm of bounded density}]
	Let $A_{\cdot j}$ be the $j$-th column of $\bbA$, by \cite[Corollary 2]{bobkov2014bounds}, the density function of $\bba'A_{\cdot j}$ is bounded by $\Vert\bba\Vert_2^{-1}e^{1/2}M$. Hence, $\bbA'\bba$ is an $m$-dimensional random vector whose entries are independent with uniformly bounded density. Combining with \cite[Theorem 1.1]{rudelson2015small}, since $\bbb_r$ are mutually orthogonal unit vectors for $r=1,\cdots,R$, the the density function of $(\bba'\bbA\bbb_1,\cdots,\bba'\bbA\bbb_R)'$ is bounded by $C_R\Vert\bba\Vert_2^{-R}e^{R/2}M^R$ for a constant $C_R>0$. 
\end{proof}

\begin{proof}[Proof of Lemma \ref{Lem of finite 8th moment}]
	Let $A_{i\cdot}$ be the $i$-th row of $\bbA$, then $A_{i\cdot}\bbb$ are independent and $\mbE[A_{i\cdot}\bbb]=0,\mbE[|A_{i\cdot}\bbb|^2]=1$. By \cite[Lemma B.26]{bai2010spectral} and Hölder's inequality, $\mbE[||A_{i\cdot}\bbb|^2-1|^k]\leq C_k\kappa_{2k}$ and
	$$\mbE[|A_{i\cdot}\bbb|^{2k}]\leq2^{k-1}\mbE[||A_{i\cdot}\bbb|^2-1|^k]+2^{k-1}\leq C_k\kappa_{2k}.$$
	By \cite[Lemma B.26]{bai2010spectral} again, we have
	\begin{align*}
		\mbE[|\bbb'\bbA'\bba\bba'\bbA\bbb-\Vert\bba\Vert_2^2|^k]\leq C_k\kappa_{2k}\Vert\bba\Vert_2^{2k},
	\end{align*}
	which implies that $\mbE[|\bba'\bbA\bbb|^{2k}]\leq C_k\kappa_{2k}\Vert\bba\Vert_2^{2k}$.
\end{proof}

\section{Proof of Auxiliary lemmas in \texorpdfstring{\S\ref{ssec of discrete Fourier transform}}{Section 2.2.3}}
\begin{proof}[Proof of Lemma \ref{Lem of singular structure 2}]
	Recall that $\mbU$ is a $T\times T$ upper toeplitz matrix such that its main diagonals are 1 and $k$-th sup-diagonals are $b_k$ for $k=1,\cdots,T-1$, where $b_k$ are determined by \eqref{Eq of k-th sup diagonal of mbU}.
	\begin{enumerate}
		\item Denote $t_{l*}:=\arg\min\{|\mathtt{p}_l-(t-1+\delta)/T|:t=1,\cdots,T\}$. By the proof of the term 1 in Lemma \ref{Lem of singular structure}, we have 
		\begin{align*}
			\Vert\mbU\mfu_{t_{1*}}\Vert_2\gtrsim\Vert\mcT(\mfr_1)^{-d_1}\mfu_{t_{1*}}\Vert_2.
		\end{align*}
		Without loss of generality, we assume that $0<\mathtt{p}_1-(t_{1*}-1)/T<T^{-1}$. Since the $k$-th sup-diagonals of $\mcT(\mfr_1)^{-d_1}$ are $\binom{k+d_1-1}{d_1-1}\mfr_1^k$, then for $1\leq s\leq T/8$
		\begin{align*}
			&\frac{1}{T}\left|\sum_{l=0}^{s-1}\binom{l+d_1-1}{d_1-1}\Re(e^{2\pi{\rm i}(\mathtt{r}_l-(t_{1*}-1+\delta)/T)l})\right|^2\asymp\frac{1}{T}\left(\sum_{l=0}^{s-1}l^{d_1-1}\right)^2\asymp s^{2d_1}T^{-1},
		\end{align*}
		where we use the fact that $\Re(e^{2\pi{\rm i}(\mathtt{r}_l-(t_{1*}-1+\delta)/T)l})\geq1/\sqrt{2}$ for $1\leq l\leq T/8$. It implies that 
		\begin{align*}
			\theta_1^2\gtrsim\Vert\mcT(\mfr_1)^{-d_1}\mfu_{k_{1*}}\Vert_2^2\geq\frac{1}{T}\sum_{s=1}^{\lceil T/8\rceil}\left|\sum_{l=0}^{s-1}b_le^{-2\pi{\rm i}(t_{1*}-1+\delta)(l-1)/T}\right|^2\asymp T^{2d_1}.
		\end{align*}
		On the other hand, by the term 1 in Lemma \ref{Lem of singular structure}, we have $\theta_1=\max_{1\leq t\leq T}\Vert\mbU\mfu_t\Vert_2\leq\Vert\mbU\Vert\asymp T^{d_1}$, so $\theta_1\asymp T^{d_1}$.
		\item For each $t=1,\cdots,T$, let's further denote 
		\begin{align}
			\mfg_{l,t}:=|1-e^{2\pi{\rm i}(\mathtt{p}_l-(t-1+\delta)/T)}|\quad\text{and}\quad l_{t*}:=\arg\min\{\mfg_{l,t}:l=1,\cdots,m\},\label{Eq of mfg}
		\end{align}
		According to \eqref{Eq of k-th sup diagonal of mbU}, for $k\geq1$, we have
		\begin{align*}
			b_k=\sum_{l_1+\cdots+l_{\fd}=k}\binom{k}{l_1,\cdots,l_{\fd}}\prod_{r=1}^{\fd}\binom{l_r+d_r-1}{d_r-1}\mfr_r^{l_r},
		\end{align*}
		then the $s$-th entry of $\mbU\mfu_t$ is
		\begin{align*}
			&\frac{e^{-2\pi{\rm i}(t-1+\delta)(s-1)/T}}{\sqrt{T}}\sum_{k=0}^{T-s}b_ke^{-2\pi{\rm i}(t-1+\delta)k/T}\\
			&=\frac{e^{-2\pi{\rm i}(t-1+\delta)(s-1)/T}}{\sqrt{T}}\sum_{k=0}^{T-s}\sum_{l_1+\cdots+l_{\fd}=k}\binom{k}{l_1,\cdots,l_{\fd}}\prod_{r=1}^{\fd}\binom{l_r+d_r-1}{d_r-1}\{\mfr_re^{-2\pi{\rm i}(t-1+\delta)/T}\}^{l_r},
		\end{align*}
		which implies that
		\begin{align*}
			\sqrt{T}\mbU\mfu_t=\diag(1,e^{-2\pi{\rm i}(t-1+\delta)/T},\cdots,e^{-2\pi{\rm i}(t-1+\delta)(T-1)/T})\prod_{l=1}^{\fd}\mcT(\mfr_le^{-2\pi{\rm i}(t-1+\delta)/T})^{-d_l}\boldsymbol{1}_{T\times1}.
		\end{align*}
		For simplicity, denote
		\begin{align}
			\vec{c}_t=(c_{1,t},\cdots,c_{T,t})':=\sqrt{T}\diag(1,e^{2\pi{\rm i}(t-1+\delta)/T},\cdots,e^{2\pi{\rm i}(t-1+\delta)(T-1)/T})\mbU\mfu_t,\notag
		\end{align}
		then we have
		\begin{align}
			\prod_{l=1}^{\fd}\mcT(\mfr_le^{-2\pi{\rm i}(t-1+\delta)/T})^{d_l}\vec{c}_t=\boldsymbol{1}_{T\times1},\label{Eq of vec c}
		\end{align}
		it suffices to show that $\Vert\vec{c}_t\Vert_{\infty}/\Vert\vec{c}_t\Vert_2\leq\mrO(1)$ for $t=1,\cdots,T$.
		
		First, to derive the upper bound of $\Vert\vec{c}_t\Vert_{\infty}$, \eqref{Eq of vec c} implies that
		\begin{align*}
			c_{s,t}+\sum_{r=1}^d\tilde{a}_{r,t}c_{s+r,t}=1,
		\end{align*}
		where $c_{T+1,t}=\cdots=c_{T+d,t}=0$ and $\tilde{a}_{r,t}$ are determined by
		\begin{align*}
			x^d+\sum_{r=1}^d\tilde{a}_{r,t}x^{d-r}=\prod_{l=1}^{\fd}(x-\mfr_le^{-2\pi{\rm i}(t-1+\delta)/T})^{d_l}.
		\end{align*}
		Since we choose the parameter $\delta\in[0,1)$ such that $\mfr_le^{-2\pi{\rm i}(t-1+\delta)/T}\neq1$ for all $t=1,\cdots,T$ and $l=1,\cdots,m$, then $1+\sum_{r=1}^d\tilde{a}_{r,t}\neq1$. Therefore, let $\tilde{c}_{s,t}=c_{T+1-s,t}-(1+\sum_{r=1}^d\tilde{a}_{r,t})^{-1}$, we have
		\begin{align*}
			\tilde{c}_{s,t}+\sum_{r=1}^d\tilde{a}_{r,t}\tilde{c}_{s-r,t}=0,
		\end{align*}
		with initials $\tilde{c}_{0,t}=\cdots=\tilde{c}_{1-d,t}=-(1+\sum_{r=1}^d\tilde{a}_{r,t})^{-1}$, and the explicit formula for the sequence $\{\tilde{c}_{s,t}\}$ is given as
		\begin{align*}
			\tilde{c}_{s,t}=\sum_{l=1}^{\fd}e^{2\pi{\rm i}s(\mathtt{p}_l-(t-1)/T)}\sum_{r=0}^{d_l-1}P_{r,l}^{(t)}s^r,
		\end{align*}
		where $P_{r,l}^{(t)}\in\mbC$ are determined by initials $\tilde{c}_{0,t}=\cdots=\tilde{c}_{1-d,t}=-(1+\sum_{r=1}^d\tilde{a}_{r,t})^{-1}$. Consequently, we obtain that for $1\leq s\leq T$
		\begin{align}
			c_{T+1-s,t}=\frac{1}{\prod_{l=1}^{\fd}(1-e^{2\pi{\rm i}(\mathtt{p}_l-(t-1)/T)})^{d_l}}+\sum_{l=1}^{\fd}e^{2\pi{\rm i}s(\mathtt{p}_l-(t-1)/T)}\sum_{r=0}^{d_l-1}P_{r,l}^{(t)}s^r,\label{Eq of ct formula}
		\end{align}
		which implies that
		\begin{align*}
			\Vert\vec{c}_t\Vert_{\infty}\lesssim\prod_{l=1}^{\fd}\mfg_{l,t}^{-d_l}+T^{d_1-1},\quad t=1,\cdots,T,
		\end{align*}
		where $\mfg_{l,t}$ is defined in \eqref{Eq of mfg}. Since all roots $\mfr_l$ of $f(\cdot)$ in are fixed and distinct, for each $t=1,\cdots,T$, there exists at most one $\mfg_{l_{t*},t}^{-1}\lesssim T$ while $\mfg_{l,t}\gtrsim1$ for all other $l\neq l_{t*}$. 
		
		If $d_{l_{t*}}<d_1$, then $\mfg_{1,t}\gtrsim1$ and $\prod_{l=1}^{\fd}\mfg_{l,t}^{-d_l}+T^{d_1-1}\lesssim T^{d_1-1}$. By the same argument as those in the term 1 of Lemma \ref{Lem of singular structure}, we have $\Vert\vec{c}_t\Vert_2\gtrsim\sqrt{T}\Vert\mcT(\mfr_1)^{-d_1}\mfu_t\Vert_2$. By direct calculations, when $s\in\{1,\cdots,T\}$ is sufficiently large, the norm of the $s$-th entry of $\sqrt{T}\mcT(\mfr_1)^{-d_1}\mfu_t$ satisfies that
		\begin{align}
			&\Bigg|\sum_{r=0}^{s-1}\binom{r+d_1-1}{d_1-1}\mfr_1^re^{-2\pi{\rm i}(t-1+\delta)r/T}\Bigg|\label{Eq of vec c infinite norm}\\
			&=\Bigg|\frac{(d_1-1)!}{(1-\mfr_1e^{-2\pi{\rm i}(t-1+\delta)/T})^{d_1}}-\sum_{r=0}^{d_1-1}\frac{(d_1-1)!(s+d_1-1)!}{r!(s+d_1-1-r)!}\frac{\{\mfr_1e^{-2\pi{\rm i}(t-1+\delta)/T}\}^{s+d_1-1-r}}{(1-\mfr_1e^{-2\pi{\rm i}(t-1+\delta)/T})^{d_1-r}}\Bigg|\notag\\
			&\asymp s^{d_1-1},\notag
		\end{align}
		where we use $\mfg_{1,t}=|1-\mfr_1e^{-2\pi{\rm i}(t-1+\delta)/T}|\gtrsim1$. Thus, it derives that $\sqrt{T}\Vert\mcT(\mfr_1)^{-d_1}\mfu_t\Vert_2\gtrsim T^{d_1-1/2}$ and
		\begin{align}
			\frac{\Vert\vec{c}_t\Vert_{\infty}}{\Vert\vec{c}_t\Vert_2}\lesssim\frac{\prod_{l=1}^{\fd}\mfg_{l,t}^{-d_l}+T^{d_1-1}}{\sqrt{T}\Vert\mcT(\mfr_1)^{-d_1}\mfu_t\Vert_2}\lesssim\frac{T^{d_1-1}}{T^{d_1-1/2}}=T^{-1/2}.\notag
		\end{align}
		If $d_{l_{t*}}=d_1$, then $\prod_{l=1}^{\fd}\mfg_{l,t}^{-d_l}\lesssim\mfg_{l_{t*},t}^{-d_1}$. Since $\Vert\vec{c}_t\Vert_2\gtrsim\sqrt{T}\Vert\mcT(\mfr_{l_{t*}})^{-d_{l_{t*}}}\mfu_t\Vert_2$, by the same arguments as those for \eqref{Eq of vec c infinite norm}, the norm of the $s$-th entry of $\sqrt{T}\mcT(\mfr_1)^{-d_1}\mfu_t$ satisfies that
		\begin{align*}
			&\Bigg|\sum_{r=0}^{s-1}\binom{r+d_{l_{t*}}-1}{d_{l_{t*}}-1}\mfr_{l_{t*}}^re^{-2\pi{\rm i}(t-1+\delta)r/T}\Bigg|\asymp\max\{s^{d_{l_{t*}}-1},|\mfg_{l_{t*},t}|^{-d_{l_{t*}}}\}.
		\end{align*}
		For those $t\in\{1,\cdots,T\}$ such that $|\mfg_{l_{t*},t}|^{-d_{l_{t*}}}\geq CT^{d_1-1}$, where $C>0$ is a sufficiently large constant, we have $\sqrt{T}\Vert\mcT(\mfr_{l_{t*}})^{-d_{l_{t*}}}\mfu_t\Vert_2\gtrsim T^{1/2}|\mfg_{l_{t*},t}|^{-d_{l_{t*}}}$ and
		\begin{align*}
			\frac{\Vert\vec{c}_t\Vert_{\infty}}{\Vert\vec{c}_t\Vert_2}\lesssim\frac{\prod_{l=1}^{\fd}\mfg_{l,t}^{-d_l}+T^{d_1-1}}{\sqrt{T}\Vert\mcT(\mfr_{l_{t*}})^{-d_{l_{t*}}}\mfu_t\Vert_2}\lesssim\frac{|\mfg_{l_{t*},t}|^{-d_{l_{t*}}}}{T^{1/2}|\mfg_{l_{t*},t}|^{-d_{l_{t*}}}}\lesssim T^{-1/2}.
		\end{align*}
		Otherwise, we have $\sqrt{T}\Vert\mcT(\mfr_{l_{t*}})^{-d_{l_{t*}}}\mfu_t\Vert_2\gtrsim\big(\sum_{s=1}^Ts^{2(d_{l_{t*}}-1)}\big)^{1/2}\asymp T^{d_1-1/2}$ and
		\begin{align*}
			\frac{\Vert\vec{c}_t\Vert_{\infty}}{\Vert\vec{c}_t\Vert_2}\lesssim\frac{\prod_{l=1}^{\fd}\mfg_{l,t}^{-d_l}+T^{d_1-1}}{\sqrt{T}\Vert\mcT(\mfr_{l_{t*}})^{-d_{l_{t*}}}\mfu_t\Vert_2}\lesssim\frac{T^{d_1-1}}{T^{d_1-1/2}}\lesssim T^{-1/2}.
		\end{align*}
		\item Recall $\vec{c}_t$ defined in \eqref{Eq of vec c}, we have $\theta_t=\Vert\mbU\mfu_t\Vert_2=\Vert\vec{c}_t\Vert_2/\sqrt{T}$. Since $\theta_1\asymp T^{d_1}$, it gives that $\frac{\theta_t}{\theta_1}\lesssim\frac{\Vert\vec{c}_t\Vert_2}{T^{d_1+1/2}}$. According to \eqref{Eq of ct formula}, we have
		\begin{align*}
			\frac{c_{T+1-s,t}}{T^{d_1+1/2}}\lesssim\frac{1}{T^{d_1+1/2}\prod_{l=1}^{\fd}(1-e^{2\pi{\rm i}(\mathtt{p}_l-(t-1)/T)})^{d_l}}+\frac{s^{d_1-1}}{T^{d_1+1/2}}.
		\end{align*}
		For a small $\epsilon>0$, when $|k-\mathtt{p}_l|<\epsilon T$, let $k_l:=\arg\min\{|k-\mathtt{p}_l|:k=1,\cdots,T\}$, 
		$$\frac{c_{T+1-s,k}}{T^{d_1+1/2}}\lesssim T^{d_l-d_1-1/2}|k-k_l|^{-d_l}+\frac{s^{d_1-1}}{T^{d_1+1/2}},$$
		which implies that
		\begin{align*}
			\frac{\Vert\vec{c}_t\Vert_2}{T^{d_1+1/2}}\lesssim T^{d_l-d_1}|k-k_l|^{-d_l}+T^{-1}.
		\end{align*}
		Otherwise, when $|k-\mathtt{p}_l|\geq\epsilon T$ for all $l=1,\cdots,m$, by the previous argument, we have
		\begin{align*}
			\frac{\Vert\vec{c}_t\Vert_2}{T^{d_1+1/2}}\lesssim T^{-d_1}\epsilon^{-d_l}+T^{-1},
		\end{align*}
		which implies that $\theta_t/\theta_1\lesssim\max\{t^{-d_1},T^{-1}\}$ for $1\leq t\leq T$.
	\end{enumerate}
\end{proof}

\begin{proof}[Proof of Lemma \ref{Lem of one nonstationary root}]
	By the SVD of $\mbU$ in \eqref{Eq of beta k} and $\mbT=\mbU^{-1}$, we have $\mbT=\sum_{t=1}^T\sigma_t^{-1}\mfw_t\mfv_t'$ for $\sigma_1\geq\cdots\geq\sigma_T$ and
	\begin{align*}
		\sigma_k^{-1}\mfc_{k,t}=\langle\mfu_{\pi(t)},\sigma_k^{-1}\mfw_k\rangle=\langle\mfu_{\pi(t)},\mbT\mfv_k\rangle=\mfu_{\pi(t)}'\mbT\mfv_k,
	\end{align*}
	it implies that $|\mfc_{k,t}|\leq\sigma_k\Vert\mbT'\mfu_{\pi(t)}\Vert_2$. For $t< T^{\epsilon}$, by proofs of Lemma \ref{Lem of singular structure 2}, there exists an $l\in\{1,\cdots,m\}$ such that $|\pi(t)-k_l|\lesssim T^{\epsilon}$, where $k_l=\arg\min\{|k-\mathtt{p}_l|:k=1,\cdots,T\}$ for $\mathtt{p}_l$ defined in \eqref{Eq of nonstationary roots}. By direct calculations, let $\mbT_{\cdot r}$ be the $r$-th column of $\mbT$, we have
	$$\big|\langle\mbT_{\cdot r},\mfu_{\pi(t)}\rangle\big|\left\{\begin{array}{ll}
		=T^{-1/2}\big|f_X(e^{-2\pi{\rm i}(\pi(t)-1+\delta)})\big|, & d+1\leq r\leq T, \\
		\lesssim T^{-1/2}, & 1\leq r\leq d,
	\end{array}\right.$$
	where $\big|f_X(e^{-2\pi{\rm i}(\pi(t)-1+\delta)})\big|\lesssim t/T$, so $\Vert\mbT'\mfu_{\pi(t)}\Vert_2\lesssim T^{-1/2}+tT^{-1/2}$ for $t\leq T^{\epsilon}$. Thus, for $k>T^{1/2+2\epsilon}$, Lemma \ref{Lem of singular structure} gives that $\sigma_k/T\lesssim k^{-1}$ when $d_1=1$, then
	\begin{align*}
		|\mfc_{k,t}|\leq\sigma_k\Vert\mbT'\mfu_{\pi(t)}\Vert_2\lesssim\frac{t\sqrt{T}}{k}\lesssim T^{-\epsilon},
	\end{align*}
	which completes the proof.
\end{proof}

\section{Proof of Auxiliary lemmas in \texorpdfstring{\S\ref{sec of Covariance analysis}}{Section 2.2.4}}
\begin{proof}[Proof of Lemma \ref{Lem of positive covariance}]
	For simplicity, we denote $\tau:=\arcsin\rho$ for \(\tau\in[-\pi/2,\pi/2]\), then
	$$h_t|q_t\sim\mcN\left(q_t\sin\tau,\cos^2\tau\right),\quad t=1,\cdots,T,$$
	which further implies
	\begin{align*}
		&\mbE\left[\frac{h_kh_l}{\sum_{t=1}^T\beta_t^2h_t^2}\frac{q_kq_l}{\sum_{t=1}^T\beta_t^2q_t^2}\right]=\mbE\left[\frac{q_kq_l}{\sum_{t=1}^T\beta_t^2q_t^2}\mbE\left[\frac{h_kh_l}{\sum_{t=1}^T\beta_t^2h_t^2}\Bigg|(q_1,\cdots,q_T)\right]\right]\\
		&=\mbE\left[\frac{(q_k\sin\tau+w_k\cos\tau)(q_l\sin\tau+w_l\cos\tau)}{\sum_{t=1}^T\beta_t^2(q_t\sin\tau+w_t\cos\tau)^2}\frac{q_kq_l}{\sum_{t=1}^T\beta_t^2q_t^2}\right],
	\end{align*}
	where \(\bbw=(w_1,\cdots,w_T)\sim\mcN(\boldsymbol{0},\bbI_T)\) is independent with \((q_1,\cdots,q_T)\). Next, define
	$$H_{k,l}(\tau):=\mbE\left[\frac{(q_k\sin\tau+w_k\cos\tau)(q_l\sin\tau+w_l\cos\tau)}{\sum_{t=1}^T\beta_t^2(q_t\sin\tau+w_t\cos\tau)^2}\frac{q_kq_l}{\sum_{t=1}^T\beta_t^2q_t^2}\right],$$
	then
	\begin{align*}
		&\Cov\left(\frac{h_kh_l}{\sum_{t=1}^T\beta_t^2h_t^2},\frac{q_kq_l}{\sum_{t=1}^T\beta_t^2q_t^2}\right)=H_{k,l}(\tau)-\mbE\left[\frac{h_kh_l}{\sum_{t=1}^T\beta_t^2h_t^2}\right]\mbE\left[\frac{q_kq_l}{\sum_{t=1}^T\beta_t^2q_t^2}\right]
	\end{align*}
	is a smooth function of \(\tau\) and \(\frac{d}{d\tau}\Cov\Big(\frac{h_kh_l}{\sum_{t=1}^T\beta_t^2h_t^2},\frac{q_kq_l}{\sum_{t=1}^T\beta_t^2q_t^2}\Big)=\frac{d}{d\tau}H_{k,l}(\tau)\). Since
	\begin{align*}
		&\frac{d}{d\tau}H_{k,l}(\tau)=\mbE\left[\frac{q_kq_l(h_k(q_l\cos\tau-w_l\sin\tau)+h_l(q_k\cos\tau-w_k\sin\tau))}{\big(\sum_{t=1}^T\beta_t^2q_t^2\big)\big(\sum_{t=1}^T\beta_t^2h_t^2\big)}\right]\\
		&+2\sum_{s=1}^T\beta_s^2\mbE\left[\frac{q_kq_lh_kh_lh_s(-q_s\cos\tau+w_s\sin\tau)}{\big(\sum_{t=1}^T\beta_t^2q_t^2\big)\big(\sum_{t=1}^T\beta_t^2h_t^2\big)^2}\right],
	\end{align*}
	and
	\begin{small}
		\begin{align*}
			&\frac{d}{d\tau}H_{k,l}(0)=\mbE\left[\frac{q_kq_l^2w_k+q_k^2q_lw_l}{\big(\sum_{t=1}^T\beta_t^2q_t^2\big)\big(\sum_{t=1}^T\beta_t^2w_t^2\big)}\right]-2\sum_{s=2}^T\beta_s^2\mbE\left[\frac{q_kq_lq_sw_kw_lw_s}{\big(\sum_{t=1}^T\beta_t^2q_t^2\big)\big(\sum_{t=1}^T\beta_t^2w_t^2\big)^2}\right],
		\end{align*}
	\end{small}\noindent
	where
	\begin{align*}
		&\mbE\left[\frac{q_kq_l^2w_k}{\big(\sum_{t=1}^T\beta_t^2q_t^2\big)\big(\sum_{t=1}^T\beta_t^2w_t^2\big)}\right]=\mbE\left[\frac{q_kq_l^2}{\sum_{t=1}^T\beta_t^2q_t^2}\right]\mbE\left[\frac{w_k}{\sum_{t=1}^T\beta_t^2w_t^2}\right]=0,
	\end{align*}
	and
	\begin{align*}
		&\mbE\left[\frac{q_kq_lq_sw_kw_lw_s}{\big(\sum_{t=1}^T\beta_t^2q_t^2\big)\big(\sum_{t=1}^T\beta_t^2w_t^2\big)^2}\right]=\mbE\left[\frac{q_kq_lq_s}{\sum_{t=1}^T\beta_t^2q_t^2}\right]\mbE\left[\frac{w_kw_lw_s}{\big(\sum_{t=1}^T\beta_t^2w_t^2\big)^2}\right]=0.
	\end{align*}
	Hence, \(\frac{d}{d\tau}H_{k,l}(0)=0\). Next, note that
	\begin{align*}
		&\frac{d^2}{d\tau^2}H_{k,l}(\tau)=\mbE\left[\frac{2q_kq_l(h_k'h_l'-h_kh_l)}{\big(\sum_{t=1}^T\beta_t^2q_t^2\big)\big(\sum_{t=1}^T\beta_t^2h_t^2\big)}\right]-\sum_{s=1}^T\beta_s^2\mbE\left[\frac{2q_kq_l(h_kh_l'+h_k'h_l)h_sh_s'}{\big(\sum_{t=1}^T\beta_t^2q_t^2\big)\big(\sum_{t=1}^T\beta_t^2h_t^2\big)^2}\right]\\
		&-\sum_{s=1}^T\beta_s^2\mbE\left[\frac{2q_kq_l\big(h_kh_lh_sh_s'\big)'}{\big(\sum_{t=1}^T\beta_t^2q_t^2\big)\big(\sum_{t=1}^T\beta_t^2h_t^2\big)^2}\right]+\sum_{s,r=1}^T\beta_s^2\tsigma_r^2\mbE\left[\frac{8q_kq_l\big(h_kh_lh_sh_s'h_rh_r'\big)}{\big(\sum_{t=1}^T\beta_t^2q_t^2\big)\big(\sum_{t=1}^T\beta_t^2h_t^2\big)^3}\right],
	\end{align*}
	where \(h_t'=\frac{d}{d\tau}h_t=q_t\cos\tau-w_t\sin\tau\sim\mcN(0,1)\) and \((h_t')'=-h_t\). For the first term in the above equation, by Lemma \ref{Thm of cdf order}, we have
	\begin{align}
		&\left|\mbE\left[\frac{q_kq_l(h_k'h_l'-h_kh_l)}{\big(\sum_{t=1}^T\beta_t^2q_t^2\big)\big(\sum_{t=1}^T\beta_t^2h_t^2\big)}\right]\right|\leq\mbE\left[\left(\sum_{t=1}^T\beta_t^2h_t^2\right)^{-4}\right]^{1/4}\mbE\left[\left(\sum_{t=1}^T\beta_t^2q_t^2\right)^{-4}\right]^{1/4}\notag\\
		&\times\mbE[q_k^4q_l^4]^{1/4}\mbE[(h_k'h_l'-h_kh_l)^4]^{1/4}\leq C,\notag
	\end{align}
	where we use the fact that \(\sum_{t=1}^T\beta_t^2h_t^2\geq\tsigma_{10}^2\sum_{t=1}^{10}h_t^2\sim\tsigma_{10}^2\chi^2(10)\) and the inverse chi square distribution with degree of freedom greater than \(10\) has the finite 4th moment; and all \(q_k,q_l,h_k,h_l,h_k',h_l'\) are standard normal. It is easy to see this constant \(C\) is independent of \(\tau\). Similarly, we can also show that all other three terms are bounded by some constants independent of \(\tau\), so we omit details here. Combining with the fact that $\sum_{t=1}^T\beta_1^2\lesssim1$ in Lemma \ref{Lem of singular structure}, we show that
	$$\left|\frac{d^2}{d\tau^2}H_{k,l}(\tau)\right|<C,\quad{\rm for\ }\tau\in[-\pi/2,\pi/2],$$
	then
	$$\left|\frac{d}{d\tau}H_{k,l}(\tau)\right|\leq\frac{d}{d\tau}H_{k,l}(0)+\left|\int_0^{\tau}\frac{d^2}{d\tau^2}H_{k,l}(\tau)d\tau\right|\leq\int_0^{\tau}\left|\frac{d^2}{d\tau^2}H_{k,l}(\tau)\right|d\tau\lesssim|\tau|,$$
	and
	$$\left|\Cov\left(\frac{h_kh_l}{\sum_{t=1}^T\beta_t^2h_t^2},\frac{q_kq_l}{\sum_{t=1}^T\beta_t^2q_t^2}\right)\right|=\left|\int_0^{\tau}\frac{d}{d\tau}H_{k,l}(\tau)d\tau\right|\leq\int_0^{\tau}\left|\frac{d}{d\tau}H_{k,l}(\tau)\right|d\tau\lesssim\rho^2,$$
	where we use \(\rho=\sin\tau\geq\frac{2}{\pi}\tau\) in the last inequality. Now, we complete the proof.
\end{proof}

\section{Proof of Lemma \texorpdfstring{\ref{Lem of spectral gap SN}}{3.6}}\label{App of spectral gap SN}
\begin{proof}[Proof of Lemma \ref{Lem of spectral gap SN}]
	\begin{enumerate}
		\item By direct calculations, $\mcT(\rho)^{-d}$ is a $T\times T$ upper Toeplitz matrix such that its $(t_1,t_2)$-th ($1\leq t_1\leq t_2\leq T$) entry is
		\begin{align*}
			\binom{t_2-t_1+d-1}{d-1}\rho^{t_2-t_1},
		\end{align*}
		which implies that
		\begin{align*}
			\|\mcT(\rho)^{-d}\|\geq\|\mcT(\rho)^{-d}\|_{\infty}\asymp T^{d-1}|\rho|^T.
		\end{align*}
		On the other hand, by $\|\mcT(\rho)^{-d}\|^2\leq\|\mcT(\rho)^{-d}\|_1\|\mcT(\rho)^{-d}\|_{\infty}$, since
		\begin{align*}
			\|\mcT(\rho)^{-d}\|_1\leq\sum_{t=0}^{T-1}\binom{t+d-1}{d-1}|\rho|^t\lesssim T^{d-1}|\rho|^T,
		\end{align*}
		it completes the proof of the first conclusion. 
		\item For simplicity, we define an auxiliary matrix $\bbC\in\mbR^{T\times T}$ as below:
		\begin{align*}
			\bbC=[C_{t_1,t_2}]_{T\times T},\quad C_{t_1,t_2}=\binom{t_2-t_1+d-1}{d-1}\rho^{t_2-t_1}.
		\end{align*}
		By direct calculations, we have $\|\bbC-\mcT(\rho)^{-d}\|\lesssim1$ and
		\begin{align*}
			C_{t_1,t_2}=\sum_{l=0}^{d-1}u_{t_1,l}(\rho)\overline{v_{l,t_2}(\rho)},
		\end{align*}
		where $u_{t_1,l}(\rho),v_{l,t_2}(\rho)$ are defined in \eqref{Eq of rank-one approximation} and we use the following identity:
		\begin{align*}
			\binom{t_2+(d-1-t_1)}{d-1} = \sum_{l=0}^{d-1} \binom{t_2}{l} \binom{d-1-t_1}{d-1-l}.
		\end{align*}
		Hence, by \eqref{Eq of rank-one approximation}, we obtain that
		\begin{align*}
			\bbC=\bbu_{d-1}(\rho)\bbv_{d-1}(\rho)^*+\bbU(\rho)\bbV(\rho)^*,
		\end{align*}
		where $\bbU=\bbU(\rho)=[U_{t,l}]_{T\times (d-1)}$ such that $U_{t,l}=u_{t,l}(\rho)$, so does $\bbV=\bbV(\rho)=[V_{t,l}]_{T\times (d-1)}$.

		First, for any $1\leq l_1,l_2\leq d-1$, since
		\begin{align*}
			\lim_{T\to\infty}\sum_{t=1}^TU_{t,l_1}\bar{U}_{t,l_2}=\lim_{T\to\infty}\sum_{t=1}^T|\rho|^{-2t}\binom{t-l_1-1}{d-1-l_1}\binom{t-l_2-1}{d-1-l_2}
		\end{align*}
		exists, then $\bbU^*\bbU$ converges to a deterministic $(d-1)\times(d-1)$ matrix as $T\to\infty$, and its spectral norm is bounded. 
		
		For $\bbV^*\bbV$, by the Gram-Schmidt orthogonalization, there exists a unique sequence of orthonormal polynomials $q_k(\cdot)$ of degree $k$ such that $\sum_{i=0}^{\infty}|\rho|^{-2i}q_{k_1}(i)q_{k_2}(i) = \delta_{k_1,k_2}$. Since $V_{t,l}=\bar{\rho}^t\binom{t}{l}$, where $\binom{t}{l}$ is a polynomial of degree exactly $l$ in the variable $t$, and $\{q_0(\cdot), q_1(\cdot), \cdots, q_l(\cdot)\}$ forms a basis for the vector space of all polynomials of degree no more than $l$, then
		\begin{align*}
			V_{t,l}=\bar{\rho}^t\binom{t}{l}=\bar{\rho}^t\binom{T-(T-t)}{l}=\bar{\rho}^t\sum_{k=0}^l P_{k,l}q_k(T-t).
		\end{align*}
		Let $\bbP=[P_{k,l}1_{k\leq l}]_{(d-1)\times(d-1)}$ and $\bbW=[W_{t,k}]_{T\times(d-1)}$, where $W_{t,k}=\bar{\rho}^tq_k(T-t)$, then $\bbV=\bbW\bbP$. Here, we expand
		\begin{align*}
			\binom{T-(T-t)}{l}=\sum_{k=0}^lb_k(T-t)^k,\quad b_k=\frac{(-1)^k}{l!}\sum_{|S|=l-k}\prod_{j\in S}(T-j),
		\end{align*}
		where $S$ is a specific choice of $l-k$ indices from the set $\{0,1,\cdots,l-1\}$, so 
		\begin{align*}
			b_k=\frac{(-1)^k}{l!}\binom{l}{k}T^{l-k}+\mrO(T^{l-k-1})=\frac{(-1)^kT^{l-k}}{k!(l-k)!}+\mrO(T^{l-k-1}),
		\end{align*}
		and
		\begin{align}
			P_{k,l}=\frac{b_k}{c_k}=\frac{(-1)^kT^{l-k}}{k!(l-k)!c_k}+\mrO(T^{l-k-1}),\label{Eq of bbP}
		\end{align}
		where $c_k$ is the coefficient of $x^k$ of $q_k(x)$. Besides, note that 
		\begin{align*}
			\sum_{t=1}^TW_{t,k_1}\bar{W}_{t,k_1}=|\rho|^{2T}\sum_{t=0}^{T-1}|\rho|^{-2t}q_{k_1}(t)q_{k_2}(t)=|\rho|^{2T}(\delta_{k_1,k_2}+\mro(1)),
		\end{align*}
		so $\bbW^*\bbW=|\rho|^{2T}(\bbI_{d-1}+\mro(\boldsymbol{1}_{(d-1)\times(d-1)}))$. Consequently, we have
		\begin{align*}
			\bbV^*\bbV=|\rho|^{2T}(\bbP^*\bbP+\mro(\boldsymbol{1}_{(d-1)\times(d-1)})),
		\end{align*}
		and
		\begin{align*}
			\|\bbV\|^2\lesssim|\rho|^{2T}\sum_{i,j=1}^{d-1}|P_{i,j}|^2\lesssim|\rho|^{2T}T^{2(d-2)},
		\end{align*} 
		where we use \eqref{Eq of bbP}. Combining with $\|\bbU\|\lesssim1$, we conclude that
		\begin{align*}
			\|\bbC-\bbu_{d-1}(\rho)\bbv_{d-1}(\rho)^*\|\leq\|\bbU\|\times\|\bbV\|\lesssim|\rho|^{T}T^{d-2},
		\end{align*}
		which completes the proof of the second conclusion by $\|\bbC-\mcT(\rho)^{-d}\|\lesssim1$.
	\end{enumerate}
\end{proof}


\begin{thebibliography}{10}

\bibitem{bai2010spectral}
Z.~Bai and J.~W. Silverstein.
\newblock {\em Spectral analysis of large dimensional random matrices}, volume~20.
\newblock Springer, 2010.

\bibitem{bai2008large}
Z.~Bai and W.~Zhou.
\newblock Large sample covariance matrices without independence structures in columns.
\newblock {\em Statistica Sinica}, pages 425--442, 2008.

\bibitem{bao2024spectral}
Z.~Bao, J.~Hu, X.~Xu, and X.~Zhang.
\newblock Spectral statistics of sample block correlation matrices.
\newblock {\em The Annals of Statistics}, 52(5):1873--1898, 2024.

\bibitem{bao2012tracy}
Z.~Bao, G.~Pan, and W.~Zhou.
\newblock Tracy-widom law for the extreme eigenvalues of sample correlation matrices.
\newblock {\em Electron. J. Probab}, 17(88):1--32, 2012.

\bibitem{birge1993rates}
L.~Birg{\'e} and P.~Massart.
\newblock Rates of convergence for minimum contrast estimators.
\newblock {\em Probability Theory and Related Fields}, 97(1):113--150, 1993.

\bibitem{bobkov2014bounds}
S.~Bobkov and G.~Chistyakov.
\newblock Bounds on the maximum of the density for sums of independent random variables.
\newblock {\em Journal of Mathematical Sciences}, 199:100--106, 2014.

\bibitem{chatterjee2006generalization}
S.~Chatterjee.
\newblock A generalization of the lindeberg principle.
\newblock {\em The Annals of Probability}, 34(6):2061--2076, 2006.

\bibitem{durrett2019probability}
R.~Durrett.
\newblock {\em Probability: theory and examples (3rd ed.)}, volume~49.
\newblock Cambridge university press, 2019.

\bibitem{erdHos2010bulk}
L.~Erd{\H{o}}s, S.~P{\'e}ch{\'e}, J.~A. Ram{\'\i}rez, B.~Schlein, and H.-T. Yau.
\newblock Bulk universality for wigner matrices.
\newblock {\em Communications on Pure and Applied Mathematics: A Journal Issued by the Courant Institute of Mathematical Sciences}, 63(7):895--925, 2010.

\bibitem{fan2022estimating}
J.~Fan, J.~Guo, and S.~Zheng.
\newblock Estimating number of factors by adjusted eigenvalues thresholding.
\newblock {\em Journal of the American Statistical Association}, 117(538):852--861, 2022.

\bibitem{feller-vol-2}
W.~Feller.
\newblock {\em An introduction to probability theory and its applications. {V}ol. {II}.}
\newblock Second edition. John Wiley \& Sons Inc., New York, 1971.

\bibitem{gantmacher2002oscillation}
F.~R. Gantmacher and M.~G. Krein.
\newblock {\em Oscillation matrices and kernels and small vibrations of mechanical systems: revised edition}.
\newblock American Mathematical Society Rhode Island, 2002.

\bibitem{gao2017high}
J.~Gao, X.~Han, G.~Pan, and Y.~Yang.
\newblock High dimensional correlation matrices: The central limit theorem and its applications.
\newblock {\em Journal of the Royal Statistical Society Series B: Statistical Methodology}, 79(3):677--693, 2017.

\bibitem{gine2002rates}
E.~Gin{\'e} and A.~Guillou.
\newblock Rates of strong uniform consistency for multivariate kernel density estimators.
\newblock In {\em Annales de l'IHP Probabilit{\'e}s et statistiques}, volume~38, pages 907--921, 2002.

\bibitem{gray2006toeplitz}
R.~M. Gray et~al.
\newblock Toeplitz and circulant matrices: A review.
\newblock {\em Foundations and Trends{\textregistered} in Communications and Information Theory}, 2(3):155--239, 2006.

\bibitem{heiny2022limiting}
J.~Heiny and J.~Yao.
\newblock Limiting distributions for eigenvalues of sample correlation matrices from heavy-tailed populations.
\newblock {\em The Annals of Statistics}, 50(6):3249--3280, 2022.

\bibitem{jiang2004limiting}
T.~Jiang.
\newblock The limiting distributions of eigenvalues of sample correlation matrices.
\newblock {\em Sankhy{\=a}: The Indian Journal of Statistics}, pages 35--48, 2004.

\bibitem{jiang2019determinant}
T.~Jiang.
\newblock Determinant of sample correlation matrix with application.
\newblock {\em Annals of Applied Probability}, 29(3):1356--1397, 2019.

\bibitem{jiang2025asymptotic}
T.~Jiang and T.~Pham.
\newblock Asymptotic distributions of largest pearson correlation coefficients under dependent structures.
\newblock {\em The Annals of Statistics}, 53(3):907--928, 2025.

\bibitem{karner2003spectral}
H.~Karner, J.~Schneid, and C.~W. Ueberhuber.
\newblock Spectral decomposition of real circulant matrices.
\newblock {\em Linear Algebra and Its Applications}, 367:301--311, 2003.

\bibitem{mei2023singular}
T.~Mei, C.~Wang, and J.~Yao.
\newblock On singular values of data matrices with general independent columns.
\newblock {\em Annals of Statistics}, pages 624--645, 2023.

\bibitem{pillai2012edge}
N.~S. Pillai and J.~Yin.
\newblock Edge universality of correlation matrices.
\newblock {\em The Annals of Statistics}, pages 1737--1763, 2012.

\bibitem{romano2000more}
J.~P. Romano and M.~Wolf.
\newblock A more general central limit theorem for m-dependent random variables with unbounded m.
\newblock {\em Statistics \& probability letters}, 47(2):115--124, 2000.

\bibitem{rudelson2015small}
M.~Rudelson and R.~Vershynin.
\newblock Small ball probabilities for linear images of high-dimensional distributions.
\newblock {\em International Mathematics Research Notices}, 2015(19):9594--9617, 2015.

\bibitem{shumway2006time}
R.~H. Shumway and D.~S. Stoffer.
\newblock {\em Time series analysis and its applications: with R examples}.
\newblock Springer, 2006.

\bibitem{tao2011random}
T.~Tao and V.~Vu.
\newblock Random matrices: Universality of local eigenvalue statistics.
\newblock {\em Acta Math}, 206:127--204, 2011.

\bibitem{10.1214/11-AOP648}
T.~Tao and V.~Vu.
\newblock {Random covariance matrices: Universality of local statistics of eigenvalues}.
\newblock {\em The Annals of Probability}, 40(3):1285 -- 1315, 2012.

\bibitem{tyrtyshnikov1996unifying}
E.~E. Tyrtyshnikov.
\newblock A unifying approach to some old and new theorems on distribution and clustering.
\newblock {\em Linear algebra and its applications}, 232:1--43, 1996.

\bibitem{yin2021spectral}
Y.~Yin, C.~Li, G.-L. Tian, and S.~Zheng.
\newblock Spectral properties of rescaled sample correlation matrix.
\newblock {\em Statistica Sinica}, pages 2007--2022, 2021.

\bibitem{yin2023central}
Y.~Yin, S.~Zheng, and T.~Zou.
\newblock Central limit theorem of linear spectral statistics of high-dimensional sample correlation matrices.
\newblock {\em Bernoulli}, 29(2):984--1006, 2023.

\bibitem{yin2024inference}
Y.~Yin and W.~Zhou.
\newblock Inference for spiked eigenstructure under generalized covariance and correlation models.
\newblock {\em arXiv preprint arXiv:2408.13848}, 2024.

\bibitem{zheng2019test}
S.~Zheng, G.~Cheng, J.~Guo, and H.~Zhu.
\newblock Test for high dimensional correlation matrices.
\newblock {\em Annals of statistics}, 47(5):2887, 2019.

\end{thebibliography}

\end{document}